\documentclass[12pt]{article}
\usepackage{chngcntr}

\usepackage[english]{babel}
\usepackage{amsmath,amsfonts,amssymb,amsthm,mathtools}
\usepackage{fullpage}
\usepackage[round]{natbib}

\usepackage[para, flushleft]{threeparttable}
\usepackage{pifont}% http://ctan.org/pkg/pifont

\newtheorem{theorem}{Theorem}
\newtheorem{lemma}{Lemma}
\newtheorem{claim}{Claim}
\newtheorem{proposition}{Proposition}
\newtheorem{corollary}{Corollary}

\theoremstyle{definition}

\theoremstyle{remark}
\newtheorem{remark}{Remark}

\theoremstyle{remark}
\newtheorem{example}{Example}

\newtheorem*{continuancex}{Example \continuanceref~(cont.)}
\newenvironment{continuance}[1]
  {\newcommand\continuanceref{\ref{#1}}\continuancex}
  {\endcontinuancex}

\usepackage{geometry} % to change the page dimensions
\usepackage{graphicx} % support the \includegraphics command and options

\usepackage{booktabs} % for much better looking tables
\usepackage{array} % for better arrays (eg matrices) in maths
\usepackage{paralist} % very flexible & customisable lists (eg. enumerate/itemize, etc.)
\usepackage{verbatim} % adds environment for commenting out blocks of text & for better verbatim
\usepackage{subfig} % make it possible to include more than one captioned figure/table in a single float
\usepackage[implicit=false]{hyperref}
\usepackage{fancyhdr} % This should be set AFTER setting up the page geometry
\usepackage{hyperref}
\usepackage{sectsty}
\allsectionsfont{\sffamily\mdseries\upshape} % (See the fntguide.pdf for font help)
\usepackage[nottoc,notlof,notlot]{tocbibind} % Put the bibliography in the ToC
\usepackage[titles,subfigure]{tocloft} % Alter the style of the Table of Contents

\date{}
\title{The Structure of Equilibria in Trading Networks with Frictions}
\author{Jan Christoph Schlegel\footnote{An extended abstract under the previous title "Trading Networks with General Preferences" appeared in the Proceedings of the 20th ACM Conference on Economics and Computation (EC’19). The paper extends and supersedes~\cite{Schlegel2018} which proves similar result in the more restrictive model of job matching with salaries. I gratefully acknowledge financial support by the Swiss National Science Foundation (SNSF) under project 100018-150086. I thank Ravi Jagadeesan, Bettina Klaus, Alex Nichifor, Alex Teytelboym, Ning Yu and Klaus Zauner, seminar participants in Bristol, Lausanne and Oxford, participants of the 2018 Lisbon Game Theory Meetings, the Matching in Practice workshop in Mannheim, the 5th Match-Up workshop, the 20th ACM Conference on Economics and Computation (EC19), the 2019 North American Summer Meeting of the Econometric Society and the International Conference on Game Theory \& the 6th Microeconomics Workshop in Nanjing, for valuable comments.}\\ \small Department of Economics, City, University of London, UK\\{\small\tt jansc@alumni.ethz.ch}}
\begin{document}
\maketitle
\begin{abstract}Several structural results for the set of competitive equilibria in trading networks with frictions are established: The lattice theorem, the rural hospitals theorem, the existence of side-optimal equilibria, and a group-incentive-compatibility result hold with imperfectly transferable utility and in the presence of frictions. While our results are developed in a trading network model, they also imply analogous (and new) results for exchange economies with combinatorial demand and for two-sided matching markets with transfers.
\emph{JEL-classification:}C78, D47, D52, L14
\end{abstract}
\section{Introduction}
The assumption of quasi-linear utility in transfers is pervasive in models of matching markets, exchange economies with indivisible goods, trading networks, and in mechanism design. Quasi-linearity can simplify the analysis considerably since it allows us to exploit the duality between optimal allocations and supporting equilibrium prices. While the assumption of quasi-linearity simplifies the analysis, it is often empirically problematic. Wealth effects are present in marriage and labor markets so that matching models with transferable utility are unrealistic for these applications. Even if wealth effects are absent, transaction frictions such as those induced by taxation, subsidies, or transaction costs, make a transferable utility model inapplicable. This has motivated researchers to explore how results for matching markets with transfers~\cite[]{Demange.1985,Legros2007,Noldeke2018,Galichon2018} and for trading networks~\cite[]{Ravi,Hatfield:2018} can be generalized beyond quasi-linear utility.

In this paper, we contribute to this discussion and study imperfectly transferable utility and frictions\footnote{In the following, we distinguish language-wise between the case of pure imperfectly transferable utility where utility is not necessarily linear in transfers, but a function of the sum of transfers received, and the case of frictions where utility is not necessarily a function of the sum of transfers received, but a function of the entire vector of transfers. If frictions are present, it thus not only matters how much a firm receives in total transfers, but also through which trades it receives the transfers. This can e.g.~be the case if different trades involve different transaction costs. The results in our paper apply to both pure imperfectly transferable utility and frictions.} in the context of trading networks~\cite[]{Hatfield2013}. Trading networks with bilateral contracts model complex supply chains in an industry where firms are engaged in upstream as well as downstream contracts. They generalize two-sided matching markets, in the sense that they replace a bipartite graph of potential relations, by an arbitrary graph where each edge represents a potential trade. We show that important structural results for  trading networks do not depend on the assumption of quasi-linear utility, and establish several results about the set of competitive equilibria under minimal assumptions on utility functions. Our results apply even in the case of imperfectly transferable utility, in the presence of frictions, and if constraints make the execution of certain combinations of trades infeasible.

Our results can be summarized as follow: For a model of trading networks with frictions~\cite[]{Ravi} and under the assumptions of full substitutability~\cite[]{Inder,Ostrovsky2008,Hatfield2013} and the laws of aggregate demand and supply~\cite[]{Ostrovsky2008,Hatfield:2018}, we show that 
\begin{itemize}
	\item[-] the set of competitive equilibria is a sublattice of the price space (first part of Theorem~\ref{Util}),
	\item[-] a generalized "rural hospitals theorem" holds: the difference between the number of signed downstream and the number of signed upstream contracts is the same for each firm in each equilibrium (second part of Theorem~\ref{Util}),
	\item[-] assuming additionally "bounded compensating variations" (\citealp{Ravi}), there is an equilibrium that is most preferred by terminal sellers and an equilibrium that is most preferred by terminal buyers (Theorem~\ref{extreme}),
	\item[-]  a mechanism that selects buyer-optimal equilibria is group-strategy-proof for terminal buyers on the domain of unit-demand utility functions and similarly a mechanism that selects seller-optimal equilibria is group-strategy-proof for terminal sellers on the domain of unit-supply utility functions~(Theorem~\ref{GSP}).
	\end{itemize}

While our results are established for trading networks, the results are already new for many-to-one matching markets and for exchange economies with combinatorial demand which are special cases of our model. For matching markets, similar results were so far only known (a) under quasi-linear utility, (b) for models without transfers and with strict preferences, or (c) for one-to-one matching markets with imperfectly transferable utility. For exchange economies with indivisible goods, analogous results were so far only known for (a) quasi-linear utility, or for the case of (b) unit demand.

Working with imperfectly-transferable utility and frictions requires us to develop fundamentally new techniques: Similar results for the case of imperfectly transferable utility were so far only known for one-to-one matching markets~\cite[]{Demange.1985} and the proof techniques developed in this context (in particular the "decomposition lemma") do not adapt to more general settings. On the other hand, techniques from transferable utility models do not generalize to our model: \cite{Hatfield2013} use the efficiency of competitive equilibria and the submodularity of the indirect utility function to establish the lattice property.  For our model with frictions, competitive equilibria can fail to be efficient. Moreover, full substitutability implies only the weaker notion of quasi-submodularity~\cite[]{Ravi}. More subtly, as we will discuss below, without quasi-linearity there are several non-equivalent definitions of full substitutability and these definitions are not distinguishable through conditions on the indirect utility function alone. Thus, an approach as in~\cite{Hatfield2013} that characterizes competitive equilibria through the indirect utility function and uses properties of that function under full substitutability does not generalize. Likewise, the network flow approach to trading networks~\cite[]{Candogan2016} obtains structural results on the set of equilibria through the duality between optimal allocations and supporting prices. Since efficiency fails in our setting this duality approach does not generalize. Finally, techniques from trading networks without transfers~\cite[]{Ostrovsky2008} that rely on Tarski's fixpoint theorem do not apply to the our model. With transfers, the issue of tie-breaking arises that is not present in models without transfers and with strict preferences. While versions of Tarski's fixpoint theorem for correspondences are known~\cite[]{Zhou1994}, none of them work under sufficiently weak assumptions to be useful in our environment.

Since existing techniques do not work for our setting, we introduce a new approach to establish structural results for the set of competitive equilibria.  The approach can be characterized as a tie-breaking approach: we show that for each finite set of (equilibrium) price vectors and each firm a single-valued selection from the demand correspondence can be made such that the properties of full substitutability and the laws of aggregate demand and supply are satisfied by the selection, and moreover, relevant trades that are demanded in the supporting equilibrium allocations are demanded in the selection. The assumption of the Laws of Aggregate Demand and Supply is necessary for our tie-breaking argument, and our result does not hold under Full Substitutability alone (see Example~\ref{ex22}).

We also make a more technical contributions to the literature on trading networks with imperfectly transferable utility nad frictions and clarify issues related to the definition of Full Substitutability:
For the quasi-linear model, there are various equivalent definitions of Full Substitutability~\cite[]{Hatfield2018}. The equivalence, however, breaks down if we go beyond quasi-linear utility, and, for our results, it matters which of the full substitutability notions is used. More specifically, it matters how full substitutability restricts the demand at price vectors at which the demand is multi-valued.  We introduce weak notions of  Full Substitutability and the Laws of Aggregate Demand and Supply that only restrict the demand at price vectors where the demand is single-valued and stronger versions that also restrict it at prices where the demand is multi-valued. The notions are equivalent under quasi-linear utility, but not in general. %Moreover, the different notions are not distinguishable through restrictions on the indirect utility function; for each utility function satisfying the weak versions of the conditions, we can construct a corresponding utility function satisfying the strong versions of the conditions that induces the same indirect utility (Proposition~1).
 The set of competitive equilibria is a lattice only under the strong versions of Full Substitutability (see Example~\ref{ex}) and the rural hospitals theorem requires the strong versions of the law of aggregate demand and supply (see Example~\ref{ex3}). Our group-strategy-proofness result, however, holds also under the weaker notions (Corollary~\ref{weakGSP}). Thus, the exact definition of Full Substitutability matters in the model with frictions.\footnote{Related issues occur in~\cite{Hatfield:2018} where the stronger Monotone Substitutability property is needed that restricts the choice in circumstances where the choice correspondence is multi-valued. See our discussion in Appendix~\ref{Full}.}

We proceed as follows: In Section~2, we introduce the model and discuss different versions of the Full Substitutability conditions and their relation to each other. We provide a more detailed discussion of different Full Substitutability conditions in Appendix~\ref{Full}. In Section~3, we prove our main results: the lattice structure of the set of competitive equilibria, the generalized rural hospitals theorem, the existence of extremal equilibria, and group-incentive compatibility for terminal buyers.
In Section~4, we apply our main results to two-sided matching markets, and to exchange economies with indivisible goods. 

\subsection{Related Literature}
The literature on trading networks has its origins in the literature on matching markets with transfers. In a seminal paper, \cite{KelsoCrawford1982} show that, under the assumption of gross substitutability, competitive equilibria with personalized prices exist and are equivalent to core allocations in a many-to-one labor market matching model. The construction is by an approximation argument where the existence in the continuum is obtained from the existence of an equilibrium in a discrete markets with smaller and smaller price increments. Different versions of a strategy-proofness result for a many-to-one matching model with continuous transfers were established by \cite{Hatfieldetal2016,Schlegel2018,Jagadeesan2018}. Subsequent to~\cite{KelsoCrawford1982}, the question of existence of equilibria has been studied in the context of exchange economies with indivisibilities. See for example~\cite{GulStacchetti1999} and the recent contribution of~\cite{Baldwin2016}. %Incentives in this context have for example been studied

Trading networks with bilateral contracts and continuous transfers were introduced by~\cite{Hatfield2013}. Under the assumption of quasi-linear utility and full substitutability they establish many results that we generalize to the case of general utility functions. The notion of full substitutability has been studied in detail by~\cite{Hatfield2018} who show the equivalence of various different definitions of full substitutability. The existence result of~\cite{Hatfield2013} is proved via a reduction to the existence result of~\cite{KelsoCrawford1982}. An alternative approach is via a submodular version of a network flow problem~\cite[]{Candogan2016}.

The work of~\cite{Hatfield2013} builds on the work of~\cite{Ostrovsky2008} on trading networks without transfers that generalizes matching models with contracts \cite[]{HatfieldMilgrom2005,Fleiner2003,Roth1984b} beyond two-sided markets. The matching model with contracts in turn originates in the discrete version of the model of \cite{KelsoCrawford1982}.
\cite{HatfieldKominers2012} and \cite{Fleiner2016} provide additional results for the discrete trading networks model, which in many ways are parallel to the results we obtain in the continuous model. Importantly, results for the model without continuous transfers rely on the assumption of strict preferences.

All the above mentioned work for the continuous models make the assumption of quasi-linear utility.\footnote{Note however that the existence proof of~\cite{KelsoCrawford1982} is actually more general and also applies to non-quasi-linear preferences, provided they are continuous, monotonic and unbounded in transfers for each bundle.} There are few papers that deal with non-quasi-linear utility functions, frictions or constraints and that are particularly close to our work: In a classical paper,~\cite{Demange.1985} establish several structural results about the core (or equivalently the set of competitive equilibria) for a one-to-one matching model with continuous transfers. In particular, they show that the core has a lattice structure and an agent that is unmatched in one core allocation receives his reservation utility in each core allocation (the result is often called the rural hospitals theorem in the literature on discrete matching markets). Moreover, they show that the mechanism that selects an extreme point of the bounded lattice is strategy-proof for one side of the market.
Importantly, these results are established without assuming quasi-linearity in transfers. They only require that utility is increasing, continuous in transfers and satisfies a full range assumption. We generalize this work to trading networks and to situations in which utility does not satisfy the full range assumption.

In recent work, \cite{Ravi} study trading networks with imperfectly transferable utility and frictions. Their work is in many regards complementary to our work. In particular, \cite{Ravi} establish the existence of a competitive equilibrium under the assumption of Full Substitutability and mild regularity conditions. Moreover, they study the efficiency of competitive equilibria and provide conditions under which equilibria correspond to allocations satisfying different related cooperative solution concepts. We derive our results for competitive equilibria. However, by the equivalence result of~\cite{Ravi} analogous results also would hold for "trail-stable" allocations. All results of~\cite{Hatfield2013}, except for the maximal domain result (Theorem~7) are generalized to general utility functions with frictions, either in our work or by~\cite{Ravi}. Table~\ref{table:HAproperties} in Appendix~\ref{yolo} summarizes known results for trading networks with frictions.

\cite{Kojimaetal2019} introduce constraints in the job matching model of ~\cite{KelsoCrawford1982} and characterize constraints that leave the gross substitutes condition invariant. The model with constraints is a special case of the model in the current paper so that we obtain as a corollary of our results a version of a lattice and of the rural hospital theorems for their model of job matching under constraints. In a spin-off paper,~\cite{Kojimaetal2020} study comparative statics for their model and also proof versions of the lattice result and the rural hospitals theorem. These results have been obtained independently and contemporaneously with the results in the current paper.\footnote{Weaker versions of these results were obtained prior to that in~\cite{Schlegel2018}.} %The latter two results follow as a special case from our Theorem~1.
\section{Model}
The model is based on~\cite{Hatfield2013}, and the extensions of~\cite{Ravi} and~\cite{Hatfield:2018}.
We consider a finite set of {\bf firms} $F$ and a finite set of {\bf trades} $\Omega$. Each trade $\omega\in \Omega$ is associated with a buyer $b(\omega)\in F$ and a seller $s(\omega)\in F$ with $b(\omega)\neq s(\omega)$.
For a set of trades $\Psi\subseteq \Omega$ and firm $f\in F$ we define the set of {\bf downstream trades} for $f$ by $\Psi_{f\rightarrow}:=\{\omega\in \Psi:s(\omega)=f\}$ and the set of {\bf upstream trades} by $\Psi_{\rightarrow f}:=\{\omega\in \Omega:b(x)=f\}$. Moreover, we let $\Psi_f:=\Psi_{f\rightarrow}\cup \Psi_{\rightarrow f}$.
A firm $f\in F$ such that $\Omega_{f\rightarrow}=\emptyset$ is called a {\bf terminal buyer} and a firm such that $\Omega_{\rightarrow f}=\emptyset$ is called a {\bf terminal buyer}. Note that terminal buyers and/or terminal buyers do not need to exist. A {\bf contract} is a pair $(\omega,p_{\omega})\in\Omega\times\mathbb{R}$, where $p_{\omega}$ is the price attached to the trade $\omega$.

%A {\bf contract} is a pair $x=(\omega,p_{\omega})\in\Omega\times\mathbb{R}$, where $p_{\omega}$ is the price associated with the trade $\omega$. For the contract $x=(\omega,p_{\omega})$, we let $b(x)=b(\omega)$ and $s(x)=s(\omega)$.
%We denote the set of contracts by $X:=\Omega\times\mathbb{R}$.
%We use analog notation for trades.
 An {\bf allocation} is a pair $(\Psi,p)$ consisting of a set of trades $\Psi\subseteq\Omega$ and a price vector $p\in\mathbb{R}^{\Psi}$. We denote the set of allocations by $\mathcal{A}$ and we let $\mathcal{A}_f:=\{(\Psi_f,(p_{\omega})_{\omega\in\Psi_f}):(\Psi,p)\in\mathcal{A}\}$. An {\bf arrangement} is a pair $[\Psi,p]\in2^{\Omega}\times\mathbb{R}^{\Omega}$. In contrast to an allocation the price vector also contains prices for unrealized trades.

Each firm has a utility function $u^f:\mathcal{A}_f\rightarrow\mathbb{R}\cup\{-\infty\}$.
 For notational convenience we extend  $u^f$ to $2^{\Omega}\times\mathbb{R}^{\Omega}$ by defining for $\Psi\subseteq\Omega$ and $p\in\mathbb{R}^{\Omega}$, the utility $u^f (\Psi,p):=u^f(\Psi_f,(p_{\omega})_{\omega\in\Psi_f}).$
 We allow the utility function to take on a value of $-\infty$ to model technical or institutional constraints faced by a firm that make the execution of the combination of trades infeasible. We require that  \begin{itemize}
 	\item if a bundle is infeasible under some prices, then it is infeasible under all prices: if $u^f(\Psi,p)=-\infty$ for $p\in\mathbb{R}^{\Omega_f}$ then $u^f(\Psi,p')=-\infty$ for each $p'\in\mathbb{R}^{\Omega_f}$,
 	\item  at least one bundle of trades is feasible: there is a $\Psi\subseteq \Omega_f$ such that $u^f(\Psi,\cdot)>-\infty$.\footnote{In contrast to~\cite{Ravi} where it is always feasible to refuse all trades, we allow for the possibility that $u^f(\emptyset)=-\infty$. This does not really make our model more general, since we could alternatively apply a continuous and monotonic utility transformation to guarantee that utility for all feasible bundles is positive and assign zero utility to $\emptyset$.} 
 \end{itemize}
Moreover, we make the following assumptions on utility functions:
\begin{itemize}
\item {\bf Continuity}: For $\Psi\subseteq \Omega_f$ with $u^f(\Psi,\cdot)>-\infty$  the function $u^f(\Psi,\cdot)$ is continuous on $\mathbb{R}^{\Psi}$.
\item {\bf Monotonicity}:
For $\Psi\subseteq \Omega_f$ with $u^f(\Psi,\cdot)>-\infty$ and $p,p'\in\mathbb{R}^{\Psi}$ with $p'\neq p$:
\begin{enumerate}
\item If $p'_{\omega}=p_{\omega}$ for $\omega\in{\Psi_{f\rightarrow}}$ and $p_{\omega}\leq p'_{\omega}$ for $\omega\in\Psi_{\rightarrow f}$, then $u^f(\Psi,p)>u^f(\Psi,p')$.
\item If $p'_{\omega}=p_{\omega}$ for $\omega\in{\Psi_{\rightarrow f}}$ and $p_{\omega}\geq p'_{\omega}$ for $\omega\in\Psi_{f\rightarrow }$, then $u^f(\Psi,p)>u^f(\Psi,p')$.
\end{enumerate}
\end{itemize}
Thus, utility is continuous in prices and firms strictly prefer higher sell prices to lower sell prices and lower buy prices to higher buy prices.
\begin{remark}
Note that we allow utility for a set of trades $\Psi$ to be different for prices $p,p'\in\mathbb{R}^{\Psi}$, even if the transfers received are the same for both price vectors, i.e. even if $\sum_{\omega\in\Psi_{f\rightarrow}}p_{\omega}-\sum_{\omega\in\Psi_{\rightarrow f}}p_{\omega}=\sum_{\omega\in\Psi_{f\rightarrow}}p'_{\omega}-\sum_{\omega\in\Psi_{\rightarrow}f}p'_{\omega}.$ This can e.g.~be the case if different trades involve different transaction costs.

If utility only depends on the set of trades and the transfers received, we have a special case of our model: We say that $u^f$ satisfies {\bf no frictions}~\cite[]{Ravi} if there is a function $\tilde{u}^f:2^{\Omega_f}\times\mathbb{R}\rightarrow \mathbb{R}\cup\{-\infty\}$ such that $$u^f(\Psi,p)=\tilde{u}^f(\Psi,\sum_{\omega\in\Psi_{f\rightarrow}}p_{\omega}-\sum_{\omega\in\Psi_{\rightarrow f}}p_{\omega}).$$
{\bf Quasi-linear} utility corresponds to the case without frictions where $\tilde{u}^f$ is linear in transfers.  In this case there is a valuation function $v^f:2^{\Omega_f}\rightarrow\mathbb{R}\cup\{-\infty\}$ such that
$$\tilde{u}^f(\Psi,t)=v^f(\Psi)+t.$$\qed
\end{remark}
%A utility function induces a {\bf choice correspondence} $C^f:2^X\rightrightarrows\mathcal{A}_f$ by:
%$$C^f(Y):=\{Z\subseteq Y:Z\in\mathcal{A},u^f(Z)\geq u^f(Z')\text{ for all }Z'\subseteq Y\text{ with }Z'\in\mathcal{A}\}.$$
%Note that $C^f(Y)$ can be empty, if $Y$ is not compact.
A utility function $u^f$ induces an {\bf indirect utility function} $v^f:\mathbb{R}^{\Omega}\rightarrow\mathbb{R}$ by $$v^f(p):=\max_{\Psi\subseteq\Omega_f}u^f(\Psi,p),$$
and a {\bf demand correspondence} $D^f:\mathbb{R}^{\Omega}\rightrightarrows2^{\Omega}$ by:
$$D^f(p):=\text{argmax}_{\Psi\subseteq \Omega_f}u^f(\Psi,p).$$
Continuity of the utility
function implies (e.g.~by Berge's maximum theorem) that the demand correspondence is upper hemi-continuous.
Monotonicity of the utility
function implies that price
vectors where the demand is single-valued are dense in price space. We will repeatedly use these facts to
generate a single-valued selection from the demand-correspondence by
perturbing the price vector such that it becomes single-valued. The proof is straightforward and hence omitted.
\begin{lemma}\label{Fede}
For a continuous and monotonic utility function $u^f$,
\begin{enumerate}
    \item the induced demand $D^f$ is upper hemi-continuous, i.e.~for each $p\in \mathbb{R}^{\Omega_f}$ there is an $\epsilon>0$ such for any $q\in\mathbb{R}^{\Omega_f}$ with $\|p-q\|<\epsilon$ (where $\|\cdot\|$ denotes the Euclidean norm) we have $D^{f}(q)\subseteq D^f(p)$,
    \item the set of price vectors such that the induced demand is single-valued is dense in $\mathbb{R}^{\Omega_f}$, i.e.~for each $\epsilon>0$ and $p\in \mathbb{R}^{\Omega_f}$ there is a $q\in\mathbb{R}^{\Omega_f}$ with $\|p-q\|<\epsilon$  such that $|D^f(q)|=1$.
\end{enumerate}
\end{lemma}
%\begin{proof}
%	Let $\Psi\in D^f(p)$ be a minimal set of demanded trades at $p$, i.e.~such that $\Xi\notin D^f(p)$ for $\Xi\subsetneq \Psi$. Consider any vector $q$ such that $p_{\omega}=q_{\omega}$ for $\omega\in \Psi$, $p_{\omega}<q_{\omega}$ for $\omega\in\Omega_{\rightarrow f}\setminus \Psi$ and $p_{\omega}>q_{\omega}$ for $\omega\in\Omega_{f\rightarrow }\setminus \Psi$. By construction, we have $u^f(\Xi,p)=u^f(\Xi,q)$ for each $\Xi\subseteq \Psi$ and $u^f(\Xi',p)<u^f(\Xi',q)$ for each $\Xi\nsubseteq \Psi$. Thus, $D^f(q)=\{\Psi\}$. \end{proof}
%For completeness, the lemma is proved in the appendix.% All subsequent proofs that are not contained in the main text are in the appendix.
\subsection{Full Substitutability}\label{FSS}
Our results rely on a full substitutability assumption on utility
functions. Informally, the condition requires that a firm sees
upstream (downstream) trades as substitutes to each other, and
upstream and downstream trades as complements to each
other.~\cite{Hatfield2018} show that for quasi-linear utility
various ways of defining full substitutability are equivalent, and
hence one can work with either of the definitions discussed in their
paper. Going beyond quasi-linear utility makes issues more subtle:
Not all equivalence results of~\cite{Hatfield2018} generalize to
non-quasi-linear utility and it matters which of the full
substitutability conditions are used. More specifically, it matters
how the full substitutes condition is defined in instances where
indifferences matter, i.e.~when the demand is multi-valued.

We will proceed as follows: First, we introduce our main definition of full
substitutability which restrict the demand both at price vectors where the demand is single-valued and where it is multi-valued. Second, we introduce the weak version of full
substitutability that only restricts the demand at price vectors where the demand is single-valued.  We provide an example that shows that the single-valued version of full substitutability is strictly weaker than the multi-valued version. We later show, using this example, that the single-valued full substitutability condition is not sufficient for establishing the lattice and the rural hospitals theorem. Importantly, the difference between the single-valued and multi-valued version of full substitutability only matters for the "cross-side conditions" on firms' demand functions. In particular, the notions are equivalent for a two-sided market and the results for two-sided markets (see Section~\ref{two}) hold under the single-valued notion of full substitutability. Third, we show that the multi-valued and the single-valued versions are, however, closely related in the sense that for each demand correspondence satisfying weak full substitutability, a selection from the demand correspondence exists that satisfies (multi-valued) full substitutability and can be rationalized by a utility function inducing the same indirect utility. In particular, this will allow us, later on (Corollary~\ref{weakGSP}), to obtain a group-strategy-proofness result using only the weak full substitutability notion. 

In Appendix~\ref{Full} we discuss in more detail the logical relation between different notions of full substitutability. In particular, we show (Corollary~\ref{MON}) that under the assumption of continuity and monotonicity, monotone substitutability as defined in~\cite{Hatfield:2018}, is equivalent to the combination of (the multi-valued versions of) full substitutability, the laws of aggregate demand and supply. Thus, the main results in our paper hold under (the demand language version of) the substitutability notion used by~\cite{Hatfield:2018} for their main result.

\subsubsection{Multi-Valued Full Substitutability}
The following notion of full substitutability is due to~\cite{Hatfield2018}.\footnote{\citealp{Hatfield2018} call this version of full substitutability the "demand-language expansion" version of full
	substitutability (cf. Definition~A.3 in~\citealp{Hatfield2018}). There are alternative multi-valued definitions of full substitutability which we discuss in Appendix~\ref{Full}.	
	 Throughout the paper, we use ``demand language" definitions of full substitubility that restrict the demand correspondence induced by the utility function. The demand language definitions are generally weaker than the
	corresponding ``choice language" definitions that restrict the choice correspondence induced by the utility function. Consequently, all of our result would also hold under the corresponding ``choice language" notions of full substitutability.} Precursors of the full substitutability notion were introduced for exchange economies \cite[]{Inder} and for trading networks without transfers~\cite[]{Ostrovsky2008}. Full substitutability can be further decomposed in a same-side substitutability and a cross-side complementarity notion.
  \newline\newline
     \noindent
{\bf (Expansion) Same-Side Substitutability (SSS)}: For
$p,p'\in\mathbb{R}^{\Omega_f}$ and each $\Psi'\in D^f(p')$ there
exists a $\Psi\in D^f(p)$ such that if $p_{\omega}= p'_{\omega}$ for
$\omega\in\Omega_{f\rightarrow}$ and $p_{\omega}\leq p'_{\omega}$
for $\omega\in \Omega_{\rightarrow f} $, then
$$\{\omega\in\Psi_{\rightarrow f}:p_{\omega}=p'_{\omega}\}\subseteq
\Psi'_{\rightarrow f},$$ and if $p_{\omega}= p'_{\omega}$ for
$\omega\in\Omega_{\rightarrow f}$ and $p_{\omega}\geq p'_{\omega}$
for $\omega\in \Omega_{f\rightarrow } $, then
$$\{\omega\in\Psi_{f\rightarrow }:p_{\omega}=p'_{\omega}\}\subseteq \Psi'_{f\rightarrow }.$$

     \noindent {\bf (Expansion) Cross-Side Complementarity (CSC)}: For $p,p'\in\mathbb{R}^{\Omega_f}$ and each $\Psi'\in D^f(p')$ there exists a $\Psi\in D^f(p)$ such if  $p_{\omega}= p'_{\omega}$ for $\omega\in\Omega_{f\rightarrow}$ and $p_{\omega}\leq p'_{\omega}$ for $\omega\in \Omega_{\rightarrow f} $, then $$\Psi'_{f\rightarrow}\subseteq\Psi_{f\rightarrow},$$
and if $p_{\omega}= p'_{\omega}$ for $\omega\in\Omega_{\rightarrow f}$ and $p_{\omega}\geq p'_{\omega}$ for $\omega\in \Omega_{f\rightarrow } $, then $$\Psi'_{\rightarrow f}\subseteq\Psi_{\rightarrow f}.$$

The combination of the two properties is called full substitutability.
\newline\newline
     \noindent {\bf (Expansion) Full Substitutability (FS)}: The demand of firm $f$ satisfies Expansion Full Substitutability if it satisfies Expansion Same-Side Substitutability and Expansion Cross-Side Complementarity.\newline
Next we introduce monotonicity properties called the Law of Aggregate Demand respectively the Law of Aggregate Supply. Under quasi-linear utility, the two properties are implied by full substitutability. However, in general they are independent of full substitutability.\newline\newline
 \noindent  {\bf Law of Aggregate  Demand  (LAD)}: For $p,p'\in\mathbb{R}^{\Omega_f}$ and each $\Psi'\in D^f(p')$ there exists a $\Psi\in D^f(p)$ such that if  $p_{\omega}= p'_{\omega}$ for $\omega\in\Omega_{f\rightarrow}$ and $p_{\omega}\leq p'_{\omega}$ for $\omega\in \Omega_{\rightarrow f} $, then $$|\Psi_{\rightarrow f}|- |\Psi_{f\rightarrow}|\geq |\Psi'_{\rightarrow f}|- |\Psi'_{f\rightarrow}|.$$\newline
 \noindent  {\bf Law of Aggregate  Supply (LAS)}: For $p,p'\in\mathbb{R}^{\Omega_f}$ and each $\Psi'\in D^f(p')$ there exists a $\Psi\in D^f(p)$ such that if  $p_{\omega}= p'_{\omega}$ for $\omega\in\Omega_{\rightarrow f}$ and $p_{\omega}\geq p'_{\omega}$ for $\omega\in \Omega_{f\rightarrow} $, then $$|\Psi_{f\rightarrow}|- |\Psi_{\rightarrow f}|\geq |\Psi'_{f\rightarrow}|- |\Psi'_{\rightarrow f}|.$$
\subsubsection{Single-Valued Full Substitutability}
Next we introduce the weaker notion of full substitutability where the condition only needs to hold at price vectors where the demand is single-valued.
\newline\newline
\noindent
{\bf Weak Full Substitutability (weak FS)}:
For $p,p'\in\mathbb{R}^{\Omega_f}$ such that $D^f(p)=\{\Psi\}$ and $D^f(p')=\{\Psi'\}$, if $p_{\omega}= p'_{\omega}$ for $\omega\in\Omega_{f\rightarrow}$ and $p_{\omega}\leq p'_{\omega}$ for $\omega\in \Omega_{\rightarrow f} $, then $$\{\omega\in\Psi_{\rightarrow f}:p_{\omega}=p'_{\omega}\}\subseteq \Psi'_{\rightarrow f}\text{ and }\Psi'_{f\rightarrow}\subseteq\Psi_{f\rightarrow},$$
and if $p_{\omega}= p'_{\omega}$ for $\omega\in\Omega_{\rightarrow f}$ and $p_{\omega}\geq p'_{\omega}$ for $\omega\in \Omega_{f\rightarrow } $, then
$$\{\omega\in\Psi_{f\rightarrow }:p_{\omega}=p'_{\omega}\}\subseteq \Psi'_{f\rightarrow }\text{ and }\Psi'_{\rightarrow f}\subseteq\Psi_{\rightarrow f}.$$
\begin{remark}
	The weak and the expansion notions of same sided substitutability are equivalent. See Proposition~\ref{B.1} in Appendix~\ref{Full}. Thus, for two-sided markets the two notions of full substitutability are  equivalent.\qed
\end{remark}
Similarly, we can define weak versions of the laws of aggregate demand and supply.
\newline\newline
\noindent  {\bf Weak Law of Aggregate  Demand  (Weak LAD)}: For $p,p'\in\mathbb{R}^{\Omega_f}$ such that $D^f(p)=\{\Psi\}$ and $D^f(p')=\{\Psi'\}$, if  $p_{\omega}= p'_{\omega}$ for $\omega\in\Omega_{f\rightarrow}$ and $p_{\omega}\leq p'_{\omega}$ for $\omega\in \Omega_{\rightarrow f} $, then $$|\Psi_{\rightarrow f}|- |\Psi_{f\rightarrow}|\geq |\Psi'_{\rightarrow f}|- |\Psi'_{f\rightarrow}|.$$\newline
\noindent  {\bf Weak Law of Aggregate  Supply (Weak LAS)}: For $p,p'\in\mathbb{R}^{\Omega_f}$ such that $D^f(p)=\{\Psi\}$ and $D^f(p')=\{\Psi'\}$, if  $p_{\omega}= p'_{\omega}$ for $\omega\in\Omega_{\rightarrow f}$ and $p_{\omega}\geq p'_{\omega}$ for $\omega\in \Omega_{f\rightarrow} $, then $$|\Psi_{f\rightarrow}|- |\Psi_{\rightarrow f}|\geq|\Psi'_{f\rightarrow}|- |\Psi'_{\rightarrow f}|.$$\newline
The following example shows that the two notions of full substitutability that we have defined can differ for non-quasi-linear utility. In Section~\ref{lattice}, we will use the example to show that under weak FS the lattice result  in our paper does not necessarily hold.
Similar examples show that weak LAD/weak LAS is strictly weaker than LAD/LAS, see Example~\ref{ex3} in Appendix~\ref{Full}. %See Figure~\ref{Figure2} for a geometric representation of the demand in the example.
%\begin{figure}
%        \centering
%       \includegraphics[width=0.8\textwidth]{ex4} % first figure itself
%      \caption{The demand in price space for $p_{\alpha^2}=1=p_{\beta^2}$.}\label{Figure2}
%\end{figure}
%\begin{lemma}
%Let $u^f$ be a utility function inducing a demand correspondence
%$D^f$ satisfying FS, LAD and LAS.  Let $\omega\in\Omega_f$, let $p,p'\in\mathbb{R}^{\Omega_f}$ with $p'_{-\omega}=p_{\omega}$
%and $\Psi'\in D^f(p')$.
%
%If $p_{\omega}'>p_{\omega}$, then
%\begin{enumerate}
%   \item $\Psi'\cup\{\omega\}\in D^f(p)$
%   \item there exists a $\omega'\in\Omega_{\rightarrow f}\setminus\{\omega\}$ with $\Psi\setminus\{\omega'\}\cup\{\omega\}\in D^f(p)$
%   \item there exists a $\omega'\neq\omega$ and $\omega''\in\Omega_{f\rightarrow}$ with $\Psi'\cup\{\omega,\omega'\}$
%\end{enumerate}
%
%If $p_{\omega}\leq p'_{\omega}$ for
%$\omega\in\Omega_{f\rightarrow}\setminus\Psi'_{f\rightarrow}$,
%$p_{\omega}\geq p'_{\omega}$ for $\omega\in \Psi'_{f\rightarrow}$,
%$p_{\omega}\geq p'_{\omega}$ for $\omega\in\Omega_{\rightarrow
%   f}\setminus\Psi'_{\rightarrow f}$ and $p_{\omega}\leq p'_{\omega}$
%for $\omega\in \Psi'_{\rightarrow f}$, then $\Psi'\in D^f(p)$.
%Moreover, if all of the inequalities are strict, then
%$D^f(p)=\{\Psi'\}.$
%\end{lemma}
\begin{example}\label{ex}
	Consider four trades $\Omega=\{\alpha^1,\alpha^2,\beta^1,\beta^2\}$ with $f=b(\alpha^1)=b(\alpha^2)=s(\beta^1)=s(\beta^2)$.
	We let
	\begin{align*}
	&u^{f}(\emptyset)=0,\\
	&u^f(\{\alpha^i,\beta^j\},p_{\alpha^i},p_{\beta^j})=2-p_{\alpha^i}+p_{\beta^j},\text{ for }i,j,=1,2,\\
	&u^{f}(\{\alpha^1,\alpha^2,\beta^1,\beta^2\},p)=
	4-\exp\left(\frac{p_{\alpha^1}+p_{\alpha^2}}{2}-1\right)
	-\exp\left(1-\frac{p_{\beta^1}+p_{\beta^2}}{2}\right).
	\end{align*}
	We let $u^f(\Psi,p)=-\infty$ for any other $\Psi\subseteq\Omega$.
	Observe that $$D^f(1,1,1,1)=\{\{\alpha^1,\beta^1\},\{\alpha^1,\beta^2\},\{\alpha^2,\beta^1\},\{\alpha^2,\beta^2\},\{\alpha^1,\alpha^2,\beta^1,\beta^2\}\}$$
	but $$D^f(0,1,1,1)=\{\{\alpha^1,\beta^1\},\{\alpha^1,\beta^2\}\}.$$
	As $\{\alpha^1,\alpha^2,\beta^1,\beta^2\}\in D^f(1,1,1,1)$, FS would require that there is a $\Psi\in D^f(0,1,1,1)$ with $\{\beta^1,\beta^2\}\subseteq \Psi$.  Hence FS is not satisfied.  As the demand at $(0,1,1,1)$ and $(1,1,1,1)$ is multi-valued, Weak FS does not impose any structure here. More generally, note that the bundle $\{\alpha^{1},\alpha^2,\beta^1,\beta^2\}$ is only demanded at prices $(1,1,1,1)$ so that if we replace $u^f$ by the utility function $\tilde{u}^f$ such that
	\begin{align*}
	&\tilde{u}^f(\{\alpha^1,\alpha^2,\beta^1,\beta^2\},\cdot)=-\infty\\
	&\tilde{u}^f(\Psi,\cdot)=u^f(\Psi,\cdot)\text{ for }\Psi\neq\{\alpha^1,\alpha^2,\beta^1,\beta^2\},
	\end{align*}
	only the demand at prices $(1,1,1,1)$ changes. One readily checks that $\tilde{u}^f$ satisfies FS. Hence $u^f$ satisfies Weak FS. Note, moreover, that $u^f$ satisfied LAD and LAS.
	\qed
\end{example}
\begin{remark}\label{rmk}
In the example, the bundle $\{\alpha^1,\alpha^2,\beta^1,\beta^2\}$ is only demanded at prices $(1,1,1,1)$, but not at any price vector in the neighborhood. One can show (see Proposition~\ref{single} in Appendix~\ref{Full}) that if there are no such "isolated" bundles, i.e.~bundles that are demanded at a price vector but nowhere in the neighborhood of it, then weak FS and FS are equivalent. Formally this requirement is the following:\newline\newline 
\noindent  {\bf No Isolated Bundles (NIB)}: For each $p\in\mathbb{R}^{\Omega_f}$, $\Psi\in D^f(p)$ and $\epsilon>0$ there is a $q\in\mathbb{R}^{\Omega_f}$ with $\|p-q\|<\epsilon$ and $D^f(q)=\{\Psi\}.$\newline\newline
It is not generally true that, conversely, FS implies NIB (see Example~\ref{ex3} in Appendix~\ref{Full}). However, CSC (SSS is not necessary here), LAD and LAS together imply NIB (see Proposition~\ref{NIB} in Appendix~\ref{Full}). This observation will be important for the proof of our main result, Theorem~\ref{Util}.\qed
\end{remark}
Next we show that for each utility function satisfying weak FS, weak LAD and weak LAS,  we can construct a demand correspondence that satisfies FS, LAD and LAS by removing isolated bundles from the original demand. The resulting demand can be rationalized by a continuous and monotonic utility function. Put differently, we show that the two notions of Full
Substitutability are almost equivalent in the following sense: for
each utility function $u^f$ for which the induced demand $D^f$ satisfies weak FS, weak LAD, and weak LAS, there is a
corresponding utility function $\tilde{u}^f$ for which the induced demand $\tilde{D}^f$ satisfies FS, LAD, and LAS, and such that $\tilde{D}^f$ selects from $D^f$. The utility function can be chosen such that the induced indirect utility is the same. The result is proved in Appendix~B.
\begin{proposition}\label{weak}
	Let $u^f$ satisfy weak FS, weak LAD, and weak LAS. Then there is a utility function $\tilde{u}^f$ that satisfies FS, LAD, and LAS such that the induced indirect utility functions are the same $$v^f(p)=\tilde{v}^f(p)\text{ for each }p\in\mathbb{R}^{\Omega_f},$$
	and the induced demand is a selection from the original demand,
	$$\tilde{D}^f(p)\subseteq D^f(p)\text{ for each }p\in\mathbb{R}^{\Omega_f}.$$
\end{proposition}

\section{Results}\label{lattice}
\subsection{The Lattice Theorem and the Rural Hospitals Theorem}
As our first main result we establish that equilibrium prices in trading networks form a lattice and that (modulo indifferences) for each firm the difference between the number of signed upstream and downstream contracts is the same in each equilibrium. The join and meet are the coordinate-wise maximum and minimum of the two price vectors under consideration, i.e.~the lattice is a sublattice of $\mathbb{R}^{\Omega}$ with the usual partial order. These results extend results established by~\cite{Hatfield2013} for the case of quasi-linear utility functions.

In the following, a {\bf competitive equilibrium} for utility profile $u=(u^f)_{f\in F}$ is an arrangement $[\Psi,p]\in2^{\Omega}\times\mathbb{R}^{\Omega}$ such that for each $f\in F$ and the demand $D^f$ induced by $u^f$ we have $\Psi_f\in D^f(p)$. We call $(\Psi,(p_{\omega})_{\omega\in\Psi})$ the {\bf equilibrium allocation} induced by $[\Psi,p].$ We denote the set of equilibrium price vectors for $u$ by $\mathcal{E}(u)$ and define for each price vector $p\in\mathbb{R}^{\Omega}$ the (possibly empty) set $\mathcal{E}(u,p):=\{\Psi\subseteq\Omega:\Psi_f\in D^f(p)\text{ for each }f\in F\}$ of sets of trades that support $p$ as a competitive equilibrium under $u$.

\begin{theorem}\label{Util}
Let $u$ be a utility profile such that for each firm the induced demand satisfies FS, LAD and LAS.
\begin{enumerate}
\item {\bf Lattice Theorem}: Let $p,p'\in\mathcal{E}(u)$ be equilibrium prices. Then $\bar{p},\underline{p}\in\mathbb{R}^{\Omega}$ defined by $$\bar{p}_{\omega}:=\max\{p_{\omega},p'_{\omega}\},\quad \underline{p}_{\omega}:=\min\{p_{\omega},p'_{\omega}\},$$
are equilibrium prices.
\item {\bf Rural Hospitals Theorem}: Let $p,p'\in\mathcal{E}(u)$ be equilibrium prices. For each $\Psi\in\mathcal{E}(u,p)$ there exists a $\Psi'\in\mathcal{E}(u,p')$ such that for each $f\in F$ we have $|\Psi_{\rightarrow f}|-|\Psi_{f\rightarrow }|=|\Psi'_{\rightarrow f}|-|\Psi'_{f\rightarrow }|$.
\end{enumerate}
\end{theorem}
The proof and all subsequent proofs of this section are in Appendix~\ref{anotherApp}. However, for the moment we give a sketch of the proof strategy and comment on the challenges when generalizing from quasi-linear to more general preferences.
Let $p,p'\in\mathcal{E}(u)$ be two equilibrium price vectors and consider the pairwise maximum $\bar{p}\in\mathbb{R}^{\Omega}$ (a dual argument works for the pairwise minimum $\underline{p}$). Suppose for the moment that the demand for each firm is single-valued at $p$ and at $p'$, i.e.~each firm $f$ has unique optimal bundles of trades $\Psi_f$ and $\Psi_f'$ at the equilibrium prices $p$ and $p'$. For each firm $f$ let $\bar{\Psi}_f\in D^f(\bar{p})$ be a bundle demanded at $\bar{p}$. Full Substitutability for individual firms implies full substitutability for the aggregate demand so that we obtain
$$\bigcup_{f\in F}\bar{\Psi}_{f\rightarrow}\subseteq\{\omega\in\Psi:p_{\omega}\geq p'_{\omega}\}\cup\{\omega\in\Psi':p_{\omega}'\geq p_{\omega}\}\subseteq\bigcup_{f\in F}\bar{\Psi}_{\rightarrow f}.$$
The equation states that there is no excess supply of trades at $\bar{p}.$
The laws of aggregate demand and supply, can be used to show that $$|\bigcup_{f\in F}\bar{\Psi}_{f\rightarrow}|\geq|\bigcup_{f\in F}\bar{\Psi}_{\rightarrow f}|.$$
Combining these two observations shows that the market clears, i.e.
$$\bigcup_{f\in F}\bar{\Psi}_{f\rightarrow}=\bigcup_{f\in F}\bar{\Psi}_{\rightarrow f},$$
and thus $\bar{p}\in\mathcal{E}(u).$

Suppose now that we want to generalize the argument to the case of multi-valued demand at the equilibrium prices. A natural idea is to use a perturbation argument: Continuity and monotonicity of utility in transfers allows us (see Lemma~\ref{Fede}) to perturb price vectors to obtain a single-valued selection from the demand correspondence at prices $p,p'$ and $\bar{p}$.
\begin{lemma}\label{12}
    Let $u^f$ be a utility function inducing a demand correspondence $D^f$ satisfying weak FS, weak LAD and weak LAS. Let $P\subseteq\mathbb{R}^{\Omega_f}$ be finite.
    Then there is a (single-valued) demand function $\tilde{D}^f:P\rightarrow2^{\Omega_f}$ that selects from $D^f$, i.e.~$\tilde{D}^f(p)\in D^f(p)$ for $p\in P$ and satisfies FS, LAD and LAS.
\end{lemma}
Once perturbed, the argument above could be applied to the perturbed
price vectors. However, this line of argument has a flaw: There is
no guarantee that the trades demanded at the perturbed prices support an equilibrium, since not every collection of demanded
trades at an equilibrium price vector support these prices as an
equilibrium. This problem would not occur under quasi-linear utility: for quasi-linear utility it is easy to show that the set of
competitive equilibrium price vectors is convex. Thus, for quasi-linear utility either there is a unique equilibrium price vector, in which case the lattice result trivially holds, or there is  in the neighborhood of each equilibrium price vector another equilibrium price vector at which demand is single-valued for each firm (since price vectors where demand is single-valued are dense in price space). For general utility functions, convexity and
more generally connectedness of the set of equilibria can
fail\footnote{This is already true for one-to-one matching with transfers which is a special case of our model. See the example in~\cite{RothSotomayor1988}.} so that a perturbation argument might fail. 

 While a naive perturbation argument fails to work, we can
use a more intricate perturbation argument. We
perturb prices for each firm individually. Importantly, we can rely on the observation (recall Remark~\ref{rmk}) that for each firm $f$ there are prices $q$ (in general different for different firms) close to $p$ where the equilibrium set of trades $\Psi_f$ is the unique demanded bundle of trades. %In particular one can show that for
%each demanded bundle of trades there exist a price vector in the
%neighborhood of the original price vector where only this bundle of
%trades is demanded.
%The combination of full substitutability and the laws of aggregate demand and supply imply an invariance property of the demand that will be crucial for our lattice result. It states that if a bundle of trades is demanded by a firm $f$ at a price vector, then this bundle is also demanded by $f$, if all downstream trades for $f$ in the bundle become more expensive, all downstream trades $f$ not in the bundle become cheaper, all upstream trades $f$ in the bundle become cheaper, and all upstream trades for $f$ not in the bundle become more expensive.
The following lemma is the main ingredient in the proof of the theorem.
\begin{lemma}\label{lemma}
    Let $u^f$ be a utility function inducing a demand correspondence
    $D^f$ satisfying FS, LAD and LAS.
    Let $p,p'\in \mathbb{R}^{\Omega_f}$ and define $\bar{p},\underline{p}\in\mathbb{R}^{\Omega_f}$ by
    $$\overline{p}_{\omega}:=\max\{p_{\omega},p'_{\omega}\},\quad\underline{p}_{\omega}:=\min\{p_{\omega},p'_{\omega}\}.$$
    Let $\Psi\in D^f(p)$ and $\Psi'\in D^f(p')$.
    \begin{enumerate}
        \item There is a $\bar{\Psi}\in D^f(\bar{p})$ with
        \begin{align*}&\{\omega\in \Psi_{\rightarrow f}: p_{\omega}\geq p'_{\omega}\}\cup\{\omega\in  \Psi'_{\rightarrow f}:p'_{\omega}>p_{\omega}\}\subseteq \bar{\Psi}_{\rightarrow f},\\
        &\bar{\Psi}_{f\rightarrow}\subseteq\{\omega\in \Psi_{f\rightarrow}: p_{\omega}\geq p'_{\omega}\}\cup\{\omega\in  \Psi'_{f\rightarrow }:p'_{\omega}>p_{\omega}\}.\end{align*}
        \item There is a $\underline{\Psi}\in D^f(\underline{p})$ with
        \begin{align*}&
        \underline{\Psi}_{\rightarrow f}\subseteq\{\omega\in \Psi_{\rightarrow f}: p'_{\omega}\geq p_{\omega}\}\cup\{\omega\in \Psi'_{\rightarrow f}: p_{\omega}> p'_{\omega}\},\\
        &
        \{\omega\in \Psi_{f\rightarrow }: p'_{\omega}\geq p_{\omega}\}\cup\{\omega\in \Psi'_{f\rightarrow}: p_{\omega}> p'_{\omega}\}\subseteq\underline{\Psi}_{f\rightarrow}
        .\end{align*}
        \item $\bar{\Psi}$ and $\underline{\Psi}$ can be chosen such that
        $$|\underline{\Psi}_{\rightarrow f}|-|\underline{\Psi}_{f\rightarrow }|\geq|{\Psi}_{\rightarrow f}|-|{\Psi}_{f\rightarrow }|\geq|\bar{\Psi}_{\rightarrow f}|-|\bar{\Psi}_{f\rightarrow }|.$$

    \end{enumerate}
\end{lemma}
With the lemma the theorem follows straightforwardly. The lemma and the first part of the theorem fail to hold if we replace FS by weak FS, as the following example shows. A similar example shows that the second part of the theorem fails if LAD (LAS) is replaced by weak LAD (weak LAS), see Example~\ref{ex3} in Appendix~\ref{anotherApp}.
\begin{continuance}{ex}
    Consider the set of trades $\Omega=\{\alpha^1,\alpha^2,\beta^1,\beta^2\}$ and firm $f$ with the utility function $u^f$ as defined in Example~\ref{ex}. The induced demand $D^f$ satisfies weak FS as previously shown. Moreover, for each $p\in\mathbb{R}^{\Omega_f}$ and $\Psi\in D^f(p)$ we have $|\Psi_{f\rightarrow}|=|\Psi_{\rightarrow f}|$. Thus $D^f$ satisfies LAD and LAS.  Consider four additional firms $s^1,s^2,b^1,b^2$ with $s^1=s(\alpha^1), s^2=s(\alpha^2),b^1=b(\beta^1)$ and $b^2=b(\beta^2)$. Define utility functions for the additional firms as follows:
    For $i=1,2$ define
    \begin{align*}
    &u^{s^i}(\{\alpha^i\},p_{\alpha^i})=p_{\alpha^i},\\
    &u^{b^i}(\{\beta^i\},p_{\beta^i})=2-p_{\beta^i},\\
    &u^{s^i}(\emptyset)=u^{b^i}(\emptyset)=0.
    \end{align*}  Observe that the equilibria for $u$ are $[\Omega,(1,1,1,1)]$ and $[\{\alpha^i,\beta^j\},(0,0,2,2)]$ for $i,j=1,2$. In particular, the vector $(1,1,2,2)$ is not an equilibrium price vector, since $D^{s^1}(1,1,2,2)=\{\{\alpha^1\}\}$ and $D^{s^2}(1,1,2,2)=\{\{\alpha^2\}\}$ but $D^f(1,1,2,2)=\{\{\alpha^1,\beta^1\},\{\alpha^1,\beta^2\},\{\alpha^2,\beta^1\},\{\alpha^2,\beta^2\}\}$. 
    \qed
%%      Note that the Lemma~\ref{lemma} fails to hold under weak FS: For the utility
  %     function $u^f$ in Example~1, we have
   %    $\Psi'=\{\alpha_1,\alpha_2,\beta_1,\beta_2\}\in D^f(1,1,1,1),$ but
    %   $\Psi'\notin D^f(0,1,1,1)$.\qed}
    %\qed
\end{continuance}
It is well-known that the theorem fails to hold without FS, even for quasi-linear utility functions. The following example shows that the first part of the theorem fails without LAD. More generally, the example shows that without LAD the set of equilibria can even fail to be a lattice with respect to the (weaker) partial ordering induced by terminal sellers' preferences.
\begin{example}\label{ex22}
	Let $\Omega=\{\omega_1,\omega_2,\omega_3\}$. Let $b(\omega_i)=f$ for $i=1,2,3$ and $s(\omega_i)\neq s(\omega_j)$ for $i\neq j$. We let $u^{s(\omega_i)}(\omega_i,p_{\omega_i})= p_{\omega_i},u^{s(\omega_i)}(\emptyset)=0$ for $i=1,2$ and $u^{s(\omega_3)}(\omega_3,p_{\omega_3})=-\infty,u^{s(b_3)}(\emptyset)=0$.  We define $u^f$ by
	\begin{align*}
	&u^{f}(\emptyset)=0,\\
	&u^{f}(\{\omega_1,\omega_2\},(p_{\omega_1},p_{\omega_2}))=2-p_{\omega_1}-p_{\omega_2},\\
	&u^{f}(\{\omega_1,\omega_2,\omega_3\},p)=1-\frac{1}{1+\exp(-(p_{\omega_1}+p_{\omega_2}+p_{\omega_3}))},
	\end{align*}
	and $u^f(\Psi,\cdot)=-\infty$ else.
	%    See Figure~\ref{ex2} for a geometric representation of the demand of $f$, and for an illustration why LAD fails for firm $f$. The utility functions of the three sellers are unit supply and thus satisfy FS. In Appendix~\ref{FS} we show that $u^f$ satisfies FS.
	%    \begin{figure}
	%    \centering
	%    \includegraphics[width=0.65\textwidth]{ex6} % first figure itself
	%    \caption{The demand of $f$ in price space for $p_{1}=p_{2}$ is represented by black lines. In the orange area price vectors are equilibrium prices. Note that for price vectors $p\leq p'$ as in the picture, LAD is violated: $D^f(p)=\{\{\omega_1,\omega_2\}\}$, but $D^f(p')=\{\omega_1,\omega_2,\omega_3\}$.}\label{ex2}
	%\end{figure}
	
Consider the price vectors $p=(0,1,0)$ and $p'=(1,0,0)$. Note that $\{\omega_1,\omega_2\}\in D^f(p)$ and $\{\omega_1,\omega_2\}\in D^f(p')$. Moreover, we have $  D^{s(\omega_1)}(p)=\{\{\omega_1\}\}= D^{s(\omega_1)}(p')$, $ D^{s(\omega_2)}(p)=\{\{\omega_2\}\}= D^{s(\omega_2)}(p')$ and $D^{s(\omega_3)}(p)=\{\emptyset\}=D^{s(\omega_3)}(p')$.	Thus $p$ and $p'$ are equilibrium price vectors. Suppose there is a $\bar{p}$ that each terminal seller weakly prefers to $p$ and $p'$, i.e.~$v^{s(\omega_1)}(\bar{p})\geq\max\{v^{s(\omega_1)}(p),v^{s(\omega_1)}(p')\}=1$, $v^{s(\omega_2)}(\bar{p})\geq\max\{v^{s(\omega_2)}(p),v^{s(\omega_2)}(p')\}=1$, and $v^{s(\omega_3)}(\bar{p})\geq v^{s(\omega_3)}(p)=v^{s(\omega_3)}(p')=u^{s(\omega_3)}(\emptyset)=0$. Thus $\bar{p}_{\omega_1}\geq 1$ and $\bar{p}_{\omega_2}\geq1$. But then $D^f(\bar{p})=\{\{\omega_1,\omega_2,\omega_3\}\}$. Moreover, $D^{s(\omega_3)}(\bar{p})=\{\emptyset\}.$ Thus,  there is no such equilibrium price vector $\bar{p}$.
	
To check that $u^f$ satisfies FS, first note that for each $p\in\mathbb{R}^{\Omega_f}$, we have $u^f(\{\omega_1,\omega_2,\omega_3\},p)>0=u^f(\emptyset)$. Thus, at each $p\in\mathbb{R}^{\Omega_f}$ we have $D^f(p)\subseteq \{\{\omega_1,\omega\},\{\omega_1,\omega_2,\omega_3\}\}$ and the only possible FS violation could occur for $p\leq p'$ with $p'_{\omega_3}=p_{\omega_3}$ and $\{\omega_1,\omega_2,\omega_3\}\in D^f(p)$. However, if $u^f(\{\omega_1,\omega_2,\omega_3\},p)\geq u^f(\{\omega_1,\omega_2\},p)$, then, as $u^f(\{\omega_1,\omega_2,\omega_3\},p)- u^f(\{\omega_1,\omega_2\},p)$ is increasing in $p_{\omega_1}$ and in $p_{\omega_2}$ for each $p_{\omega_3}$, we have $u^f(\{\omega_1,\omega_2,\omega_3\},p')\geq u^f(\{\omega_1,\omega_2\},p').$ 
Thus, FS holds.
	\qed
	%   $$u^{f_1}(\{\alpha^1,\beta\},p)=\begin{cases}
	%   3-p_1-p_2\quad&\text{ if }p_1+p_2\leq0\\
	%   3-2\sqrt{\frac{p_{{1}}+p_{{2}}}{4}}\quad&\text{ if }4\geq p_{{1}}+p_{{2}}>0,\\5-p_{{1}}-p_{{2}}\quad&\text{ else}
	%   \end{cases}$$
	%   Let $$u^{f_3}(\beta,p_{\beta})=p_{\beta},u^{f_3}(\emptyset)$$
\end{example}
 \subsection{Extremal equilibria}
So far we have not considered whether competitive equilibria exist in our model and, in principle, the lattice in Theorem~\ref{Util} could be empty. Next we show that under the additional assumption of bounded compensating variations, as introduced by~\cite{Ravi}, side-optimal equilibria exist, i.e.~there exist an equilibrium that is a most preferred equilibrium for all terminal buyers and an equilibrium that is a most preferred equilibrium  for all terminal sellers.
\newline\newline
\noindent
{\bf Bounded compensating variations}: The utility function of firm $f$ satisfies bounded compensating variations if for each $\Psi\subseteq\Omega$ we have $$\inf_{p\in\mathbb{R}^{\Psi}:u^f(\Psi,p)>u^f(\emptyset)}\left(\sum_{\omega\in\Psi_{f\rightarrow}}p_{\omega}-\sum_{\omega\in\Psi_{\rightarrow f}}p_{\omega}\right)>-\infty.
$$
%Next we provide an example to illustrate that BCV together with FS, LAD and LAS is not sufficient to guarantee the existence of extremal equilibria. While BCV guarantees that prices for trades that terminal buyers obtain in equilibrium are bounded, this does not guarantee the existence of an extremal equilibrium. The problem is that the projection of the (closed) set of equilibrium price vectors to the space of trades involving terminal buyers can fail to be closed.

%Next we strengthen the result of~\cite{Ravi}, by showing that the combination of BCV, FS, LAD/LAS does not only imply the existence of an equilibrium, but also the existence of side-optimal equilibria.
% Note that in contrast to Theorem~\ref{Util}, the result does not assume the Law of Aggregate Demand and Supply.
The condition rules out for example the case that for a trade the seller would never sell under any price and the buyer would buy under any price, by  guaranteeing that utility for individually rational allocations is bounded for all agents. \cite{Ravi} show that under the assumption of BCV, (weak) FS equilibria exist. We generalize the result by proving the existence of side-optimal equilibria:
\begin{theorem}[\bf Existence of Extremal Equilibria]\label{extreme}
Under the assumption of BCV, FS, LAD, LAS, there exists a seller-optimal equilibrium, i.e.~a $\bar{p}\in\mathcal{E}(u)$ such that for each terminal seller $f\in F$:
$$v^f(\bar{p})\geq v^f(p)\text{ for each }p\in\mathcal{E}(u),$$
and a buyer-optimal equilibrium, i.e.~a~$\underline{p}\in\mathcal{E}(u)$ such that for each terminal buyer $f\in F$:
$$v^f(\underline{p})\geq v^f(p)\text{ for each }p\in\mathcal{E}(u).$$
\end{theorem}
\begin{remark}
    \cite{Ravi} also introduce a stronger regularity condition, called bounded willigness to pay (BWP), which guarantees that prices in individually rational allocations are bounded for all agents.
    \newline\newline
    \noindent
    {\bf Bounded willingness to pay (BWP)}: The utility function $u^f$ satisfies bounded willingness to pay if there exists a $K\geq0$ such that for all $p\in\mathbb{R}^{\Omega_f}$ and $\Psi\in D^f(p)$
    if $\omega\in\Psi_{\rightarrow f}$ then $p_{\omega}<K$ and
    if $\omega\in \Psi_{f\rightarrow }$ then $p_{\omega}>-K$.\newline\newline
     The previous result also holds if BCV is replaced by BWP, and more generally one can show that in this case the set of equilibrium prices for realized trades is compact (see Proposition~\ref{BWP} and its proof in the appendix).
Under the assumptions of (weak) FS and BWP, \cite{Ravi} establish that equilibrium allocations are equivalent to trail-stable allocations. Thus, under BWP, FS, LAD, LAS there is a seller-optimal trail-stable allocation and a buyer-optimal trail-stable allocation. In the case of no frictions, BWP is implied by requiring, that $u^f(\emptyset)>-\infty$ and utility functions have {\bf full range}: for each $\Psi\subseteq \Omega_f$, $\Psi\neq\emptyset$, $\tilde{u}^f(\Psi,\cdot)$ is a surjective function onto $\mathbb{R}$.
      \qed %The assumption of finite utility alone, however, does not imply BWP for general utility functions.% as the following example demonstrates:   \begin{example}
    %Consider two trades $\Omega=\{\omega_1,\omega_2\}$ with $s(\omega_1)=f=b(\omega_2)$. Let $u^f(\emptyset),u^f(\omega_1,p_{\omega_1})=p_{\omega_1},u^f(\omega_2,p_{\omega_2})=-p_{\omega_2}$ and $u^f(\{\omega_1,\omega_2\},p)=2(p_{\omega_1}-p_{\omega_2}).$ Observe that for each $p\in\mathbb{R}^{\omega}$ with $p_{\omega_1}>2p_{\omega_2}\geq0,$ we have $D^f(p)=\{\{\omega_1,\omega_2\}\}$. In particular, for each $K>0$ there is a $p$ with $p_{\omega_2}>K$ such that $\omega_2$ is demanded under $p$.
    %\end{example}
\end{remark}
\subsection{Strategic Considerations}
The existence of buyer-optimal equilibria established in Theorem~\ref{extreme}, allows us to obtain a group-incentive compatibility result.\footnote{In the following we talk about incentives for terminal buyers. A completely analogous result also holds for terminal sellers.} In the following, a domain of utility profiles is a set $\mathcal{U}=\bigtimes_{f\in F}\mathcal{U}_f$ where $\mathcal{U}_f$ is a set of (continuous and monotonic) utility functions for firm $f$. A {\bf mechanism}  is a function  $\mathcal{M}:\mathcal{U}\rightarrow\mathcal{A}$. A mechanism is {\bf (weakly) group-strategy-proof} for a set of workers $F'\subseteq F$ on the domain $\mathcal{U}'\subseteq\mathcal{U}$ if for each $u,\tilde{u}\in\mathcal{U}'$ with  $\tilde{u}^{-F'}={u}^{-F'}$, there exist a $f\in F'$ with
$$u^f(\mathcal{M}(u))\geq u^f(\mathcal{M}(\tilde{u})).$$

Theorem~\ref{extreme} allows us to define a class of focal mechanisms on the domain of utility profiles satisfying  BCV, FS, LAD and LAS: a {\bf buyer-optimal mechanism} maps to each utility profile a buyer-optimal equilibrium allocation.

To obtain a group-strategy-proofness results for terminal buyers for buyer-optimal mechanisms, we have to restrict the domain. In the following a {\bf unit demand} utility function is a $u^f$ such that for the induced demand $D^f$ at each $p\in\mathbb{R}^{\Omega_f}$ and $\Psi\in D^f(p)$ we have $|\Psi_{\rightarrow f}|\leq1$.

\begin{theorem}[\bf Group-Strategy-Proofness]\label{GSP}
Each buyer-optimal mechanism is group-strategy-proof for terminal buyers on the domain of utility profiles such that terminal buyers' utility functions satisfy Unit Demand and BCV and all other firms' utility functions satisfy BCV, FS, LAD and LAS.
\end{theorem}
In view of Proposition~\ref{weak}, we can extend the construction to profiles  satisfying  BCV, weak (!) FS, weak LAD and weak LAS. For each such profile $u$ there exists a corresponding profile $\tilde{u}$ satisfying BCV, FS, LAD and LAS such that the indirect utility functions are the same for both profiles. The mechanism that assigns to each profile $u$ a buyer-optimal equilibrium allocation under a corresponding profile $\tilde{u}$ is group-strategy-proof for terminal buyers (since for terminal buyers (and terminal sellers) the weak FS and the FS condition coincide), and the assigned allocations are equilibrium allocations under $u$ as well.
\begin{corollary}[\bf Group-Strategy-Proofness under weak FS]\label{weakGSP}
    On the domain of utility profiles such that terminal buyers' utility functions satisfy Unit Demand and BCV and all other firms' utility functions satisfy BCV, weak FS, weak LAD and weak LAS, there exists a group-strategy-proof mechanisms for terminal buyers that implements a competitive equilibrium.
\end{corollary}
\section{Applications}
\subsection{Two-sided Matching Markets}\label{two}
The results in the previous sections immediately apply to two-sided matching markets. In this case, the results generalize previously known results for two-sided matching markets in two directions: we provide a lattice result, a rural hospitals theorem and a group-strategy-proofness result for markets with a) possibly non-quasi-linear preferences for both sides of the market b) the possibility that it is infeasible for a hospital to hire certain groups of doctors. As remarked in Section~\ref{FSS}, the weak version of Full Substitutability is sufficient to obtain the results for two-sided markets.
 
Instead of a set of firms, the economy now consists of a finite set of {\bf hospitals} $H$ and a finite set of {\bf doctors} $D$. Each hospital  $h$ has a utility function $u^h:\{(D',p):D'\subseteq D,p\in \mathbb{R}^{D'}\}\rightarrow \mathbb{R}\cup\{-\infty\}$ that assigns to each $D'\subseteq D$ and price vector $p\in\mathbb{R}^{D'}$ a utility level. We extend $u^h$ to $2^D\times \mathbb{R}^{D}$ by letting $u^h(D',p):=u^h(D',(p_{d})_{d\in D'})$. We allow the utility function to take on a value of $-\infty$ to indicate that it is infeasible for the hospital to hire a particular group of doctors. This allows us for example to incorporate institutional constraints such as the ``generalized interval constraints" characterized by~\cite{Kojimaetal2019} which specify a lower and an upper bound on the number of doctors a hospital can hire. We assume that  $u^h(D',p)=-\infty$ implies $u^h(D',p')=-\infty$ for each $p'\in\mathbb{R}^{D'}$. We assume that there is at least one group of doctors $D'\subseteq D$ that is feasible to hire, i.e.~such that $u^h(D',\cdot)>-\infty.$  Moreover, we require that for $u^h(D',\cdot)\neq -\infty$, the utility function $u^h (D',\cdot)$ is continuous and strictly decreasing in prices. The utility function induces a demand correspondence $D^h:\mathbb{R}^D\rightrightarrows 2^D$ by $D^h(p):=\text{argmax}_{D'\subseteq D}u^h(D',p).$ We assume that doctors are gross substitutes for hospitals. We only need to require the condition for price vectors where the demand is single-valued.
\newline\newline
\noindent
{\bf Weak Gross Substitutability}:
For $p,p'\in\mathbb{R}^{D}$ with $p\leq p',$  ${D}^h(p)=\{D'\}$ and ${D}^h(p')=\{D''\}$ we have $\{d\in D':p'_{d}=p_x\}\subseteq D''$.\newline\newline
Moreover, we require the law of aggregate demand.
\newline\newline
{\bf Law of Aggregate Demand}:
For $p,p'\in\mathbb{R}^{D}$ with $p\leq p'$ and each $D'\in D^f(p')$ there is a $\tilde{D}\in D^f(p)$  with $|\tilde{D}|\geq|D'|$.
\newline\newline
%Analogously, the {\bf Weak Law of Aggregate Demand} requires that the law of aggregate demand holds for price vectors where demand is single-valued.
%We say that the demand function of firm $f$ satisfies {\bf bounded willingness to pay} if there exists a $K\geq0$ such that for all $p\in\mathbb{R}^{W}$ and $W'\in D^f(p)$
%if $w\in W'$ then $p_{w}<K$. Bounded willingness to pay is for example satisfied, if utility functions are of full range for $u^f(W',\cdot)>-\infty$, and for each worker $w\in W$ there is a $W'\subseteq W$ such that $u^f(W',\cdot)>-\infty.$
Each doctor $d$ has a utility function $u^d:H\times\mathbb{R}\cup\{\emptyset\}\rightarrow \mathbb{R}$ that is strictly increasing and continuous in its second argument. We extend $u^d$ to $H\cup\{\emptyset\}\times \mathbb{R}^{H}$ by letting $u^d(h,p):=u^d(h,p_{hd})$ and $u^d(\emptyset,p):=u^d(\emptyset)$.

%Each doctor  $d$ has a utility function $u^d:\{(H',p):H'\subseteq H,p\in \mathbb{R}^{H'}\}\rightarrow \mathbb{R}\cup\{-\infty\}$ that assigns to each $H'\subseteq H$ and price vector $p\in\mathbb{R}^{H'}$ a utility level. We extend $u^d$ to $2^H\times \mathbb{R}^{H}$ by letting $u^d(H',p):=u^d(H',(p_{h})_{h\in H'})$.  As for hospitals, we allow the utility function to take on a value of $-\infty$. We assume that  $u^d(H',p)=-\infty$ implies $u^d(H',p')=-\infty$ for each $p'\in\mathbb{R}^{H'}$. We assume that there is at least one group of hospitals $H'\subseteq H$ such that $u^d(H',\cdot)>-\infty.$  We assume that for $u^d(H',\cdot)\neq -\infty$, the utility function $u^d(H',\cdot)$ is continuous and strictly increasing in prices. The utility function induces a demand function $D^d:\mathbb{R}^H\rightarrow 2^H$ by $D^d(h):=\text{argmax}_{D'\subseteq D}u^h(D',p).$ We assume that doctors are gross substitutes for hospitals. We only need to require the condition for price vectors where the demand is single-valued.
%\newline\newline
%\noindent
%{\bf Weak Gross Substitutability}:
%For $p,p'\in\mathbb{R}^{D}$ with $p\leq p',$  ${D}^h(p)=\{D'\}$ and ${D}^h(p')=\{D''\}$ we have $\{d\in D':p'_{d}=p_x\}\subseteq D''$.
%\newline\newline
%{\bf Weak Law of Aggregate Demand}:
%For $p,p'\in\mathbb{R}^{D}$ with $p\leq p',$  ${D}^h(p)=\{D'\}$ and ${D}^h(p')=\{D''\}$ we have $|D'|\geq|D''|$.
%\newline\newline

A {\bf matching} is a function $\mu:H\times D\rightarrow 2^D\cup H$ with $\mu(h)\subseteq D$ for each $h\in H$ and $\mu(d)\in H\cup\{\emptyset\}$ for each $d\in D$ such that $d\in\mu(h)$ if and only if $h= \mu(d).$ 
A {\bf competitive equilibrium} $(\mu,p)$ is a pair consisting of a matching $\mu$, and a price vector $p\in\mathbb{R}^{H\times D}$ such that for each $h\in H$ and $p_h:=(p_{hd})_{d\in D}$ we have $\mu(h)\in D^h(p_h)$ and for each $d\in D$ and $p_d=(p_{hd})_{h\in H}$ we have
$u^d(\mu(d),p_{d})=\max_{h\in H\cup\{\emptyset\}}u^d(h,p_d).$
The following is an immediate consequence of Theorems~\ref{Util},~\ref{extreme} and~\ref{GSP}.
\begin{corollary}
    For each matching market such that doctors are weak gross substitutes for hospitals and the law of aggregate demand holds the following is true:
\begin{enumerate}
        \item Let $p,p'\in\mathbb{R}^{H\times D}$ be equilibrium prices. Then  $\bar{p},\underline{p}\in\mathbb{R}^{H\times D}$ defined by
        $$\bar{p}_{hd}=\max\{p_{hd},p_{hd}'\},\quad \underline{p}_{hd}=\min\{p_{hd},p_{hd}'\}$$
        are equilibrium prices.
        \item 
        Let $p,p'\in\mathbb{R}^{H\times D}$ be equilibrium prices. For each matching $\mu$ supporting $p$ as an equilibrium $(\mu,p)$ there is a matching $\mu'$ supporting $p'$ as an equilibrium $(\mu',p')$ such that
        \begin{enumerate}
        	\item  a doctor is unemployed in $\mu$ if and only if he is unemployed in $\mu'$, i.e. $\mu(d)=\emptyset\Leftrightarrow\mu'(d)=\emptyset,$ for each $d\in D$,
        	\item each hospital hires the same number of doctors in $\mu$ and $\mu'$, i.e. $|\mu(h)|=|\mu'(h)|$ for each $h\in H.$ 
        \end{enumerate}
    \item If utility functions satisfy, moreover, BCV, then there exists a worker-optimal equilibrium allocation and a hospital-optimal equilibrium allocation.
    \item The worker-optimal mechanism is group-strategy-proof for workers on the domain of utility profiles such that workers' utility functions satisfy Unit Supply and BCV and hospitals' utility functions satisfy BCV, weak GS, and LAD.
        \end{enumerate}
\end{corollary}
\begin{proof}
We can construct a corresponding trading network with $\Omega=H\cup D$ and $\tilde{u}^h(\Psi,p)=u^h(\{d:(h,d)\in\Psi \},p)$ for $\Psi\subseteq \Omega_f$ and $\tilde{u}^d(\{(h,d)\},p_{hd})=u^d(h,p_{hd}),$ $\tilde{u}^d(\emptyset)=u^d(\emptyset)$ and $\tilde{u}^d(\Psi,\cdot)=-\infty$ if $\Psi\subseteq \Omega_f$ with $|\Psi|>1$. The weak gross substitutes condition then corresponds to the weak SSS condition (see Appendix~\ref{Full}) which by Proposition~\ref{B.1} in Appendix~\ref{Full} is equivalent to the SSS condition. Since the market is two-sided, SSS and FS are equivalent. The corollary follows from Theorems~\ref{Util},~\ref{extreme} and~\ref{GSP}.
\end{proof}
\subsection{Exchange economies with uniform pricing}\label{Exchange}
Next, we apply the model to the exchange of indivisible objects. The result extend results of~\cite{GulStacchetti1999} and~\cite{Hatfield2013} (see the discussion in their Section IV.B) to imperfectly transferable utility. As in~\cite{GulStacchetti1999}, we maintain the assumption that the market is cleared through transfers of a perfectly divisible good and there is no constraint on the amount of the divisible good an agent can consume. Moreover, negative quantities of the divisible good can be consumed. However, we do not assume that utility in the divisible good is quasi-linear. Similar assumption are standard in the object allocation literature with general preferences, see for example~\cite{Shige}.

In the following, we let $X$ be a finite set of heterogeneous indivisible {\bf objects}. From now on, we use the term agents in lieu of firms. Agents have utility functions over bundles of objects and transfers, $\tilde{u}^f:2^X\times\mathbb{R}\rightarrow \mathbb{R}$ such that for each $Y\subseteq X$, $\tilde{u}^f(Y,\cdot)$ is continuous, strictly increasing and has full range,\footnote{This assumption is only necessary for the existence of side-optimal allocations and otherwise redundant.} and for each $t\in \mathbb{R}$ and $Y\subseteq Y'\subseteq X$, we have $\tilde{u}^f(Y,t)\leq\tilde{u}^f(Y',t)$. Each agent $f$ is endowed with a bundle of objects $X_f\subseteq X$ such that $X_f\cap X_{f'}=\emptyset$ for $f\neq f'$ and $\bigcup_{f\in F}X_f=X$.
 An {\bf exchange economy} is a pair $(\tilde{u},(X_f)_{f\in F})$ of utility functions and endowments for each agent. We define for each $f\in F$ a {\bf demand} correspondence $\tilde{D}^f:\mathbb{R}_{+}^X\times2^X\rightrightarrows2^X$ by $$\tilde{D}^f(p,X_f):=\text{argmax}_{Y\subseteq X}\tilde{u}^f(Y,\sum_{x\in X_f\setminus Y}p_{x}-\sum_{x\in Y\setminus X_f}p_{x}).$$
     \begin{remark}
    In contrast to quasi-linear utility, the demand can differ with the endowment, i.e. in general $\tilde{D}^f(p,X_f)\neq \tilde{D}^f(p,\tilde{X}_f)$ for $X_f\neq \tilde{X}_f$.\qed
  \end{remark}
   We assume that objects are {gross substitutes} for agents. %Due to the possibility of indifferences and non-quasi-linear utility, we consider two notions of gross substitutability.
%    \newline\newline
%  \noindent
%  {\bf Weak Gross Substitutability}:
%   For $p,p'\in\mathbb{R}_{+}^{X}$ with $p\leq p',$  $\tilde{D}^f(p,X_f)=\{Y\}$ and $\tilde{D}^f(p',X_f)=\{Y'\}$ we have $\{x\in Y:p'_{x}=p_x\}\subseteq Y'$.
\newline\newline
  \noindent
{\bf Gross Substitutability}:
For $p,p'\in\mathbb{R}_{+}^{X}$ with $p\leq p'$, if $p'_x=p_x$ for $x\in X_f$, then for each  $Y'\in\tilde{D}^f(p',X_f)$ there exists a $Y\in\tilde{D}^f(p,X_f)$ such that $\{x\in Y:p'_{x}=p_x\}\subseteq Y'$, and if $p'_x=p_x$ for $x\in X\setminus X_f$, then for each  $Y\in\tilde{D}^f(p,X_f)$ there exists a $Y'\in\tilde{D}^f(p',X_f)$, such that $\{x\in Y:p'_{x}=p_x\}\subseteq Y'$.
\newline\newline
 Moreover, we assume the {law of aggregate demand}:
 \newline\newline
 \noindent
 {\bf Law of Aggregate Demand}:
 For $p,p'\in\mathbb{R}_{+}^{X}$ with $p\leq p'$, if $p'_x=p_x$ for $x\in X_f$, then for each  $Y'\in\tilde{D}^f(p',X_f)$ there exists a $Y\in\tilde{D}^f(p,X_f)$, and if $p'_x=p_x$ for $x\in X\setminus X_f$, then for each  $Y\in\tilde{D}^f(p,X_f)$ there exists a $Y'\in\tilde{D}^f(p',X_f)$, such that $|Y|\geq |Y'|$.
\begin{remark}
    We assume that there is only one copy of each object. Our results also hold true with multiple copies provided each agent can supply at most one unit of each indivisible object and wants to consume at most one unit of each indivisible object. We can deal with multi-unit supply of the same good by a seller and multi-unit demand of the same good by a buyer in our model, by creating identical copies of objects and price units individually. However, in this case, we will generally have non-linear pricing. For the case of multi-unit demand and supply with linear pricing, we would have to modify our model and the gross substitutability condition. For the quasi-linear case with multiple units of the same good, see~\cite{Baldwin2016} and in particular their ordinary substitutability condition.\qed
\end{remark}
      An {\bf allocation of objects} is a partition $Y=(Y_f)_{f\in F}$ with $Y_f\subseteq X$ and $Y_f\cap Y_{f'}=\emptyset$ for $f\neq f'$.
   A {\bf competitive equilibrium} of the exchange economy $(\tilde{u},(X_f)_{f\in F})$ is a pair $[Y,p]$ where $Y$ is an allocation of objects and $p\in \mathbb{R}_{+}^X$ such that for each $f\in F$ we have
$Y_f\in \tilde{D}^f(p,X_f).$

For each exchange economy $(\tilde{u},(X_f)_{f\in F})$, a corresponding trading network can be defined as follows: The set of trades is $$\Omega:=\{(x,f_1,f_2)\in X\times F\times F:x\in X_{f_1},f_2\neq f_1\}$$ where for $\omega=(x,f_1,f_2)\in \Omega$ we have $s(\omega)=f_1\neq f_2=b(\omega)$. We write $x(\omega)$ for the object involved in trade $\omega$. For $\Psi\subseteq\Omega_f$ and $p\in \mathbb{R}^{\Omega_f}$ define $$X_f(\Psi):=\{x(\omega):\omega\in \Psi_{\rightarrow f}\}\cup X_f\setminus\{x(\omega):\omega\in\Psi_{f\rightarrow} \},\quad p_f(\Psi):=\sum_{\omega\in\Psi_{f\rightarrow }}p_{\omega}-\sum_{\omega\in\Psi_{\rightarrow f}}p_{\omega}.$$
Utility functions are induced by utility functions over bundles of objects and transfers; for $\Psi\subseteq \Omega_f$ and $p\in\mathbb{R}_+^{\Omega_f}$ we let
$${ u^f(\Psi,p)=\begin{cases}\tilde{u}^f(X_f(\Psi),p_f(\Psi)),&\text{if }\{x(\omega):\omega\in \Psi_{f\rightarrow}\}\subseteq X_f\text{ and }x(\omega)\neq x(\omega')\\&\text{for }\omega,\omega'\in\Psi\text{ with }\omega\neq\omega',\\-\infty,&\text{else.}
\end{cases}}$$
To apply the results from the previous sections, we also extend the utility functions to negative prices; for $\Psi\subseteq\Omega_f$ and $p\in\mathbb{R}^{\Psi}\setminus\mathbb{R}^{\Psi}_+$ we let
$$u^f(\Psi,p):= u^f(\Psi,(\max\{p_{\omega},0\})_{\omega\in\Psi})+\sum_{\omega\in\Psi_{f\rightarrow }}\min\{p_{\omega},0\}-\sum_{\omega\in\Psi_{\rightarrow f}}\min\{p_{\omega},0\}.$$
\begin{remark}\label{aha}
    Extending utility for negative prices in this way implies (see the proof of Lemma~\ref{equivalencei})  that the induced demand $D^f$ satisfies FS on $\mathbb{R}^{\Omega_f}$ whenever it satisfies FS on $\mathbb{R}^{\Omega_f}_+$. Moreover, is easy to see that for each $\Psi\subseteq\Omega_f$ with $u^f(\Psi,\cdot)>-\infty$, $u^f$ is continuous (as $u^f$ is continuous on $\mathbb{R}_+^{\Psi}$, and $\min$ and $\max$ are continuous) and monotonic on $\mathbb{R}^{\Psi}$. This will allow us to apply the results from previous sections.

    Equilibrium prices in the trading network are non-negative by our assumption that $\tilde{u}^f(Y,t)\leq \tilde{u}^f(Y',t)$ for $t\in \mathbb{R}$ and $Y\subseteq Y'\subseteq X$: Let $p\in\mathbb{R}^{\Omega}$ and define $\Omega^-:=\{\omega\in\Omega:p_{\omega}<0\}$. First note that for $\Psi\in D^f(p)$ we have $\Psi\cap\Omega^-_{f\rightarrow}=\emptyset$: Define $p^+\in\mathbb{R}^{\Omega}$ by $p^+_{\omega}:=\max\{p_{\omega},0\}$ for $\omega\in\Omega_{f\rightarrow}$ and $p^+_{\omega}:=p_{\omega}$ else. 
     Note that  $X^f(\Psi)\subseteq X^f(\Psi\setminus\Omega^-_{f\rightarrow})$. Thus, by monotonicity $$u^f(\Psi,p)\leq u^f(\Psi,p^+)\leq u^f(\Psi\setminus\Omega^-_{f\rightarrow},p^+)= u^f(\Psi\setminus\Omega^-_{f\rightarrow},p).$$ As $\Psi\in D^f(p)$, all inequalities hold with equality, in particular, $u^f(\Psi,p)= u^f(\Psi,p^+)$ and therefore, by monotonicity,  $\Psi\cap\Omega^-_{f\rightarrow}=\emptyset.$      
     By a similar argument, if $\Psi\in D^f(p)$ and $\Omega^-_{\rightarrow f}\neq\emptyset$, then $\Psi\cap\Omega^-_{\rightarrow f}\neq \emptyset$. Thus, for $p\in\mathbb{R}^{\Omega}\setminus\mathbb{R}^{\Omega}_+$ there is excess demand and for each $p\in\mathcal{E}(u)$, we have $p_{\omega}\geq0$ for each $\omega\in\Omega$. \qed
\end{remark}

The gross substitutes condition for $\tilde{u}^f$ corresponds to the full substitutability condition for $u^f$ and the law of aggregate demand for $\tilde{u}^f$ implies the laws of aggregate demand and supply for $u^f$.
\begin{lemma}\label{equivalencei}
    If objects are gross substitutes under $\tilde{u}^f$, then trades are full substitutes under $u^f$. If the law of aggregate demand holds under $\tilde{u}^f$, then the laws of aggregate demand and supply hold under $u^f$.
\end{lemma}

In general, different trades involving the same object can be priced differently.
In the following, we call $p\in\mathcal{E}(u)$ a {\bf competitive equilibrium} of the trading network {\bf with uniform pricing}, if for $\omega,\omega'\in\Omega$, with $x(\omega)=x(\omega')$ we have $p_{\omega'}=p_{\omega}$. Trades in the same object are perfect substitutes to each other for the seller of the object, and he will sell the object to a buyer who is offering the highest price. Thus, we can always construct an equilibrium with uniform pricing from an equilibrium with non-uniform pricing by setting the price of the non-realized trades to the highest price for the involved object over all trades in the trading network. Similarly, a competitive equilibrium in the exchange economy, induces a competitive equilibrium with uniform pricing in the trading network. The following theorem can be interpreted as a generalization of Theorem~10 of~\cite{Hatfield2013}.
\begin{proposition}\label{equivalence}
    \begin{enumerate}
        \item If $p\in\mathbb{R}_{+}^{\Omega}$ are equilibrium prices in the trading network induced by an exchange economy, then $(\max_{\omega\in\Omega,x=x(\omega)}p_{\omega})_{x\in X}\in\mathbb{R}_{+}^X$ are equilibrium prices in the exchange economy.
        \item If $p\in\mathbb{R}_{+}^{X}$ are equilibrium prices in an exchange economy, then $(p_{x(\omega)})_{\omega\in\Omega}\in\mathbb{R}_{+}^{\Omega}$ are equilibrium prices in the trading network induced by the exchange economy.
    \end{enumerate}
    \end{proposition}
\begin{proof}
Let $[\Psi,p]$ be an equilibrium in the induced trading network. Let $q:=(\max_{\omega\in\Omega,x=x(\omega)}p_{\omega})_{x\in X}$ and consider the allocation $[(X_f(\Psi))_{f\in F},q]$ in the exchange economy.   By construction, we have $p_{\omega}\leq q_{x(\omega)}$ for each $\omega\notin \Psi$ and $p_{\omega}=q_{x(\omega)}$ for $\omega\in\Psi$. Thus $$ \Psi_f\in {D}^f(p)\Rightarrow X_f(\Psi)\in \tilde{D}^f(q,X_f)$$
    and $[(X_f(\Psi))_{f\in F},q]$ is an equilibrium of the exchange economy.

For the second part, let $[Y,p]$ be an equilibrium of the exchange economy. Define $ q:=(p_{x(\omega)})_{\omega\in\Omega}$ and consider the set of trades $\Psi\subseteq\Omega$ defined by $$\Psi:=\{\omega\in \Omega:x(\omega)\in Y_{b(\omega)}\cap X_{s(\omega)}\}.$$
    By construction, we have $$Y_f\in \tilde{D}^f(p,X_f)\Rightarrow \Psi_f\in D^f(q).$$
    Therefore $[\Psi,q]$ is an equilibrium of the induced trading network.
\end{proof}
%An immediate consequence is the existence of competitive equilibria if objects are gross substitutes for agents.
%\begin{corollary}\label{cor}
%    In each exchange economy such that objects are weak gross substitutes for agents, there exists a competitive equilibrium.
%    \end{corollary}
%\begin{proof}
%By Lemma~\ref{equivalencei}, demand in the trading network satisfied weak FS. Moreover, by the full range assumption, there exists a $K>0$ such that for each $f\in F$ and $Y\subseteq X$ we have $\tilde{u}^f(Y,-K)<\tilde{u}^f(X_f,0)$. Let $(\Psi,p)$ be an allocation in the induced trading network and let $p^0:=(\max\{p_{\omega},0\})_{\omega\in \Psi}$. If $u^f(\Psi,p)\geq u^f(\Psi,p^0)=\tilde{u}^f(X_f(\Psi),p^0_f(\Psi))\geq\tilde{u}^f(X_f,0)=u^f(\emptyset)$ for each $f\in F$, then $p_f(\Psi)\geq p_f^0(\Psi)>-K$ for each $f\in F$. Thus, the bounded compensating variation assumption of~\cite{Ravi} is satisfied. By Theorem~1 of~\cite{Ravi} there exist an equilibrium in the induced trading network. As seen in Remark~\ref{aha}, equilibrium prices are non-negative. By the first part of Theorem~\ref{equivalence} there exists an equilibrium in the exchange economy.
%\end{proof}
Proposition~\ref{equivalence} and the previous results for trading networks imply the following:
\begin{corollary}\label{general}
    Let $(\tilde{u},(X_f)_{f\in F})$ be an exchange economy such that objects are gross substitutes for agents and the law of aggregate demand holds.
    \begin{enumerate}
        \item {\bf Lattice Theorem}: Let $p,p'\in\mathbb{R}^X_+$ be equilibrium prices. Then the price vectors $\bar{p},\underline{p}\in\mathbb{R}_+^{X}$ defined by $$\bar{p}_{x}:=\max\{p_{x},p'_{x}\},\quad \underline{p}_{x}:=\min\{p_{x},p'_{x}\},$$
        are equilibrium equilibrium prices.
        \item {\bf Rural Hospitals Theorem}: Let $p,p'$ be equilibrium prices. For each equilibrium $[Y,p]$ there exists an assignment $Y'$ such that for each $f\in F$ $|Y_f|=|Y'_f|$, i.e. $f$ consumes the same number of objects in $Y$ and $Y'$.
\item{\bf Existence of Extremal equilibria}: There exist equilibrium price vectors $\bar{p},\underline{p}\in\mathbb{R}_+^X$, such that for each equilibrium price vector $p\in \mathbb{R}^X_+$ and $x\in X$ we have $$\underline{p}_x\leq p_x\leq \bar{p}_x.$$
    \end{enumerate}
\end{corollary}
\begin{remark}
    Throughout this section, we have made the assumption that utility depends on the total amount of the divisible good, but not on how transfers of the divisible good are obtained through different trades. For the induced trading network this means that utility satisfies the no frictions assumption. Frictions for individual trades in the trading network can lead to non-uniform pricing. Suppose for example that an agent is endowed with an object and faces different transactions costs depending on whom he is selling the object to. In this case, he might have an incentive to sell the object to a buyer who is offering a lower price, if transaction costs with this buyer are lower than with other buyers who offer a higher price. Thus, Proposition~\ref{equivalence} can fail to hold in the presence of frictions. A slightly more general version of the theorem can be obtained, where it is assumed that utility is symmetric in transfers from different trades with the same objects, but transfers from trades with different objects can enter the utility asymmetrically. In this case, trades in different objects can contain different frictions, however, trades of the same objects are perfect substitutes for each other.\qed
\end{remark}
\section{Conclusion}
In this paper, we have established several structural results for the set of equilibria in trading networks with frictions, making only minimal assumption on utility functions. Quasi-linearity can be replaced by the assumption of monotonicity and continuity to establish the lattice structure for the set of equilibria, the rural hospitals theorems, the existence of buyer-optimal and seller-optimal equilibria, and a group-incentive-compatibility result. Our results are applicable to a wide range of situations where the assumption of quasi-linearity is unrealistic because of imperfectly transferable utility or the presence of frictions.

 To prove the results we used a novel tie-breaking approach where for prices where demand is multi-valued we show that a well-behaved selection from the demand correspondence can be made. After ties are broken standard techniques can be applied to obtain the lattice theorem and the rural hospitals theorem.
This motivates us to pose two open questions for future research: First, can results for the trading networks model without transfers as in~\cite{Fleiner2016}, generalize from strict preferences to preferences with ties?\footnote{For marriage markets, the lattice result extends to the case with ties for strongly stable matchings~\cite[]{Manlove2002}.} A recent result by~\cite{Hatfield:2018} is very much in this spirit. The authors show that chain-stable and stable outcomes are equivalent for trading networks without making assumptions of monotonicity or continuity on utility functions. Hence their result applies to both trading networks with and without transfers.

Second, in the case of quasi-linear utility, corresponding results to ours can be obtained through solving a generalized network
flow problem~\cite[]{Candogan2016}. The optimal solutions to the flow
problem correspond to a competitive equilibrium outcome and its dual yields supporting prices. This generalizes the linear programming approach from one-to-one matching markets with quasi-linear utility
\cite[]{ShapleyShubik1971}. Recently,~\cite{Noldeke2018} have shown that a more general non-linear duality theory can be used to obtain versions of the results of~\cite{Demange.1985} that generalize~\cite{ShapleyShubik1971} beyond quasi-linear utility. It seems unlikely that the approach generalizes to our full model with frictions. However, for the case of no frictions (in which case utility only depends on the sum of transfers and competitive equilibria are efficient), it is a natural question whether the duality approach to equilibria in trading networks generalizes beyond quasi-linear utility to obtain similar results as in the present paper through different techniques.
\bibliographystyle{RM}
\bibliography{matching}
\appendix
\newpage
\label{App}
\counterwithin{proposition}{section}
\counterwithin{lemma}{section}
\counterwithin{corollary}{section}
\section{Different Versions of Full Substitutability}\label{Full}
\subsection*{Contraction Full Substitutability}
In Appendix~A of \cite{Hatfield2018}, the authors introduce the contraction and expansion version of full substitutability that differ in regard to how they are defined at price vectors where the demand is multi-valued. Note that full substitutability relates the demand at two price vectors $p$ and  $p'$. If $p_{\omega}\leq p'_{\omega}$ for $\omega\in\Omega_{f\rightarrow}$ and $p'_{\omega}=p_{\omega}$ for $\omega\in\Omega_{\rightarrow f}$ respectively if $p_{\omega}\geq p'_{\omega}$ for $\omega\in\Omega_{\rightarrow f}$ and $p'_{\omega}=p_{\omega}$ and $\omega\in\Omega_{f\rightarrow}$, then the choice set expands from $p'$ to $p$, and it contracts from $p$ to $p'$.  The contraction version of full substitutability requires that ``for all $\Psi\in D^f(p)$ there is a $\Psi'\in D^f(p')$ such that..." whereas the expansion version inverts the order of quantification and requires that ``for all $\Psi'\in D^f(p')$ there is a $\Psi\in D^f(p)$ such that..."
We further split full substituability into same-side substitutability (SSS) and cross-side complementarity (CSC) and the laws of aggregate demand and supply, all of which are implied by all versions of full substitutability in the quasi-linear case. We use the expansion version of SSS as our main definition. Alternatively, we can consider the contraction version.

\noindent\newline
{\bf Contraction Same-Side Substitutability}:
For $p,p'\in\mathbb{R}^{\Omega_f}$ and each $\Psi\in D^f(p)$ there exists a $\Psi'\in D^f(p')$ such that if $p_{\omega}= p'_{\omega}$ for $\omega\in\Omega_{f\rightarrow}$ and $p_{\omega}\leq p'_{\omega}$ for $\omega\in \Omega_{\rightarrow f} $, then $$\{\omega\in\Psi_{\rightarrow f}:p_{\omega}=p'_{\omega}\}\subseteq \Psi'_{\rightarrow f},$$ and if $p_{\omega}= p'_{\omega}$ for $\omega\in\Omega_{\rightarrow f}$ and $p_{\omega}\geq p'_{\omega}$ for $\omega\in \Omega_{f\rightarrow } $, then
$$\{\omega\in\Psi_{f\rightarrow }:p_{\omega}=p'_{\omega}\}\subseteq \Psi'_{f\rightarrow }.$$
We have previously introduced weak FS that can be further decomposed into weak SSS and weak CSC.

\noindent
\newline
{\bf Weak Same-Side Substitutability}:
For $p,p'\in\mathbb{R}^{\Omega_f}$ such that $ D^f(p)=\{\Psi\}$ and $D^f(p')=\{\Psi'\}$ if $p_{\omega}= p'_{\omega}$ for $\omega\in\Omega_{f\rightarrow}$ and $p_{\omega}\leq p'_{\omega}$ for $\omega\in \Omega_{\rightarrow f} $, then $$\{\omega\in\Psi_{\rightarrow f}:p_{\omega}=p'_{\omega}\}\subseteq \Psi'_{\rightarrow f},$$ and if $p_{\omega}= p'_{\omega}$ for $\omega\in\Omega_{\rightarrow f}$ and $p_{\omega}\geq p'_{\omega}$ for $\omega\in \Omega_{f\rightarrow } $, then
$$\{\omega\in\Psi_{f\rightarrow }:p_{\omega}=p'_{\omega}\}\subseteq \Psi'_{f\rightarrow }.$$
All three notions of SSS, the weak, the contraction version, as well as the expansion version that we use as our main definition are equivalent under quasi-linear utility~\cite[]{Hatfield2018}. For general utility functions, the weak version and the expansion version of SSS remain equivalent (by a result of~\citealp{Ravi} in their Appendix~A\footnote{More precisely, Contraction SSS is the conjunction of the two properties that \cite{Ravi} denote by ``decreasing-price full substitutability for sales (in demand language)" and ``increasing-price full substitutability for purchases in (demand language)".}) but the contraction version of SSS is strictly stronger than the expansion version. See  Example~\ref{Ex2} after the following proposition.
\begin{proposition}[\citealp{Hatfield2018,Ravi}]\label{B.1}
	Let $u^f$ be a monotonic and continuous utility function with induced demand $D^f$.
	\begin{enumerate}
		\item $D^f$ satisfies weak SSS if and only if it satisfies (Expansion) SSS.
		\item If $D^f$ satisfies Contraction SSS then it satisfies weak SSS.
		\item If $u^f$ is quasi-linear and $D^f$ satisfies weak SSS, then $D^f$ satisfies Contraction SSS.
	\end{enumerate}
\end{proposition}
%\begin{figure}
%	\centering
%	\includegraphics[width=0.7\textwidth]{ex3} % first figure itself
%	\caption{The demand in price space for $p_3=2$.}\label{Figure}
%\end{figure}

\begin{example}\label{Ex2}
	Consider three trades $\Omega=\{\omega_1,\omega_2,\omega_3\}$ with $f=b(\omega_1)=b(\omega_2)=b(\omega_3)$.
	We let $u^{f}(\emptyset)=0$, $u^{f}(\{\omega_i\},p_{\omega_i})=3-p_{\omega_i}$ for $i=1,2,3$, $u^f(\{\omega_i,\omega_j\},p_{\omega_i},p_{\omega_j})=4-p_{\omega_i}-p_{\omega_j}$ for $i\neq j$ and  $$u^{f}(\{\omega_1,\omega_2,\omega_3\},p)=\begin{cases}
	4-p_{\omega_1}-p_{\omega_2}-p_{\omega_3}\quad&\text{ if }p_{\omega_1}+p_{\omega_2}+p_{\omega_3}\leq0\\
	4-3\sqrt{\frac{p_{\omega_1}+p_{\omega_2}+p_{\omega_3}}{6}}\quad&\text{ if }6\geq p_{\omega_1}+p_{\omega_2}+p_{\omega_3}>0,\\7-p_{\omega_1}-p_{\omega_2}-p_{\omega_3}\quad&\text{ else}
	\end{cases}$$
	%	See Figure~\ref{Figure} for a geometric representation of the demand in the example.
	Observe that $$D^f(2,2,2)=\{\{\omega_1,\omega_2,\omega_3\},\{\omega_1\},\{\omega_2\},\{\omega_3\}\},$$
	but $$D^f(3,2,2)=\{\{\omega_2\},\{\omega_3\}\}.$$
	As $\{\omega_1,\omega_2,\omega_3\}\in D^f(2,2,2)$, Contraction SSS would require that there is a $\Psi\in D^f(3,2,2)$ with $\{\omega_2,\omega_3\}\subseteq \Psi$.  Hence Contraction SSS is not satisfied.  As the demand at $(2,2,2)$ and $(3,2,2)$ is multi-valued, Weak SSS does not impose any structure here. More generally, note that if we replace $u^f$ by the quasi-linear utility functions $\tilde{u}^f$ such that $\tilde{u}^f(\{\omega_1,\omega_2,\omega_3\},p)=4-p_{\omega_1}-p_{\omega_2}-p_{\omega_3}$ for all $p\in\mathbb{R}^{\Omega_f}$ and $u^f$ remains otherwise unchanged, only the demand at prices $(2,2,2)$ changes. One readily checks that $\tilde{u}^f$ satisfies (Weak) SSS. Hence $u^f$ satisfies Weak SSS.
	\qed
\end{example}
Similarly as for Same-Side Substitutabilty, we can define an alternative version of Cross-Side Complementarity.

\noindent\newline {\bf Contraction Cross-Side Complementarity}: For each $\Psi\in D^f(p)$ there exists a $\Psi'\in D^f(p')$ such if  $p_{\omega}= p'_{\omega}$ for $\omega\in\Omega_{f\rightarrow}$ and $p_{\omega}\leq p'_{\omega}$ for $\omega\in \Omega_{\rightarrow f} $, then $$\Psi'_{f\rightarrow}\subseteq\Psi_{f\rightarrow},$$
and if $p_{\omega}= p'_{\omega}$ for $\omega\in\Omega_{\rightarrow f}$ and $p_{\omega}\geq p'_{\omega}$ for $\omega\in \Omega_{f\rightarrow } $, then $$\Psi'_{\rightarrow f}\subseteq\Psi_{\rightarrow f}.$$
As before, we also define a weak version of CSC that together with Weak SSS defines Weak FS.

\noindent\newline {\bf Weak Cross-Side Complementarity}: For each $p,p'\in\mathbb{R}^{\Omega_f}$ with $ D^f(p)=\{\Psi\}$ and $D^f(p')=\{\Psi'\}$ such if  $p_{\omega}= p'_{\omega}$ for $\omega\in\Omega_{f\rightarrow}$ and $p_{\omega}\leq p'_{\omega}$ for $\omega\in \Omega_{\rightarrow f} $, then $$\Psi'_{f\rightarrow}\subseteq\Psi_{f\rightarrow},$$
and if $p_{\omega}= p'_{\omega}$ for $\omega\in\Omega_{\rightarrow f}$ and $p_{\omega}\geq p'_{\omega}$ for $\omega\in \Omega_{f\rightarrow } $, then $$\Psi'_{\rightarrow f}\subseteq\Psi_{\rightarrow f}.$$
Similarly as for SSS, the weak, the contraction version as well as the expansion version of CSC that we use as our main definition are equivalent under quasi-linear utility~\cite[]{Hatfield2018}.
For general utility functions, the weak version and the contraction version of CSC remain equivalent (by a result of~\citealp{Ravi} in their Appendix~A\footnote{More precisely, Contraction SSS is the conjunction of the two properties that \cite{Ravi} denote ``increasing-price full substitutability for sales (in demand language)" and ``decreasing-price full substitutability for purchases (in demand language)".}) but the contraction version of SSS is strictly stronger than the expansion version (see Example~\ref{ex} in the main text).
\begin{proposition}[\citealp{Hatfield2018,Ravi}]\label{B.2}
	Let $u^f$ be a monotonic and continuous utility function with induced demand $D^f$.
	\begin{enumerate}
		\item $D^f$ satisfies weak CSC if and only if it satisfies Contraction CSC.
		\item If $D^f$ satisfies (Expansion) CSC then it satisfies Weak CSC.
		\item If $u^f$ is quasi-linear and $D^f$ satisfies Weak CSC, then $D^f$ satisfies (Expansion) CSC.
	\end{enumerate}
\end{proposition}

\subsubsection*{Monotone Substitutability, No Isolated Bundles, and Single Improvements}
Following the terminology of~\cite{Hatfield:2018}, we denote by monotone substitutability, the property that the FS, LAD, and LAS hold jointly for the same bundles of trades.\footnote{Note that we use the (weaker) demand-language version rather than the choice-language definition of the property used by~\cite{Hatfield:2018}.} \newline\newline
\noindent  {\bf Monotone Substitutability}: For $p,p'\in\mathbb{R}^{\Omega_f}$ and each $\Psi'\in D^f(p')$ there exists a $\Psi\in D^f(p)$ such that if $p_{\omega}= p'_{\omega}$ for $\omega\in\Omega_{f\rightarrow}$ and $p_{\omega}\leq p'_{\omega}$ for $\omega\in \Omega_{\rightarrow f} $, then $$\{\omega\in\Psi_{\rightarrow f}:p_{\omega}=p'_{\omega}\}\subseteq \Psi'_{\rightarrow f},\quad \Psi'_{f\rightarrow}\subseteq\Psi_{f\rightarrow},\quad|\Psi_{\rightarrow f}|- |\Psi_{f\rightarrow}|\geq |\Psi'_{\rightarrow f}|- |\Psi'_{f\rightarrow}|,$$
and if $p_{\omega}= p'_{\omega}$ for $\omega\in\Omega_{\rightarrow f}$ and $p_{\omega}\geq p'_{\omega}$ for $\omega\in \Omega_{f\rightarrow } $, then
$$\{\omega\in\Psi_{f\rightarrow }:p_{\omega}=p'_{\omega}\}\subseteq \Psi'_{f\rightarrow },\quad\Psi'_{\rightarrow f}\subseteq\Psi_{\rightarrow f},\quad|\Psi_{f\rightarrow}|- |\Psi_{\rightarrow f}|\geq |\Psi'_{f\rightarrow}|- |\Psi'_{\rightarrow f}|.$$\newline
If there are no isolated bundles, weak FS and FS are equivalent, and the combination of weak FS, weak LAD and weak LAS is equivalent to Monotonone Substitutability. %We show this result in a slightly more general version that we will need later on where we do not necessarily assume that the demand is induced by a monotonic and continuous utility function.
\begin{proposition}\label{single}
	Let $u^f$ be a continuous and monotonic utility function inducing a demand $D^f$.
	\begin{enumerate}
		\item If $D^f$ satisfies NIB and weak FS then it satisfies FS.
		\item If $D^f$ satisfies NIB, weak FS, weak LAD and weak LAS, then it satisfies Monotone Substitutability.
	\end{enumerate}
\end{proposition}
\begin{proof}
	Let $p,p'\in\mathbb{R}^{\Omega_f}$ such that $p_{\omega}= p'_{\omega}$ for $\omega\in\Omega_{f\rightarrow}$ and $p_{\omega}\leq p'_{\omega}$ for $\omega\in \Omega_{\rightarrow f} $. Let $\Psi'\in D^f(p')$.
	By upper hemi-continuity there exists an $\epsilon>0$ such that for $\|p-q\|<\epsilon$ we have $D^f(q)\subseteq D^f(p)$. By NIB, there is a $q'$ with $\|q'-p'\|<\epsilon/2$ and $D^f(q')=\{\Psi'\}$. Let $q:=p+q'-p'$. By construction $\|q-p\|=\|q'-p'\|<\epsilon/2<\epsilon$ and therefore $D^f(q)\subseteq D^f(p).$ By upper hemi-continuity there exists an $\epsilon'>0$ such that for $r'$ with $\|r'-q'\|<\epsilon'$ we have $D^f(r')=\{\Psi'\}=D^f(q')$.
	 We may choose $\epsilon'<\epsilon/2$. By the second part of Lemma~\ref{Fede}, there exists a $\tilde{p}\in\mathbb{R}^{\Omega_f}$ with $\|\tilde{p}-q\|<\epsilon'$ such that demand is single-valued,  $D^f(\tilde{p})=\{\Psi\}$ for a $\Psi\subseteq \Omega_f$. As, $\|\tilde{p}-p\|\leq\|\tilde{p}-q\|+\|p-q\| <\epsilon'+\epsilon/2<\epsilon$, we have $\Psi\in D^f(p)$. Let $\tilde{p}':=q'+\tilde{p}-q$. As $\|\tilde{p}'-q'\|=\|\tilde{p}-q\|<\epsilon'$, we have $D^f(\tilde{p}')=\{\Psi'\}$. By construction, we have $\tilde{p}=p+(q'-q)+(\tilde{p}-q)$ and $\tilde{p}'=p'+(q'-q)+(\tilde{p}-q)$. Since  $p_{\omega}= p'_{\omega}$ for $\omega\in\Omega_{f\rightarrow}$ and $p_{\omega}\leq p'_{\omega}$ for $\omega\in \Omega_{\rightarrow f} $, this implies  $\tilde{p}_{\omega}= \tilde{p}'_{\omega}$ for $\omega\in\Omega_{f\rightarrow}$ and $\tilde{p}_{\omega}\leq \tilde{p}'_{\omega}$ for $\omega\in \Omega_{\rightarrow f} $. By weak FS applied to the vectors $\tilde{p}$ and $\tilde{p}'$ and the fact that demand at both price vectors is single-valued with $D^f(\tilde{p})=\{\Psi\}$ and $D^f(\tilde{p}')=\{\Psi'\}$ we obtain
	$$\{\omega\in\Psi_{\rightarrow f}:p_{\omega}=p'_{\omega}\}\subseteq \Psi'_{\rightarrow f},\quad \Psi'_{f\rightarrow}\subseteq\Psi_{f\rightarrow}.$$
	If, moreover, weak LAD holds, then $$|\Psi_{\rightarrow f}|- |\Psi_{f\rightarrow}|\geq |\Psi'_{\rightarrow f}|- |\Psi'_{f\rightarrow}|.$$  A completely analogous argument shows the second part of the FS, resp.~of the Monotone Substitutability property.
\end{proof}
The converse of the first part of the proposition is not true, as the following example shows.
\begin{example}\label{ex3}
	Consider two trades $\Omega=\{\omega_1,\omega_2\}$ with $f=b(\omega_1)=b(\omega_2)$.
We let $u^{f}(\emptyset)=0$, $u^{f}(\{\omega_i\},p_{\omega_i})=3-p_{\omega_i}$ for $i=1,2$, and $$u^{f}(\{\omega_1,\omega_2\},p)=\begin{cases}
4-p_{\omega_1}-p_{\omega_2}\quad&\text{ if }p_{\omega_1}+p_{\omega_2}\leq2,\\
2-\frac{(p_{\omega_1}+p_{\omega_2})^2-4}{12}\quad&\text{ if }4\geq p_{\omega_1}+p_{\omega_2}>2,\\5-p_{\omega_1}-p_{\omega_2}\quad&\text{ else.}
\end{cases}$$
	See Figure~\ref{ex1} for a geometric representation of the demand in the example.
  \begin{figure}
	\centering
	\includegraphics[width=0.5\textwidth]{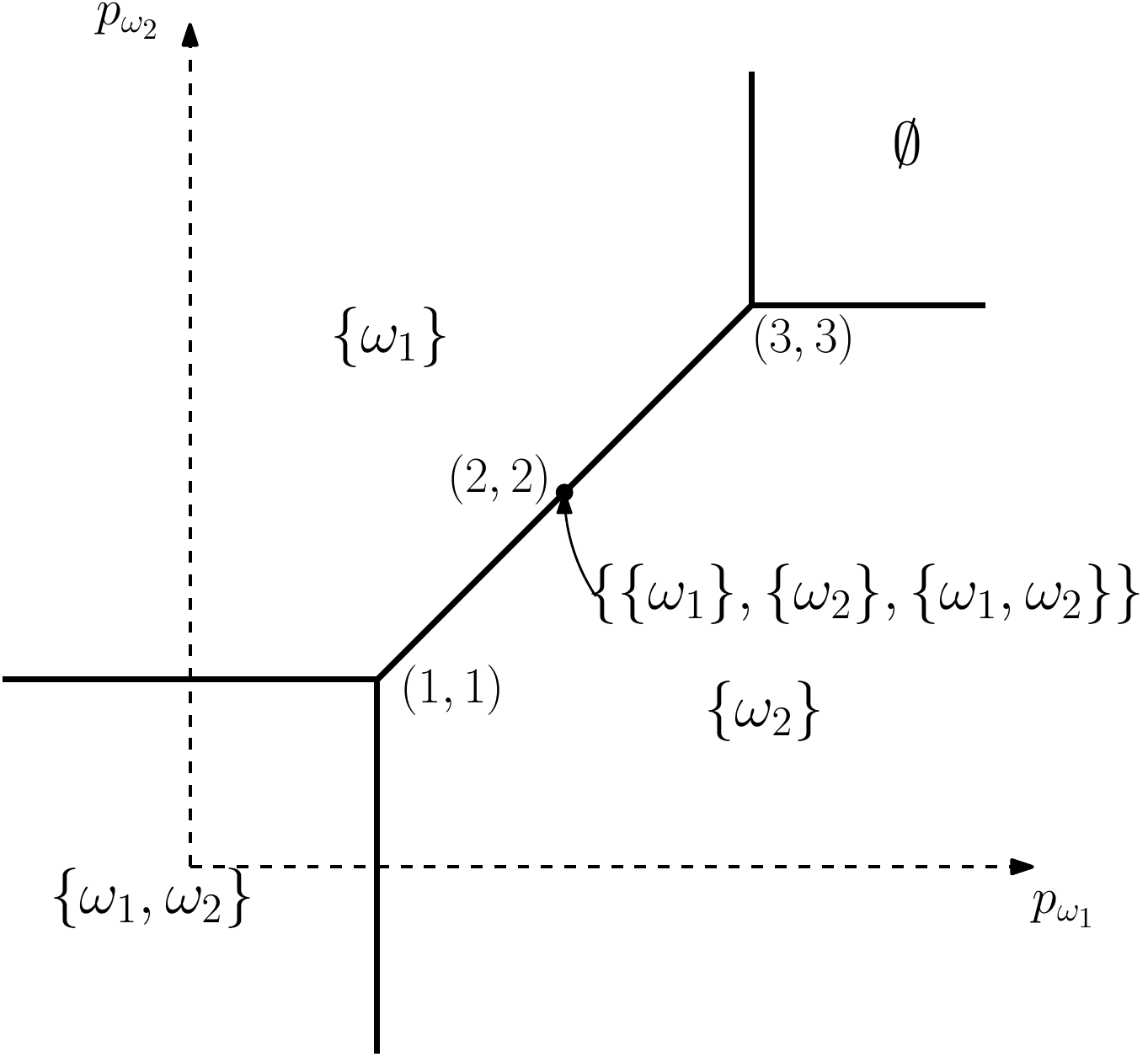} % first figure itself
	\caption{The demand of $f$ in price space  is represented by black lines. At prices $(2,2)$ the bundle $\{\omega_1,\omega_2\}$ is isolated.}\label{ex1}
\end{figure}
The induced demand violates NIB, since $\{\omega_1,\omega_2\}\in D^f(2,2)$ but $\{\omega_1,\omega_2\}\notin D^f(p)$ for $p\neq(2,2)$ with $p_{\omega_1}+p_{\omega_2}>1.$ Moreover, it violates LAD at $p'=(2,2)$ and, e.g.,~$p=(1,2)$ since $\Psi'=\{\omega_1,\omega_2\}\in D^f(p')$ but $D^f(p)=\{\{\omega_1\}\}$. 
One readily checks that $\tilde{u}^f$ satisfies FS and weak LAD.
\qed
\end{example}
The converse of the second part of Proposition~\ref{single} is true, and more generally the following holds: 
\begin{proposition}\label{NIB}
	Let $u^f$ be a continuous and monotonic utility function inducing a demand 
	$D^f$ satisfying CSC, LAD and LAS. Then $D^f$ satisfies NIB. 
\end{proposition}
\begin{proof}
Let $p\in\mathbb{R}^{\Omega_f}$ and $\Psi\in D^f(p)$. We show that for each $\epsilon>0$ there is a $q\in\mathbb{R}^{\Omega_f}$ with $\|p-q\|<\epsilon$ such that $D^f(q)=\{\Psi\}$. Let $\epsilon>0.$
First, consider a vector $\tilde{\epsilon}\in\mathbb{R}^{\Omega_f}$	with $\|\tilde{\epsilon}\|<\epsilon$ such that $\tilde{\epsilon}_{\omega}>0$ for $\omega\in\Omega_{\rightarrow f}\setminus\Psi$, $\tilde{\epsilon}_{\omega}<0$ for $\omega\in\Omega_{f\rightarrow}\setminus\Psi$, and $\tilde{\epsilon}_{\omega}=0$ for $\omega\in\Psi.$ By monotonicity of $u^f$, for each $\Xi\subseteq \Omega_f$ with $\Xi\nsubseteq \Psi$ we have $u^f(\Xi,p+\tilde{\epsilon})<u^f(\Xi,p)$, and we have $u^f(\Psi,p+\tilde{\epsilon})=u^f(\Psi,p)$.
Thus,  $D^f(p+\tilde{\epsilon})\subseteq 2^{\Psi}$ and $\Psi\in D^f(p+\tilde{\epsilon})$.
By upper hemi-continuity, there is a $\epsilon'>0$ such that for $q\in\mathbb{R}^{\Omega_f}$ with $\|q-(p+\tilde{\epsilon})\|<\epsilon'$ we have $D^f(q)\subseteq D^f(p+\tilde{\epsilon}).$ We may choose $\epsilon'<\epsilon-\|\tilde{\epsilon}\|.$ By the second part of Lemma~\ref{Fede}, there is a $q\in\mathbb{R}^{\Omega_f}$ with $\|q-(p+\epsilon)\|<\epsilon'$ such that demand is single-valued, $|D^f(q)|=1$, and we may choose it such that $q_{\omega}\leq p_{\omega}+\tilde{\epsilon}_{\omega}$ for $\omega\in\Omega_{\rightarrow f}$ and $q_{\omega}\geq p_{\omega}+\tilde{\epsilon}_{\omega}$ for $\omega\in\Omega_{f\rightarrow}.$ 
We show that for the unique $\Xi\subseteq \Omega_f$ with $D^f(q)=\{\Xi\}$ we have $\Xi_{f\rightarrow}=\Psi_{f\rightarrow}$. An analogous argument shows that $\Xi_{\rightarrow f}=\Psi_{\rightarrow f}$.

Let $r\in\mathbb{R}^{\Omega_f}$  with $r_{\omega}=q_{\omega}$ for $\omega\in\Omega_{\rightarrow f}$ and $r_{\omega}= p_{\omega}+\tilde{\epsilon}_{\omega}$ for $\omega\in\Omega_{f\rightarrow}$.  By the first part of the CSC condition applied to vectors $p+\tilde{\epsilon}$ (in the role of $p'$) and $r$ (in the role of $p$), there is a $\Phi\in D^f(r)$ such that $\Psi_{f\rightarrow}\subseteq\Phi_{f\rightarrow}$. Since $\|r-(p+\tilde{\epsilon})\|\leq \|q-(p+\tilde{\epsilon})\|<\epsilon'$, we have $D^f(r)\subseteq D^f(p+\tilde{\epsilon})\subseteq 2^\Psi$. Thus, $\Phi\subseteq\Psi$ and, by the previous observation that $\Psi_{f\rightarrow}\subseteq\Phi_{f\rightarrow}$, we have $\Phi_{f\rightarrow}=\Psi_{f\rightarrow}.$	By LAS applied to prices $r$ (in the role of $p'$) and $q$ (in the role of $p$) we have \begin{align}|\Xi_{f\rightarrow}|-|\Xi_{\rightarrow f}|\geq |\Phi_{f\rightarrow}|-|\Phi_{\rightarrow f}|.\label{ineq}\end{align}
	 Since $\|q-(p+\tilde{\epsilon})\|<\epsilon'$, we have $D^f(q)\subseteq D^f(p+\tilde{\epsilon})\subseteq 2^\Psi$. Thus, $\Xi\subseteq\Psi$ and, by the previous observation that $\Psi_{f\rightarrow}=\Phi_{f\rightarrow}$, we have $|\Xi_{f\rightarrow}|\leq|\Psi_{f\rightarrow}|=|\Phi_{f\rightarrow}|.$ Together with Inequality~(\ref{ineq}) this implies $|\Xi_{\rightarrow f}|\leq |\Phi_{\rightarrow f}|$. By the second part of the CSC condition applied to prices $r$ (in the role of $p'$) and $q$ (in the role of $p$) we have $\Phi_{\rightarrow f}\subseteq \Xi_{\rightarrow f}.$ Together with the previous inequality this implies $\Phi_{\rightarrow f}= \Xi_{\rightarrow f}.$  Furthermore, together with Inequality~(\ref{ineq}), this implies $|\Xi_{f\rightarrow}|\geq |\Phi_{f\rightarrow}|,$ and, as $\Phi_{f\rightarrow}=\Psi_{f\rightarrow}$, we have $|\Xi_{f\rightarrow}|\geq |\Psi_{f\rightarrow}|.$ As observed previously, $\Xi\subseteq \Psi$. Together with the previous inequality this implies $\Xi_{f\rightarrow}=\Psi_{f\rightarrow}.$

\end{proof}
\cite{Hatfield:2018} do not assume continuity of utility functions for their main result, and in that case, monotone substitutability is generally a stronger property than the combination of FS, LAD and LAS. With continuity, however, it follows as a corollary from the previous propositions that the conditions are equivalent.
\begin{corollary}\label{MON}
	Let $u^f$ be a continuous and monotonic utility function inducing a demand 
	$D^f$ satisfying FS, LAD and LAS. Then $D^f$ satisfies FS, LAD, and LAS if and only if $D^f$ satisfies Monotone Substitutability.
\end{corollary}
\begin{proof}
	Let $D^f$ satisfy FS, LAD and LAS. By Proposition~\ref{NIB}, $D^f$ satisfies NIB and therefore, by the second part of Proposition~\ref{single}, it satisfies Monotone Substitutability. The other direction follows from the definitions.
\end{proof}
In the context of object allocation with quasi-linear preferences,~\cite{GulStacchetti1999} show that gross substitutability is equivalent to a "single improvement property": If the price of an object increases, and previously a bundle $A$ of objects has been demanded, while now a bundle $B$ of objects is demanded, then $|B\setminus A|\leq 1$ and $|A\setminus B|\leq 1$, i.e.~at most one new object is added to the demand and at most one object is no longer demanded. We can prove a (partial) generalization of the result to general utility functions. The proof uses the following lemma which will also be useful subsequently.
\begin{lemma}\label{DemandInv}
	Let $u^f$ be a continuous and monotonic utility function inducing a demand correspondence
	$D^f$.
	For each $p,p'\in\mathbb{R}^{\Omega_f}$
	and $\Psi'\subseteq\Omega_f$ with $p_{\omega}\leq p'_{\omega}$ for
	$\omega\in\Omega_{f\rightarrow}\setminus\Psi'_{f\rightarrow}$,
	$p_{\omega}\geq p'_{\omega}$ for $\omega\in \Psi'_{f\rightarrow}$,
	$p_{\omega}\geq p'_{\omega}$ for $\omega\in\Omega_{\rightarrow
		f}\setminus\Psi'_{\rightarrow f}$ and $p_{\omega}\leq p'_{\omega}$
	for $\omega\in \Psi'_{\rightarrow f}$. 
	\begin{enumerate}
		\item If $D^f$ satisfies weak FS, weak LAD and weak LAS, then $D^f(p')=\{\Psi'\}$ implies $\Psi'\in D^f(p)$.
		\item If $D^f$ satisfies FS, LAD and LAS, then $\Psi'\in D^f(p')$ implies $\Psi'\in D^f(p)$.
	\end{enumerate}
\end{lemma}
\begin{proof}
	We first prove the first part. Let $D^f$ satisfy weak FS, weak LAD and weak LAS and $D^f(p')=\{\Psi'\}$. By monotonicity of $u^f$ it is without loss of generality to assume that $p'_{\omega}=p_{\omega}$ for $\omega\in\Omega_f\setminus\Psi'$. 
	By upper hemi-continuity of $D^f$, it suffices to show that for each $\epsilon>0$ there is a $q\in\mathbb{R}^{\Omega_f}$ with $\|p-q\|<\epsilon$ such that $\Psi'\in D^f(q).$ 
	By upper hemi-continuity of $D^f$ there is a $\epsilon'>0$ such that for $q'\in\mathbb{R}^{\Omega_f}$ with  $\|p'-q'\|<\epsilon'$ we have $D^f(q')=\{\Psi'\}=D^f(p')$.
	Define $\tilde{p}\in\mathbb{R}^{\Omega_f}$ such that $$\tilde{p}_{\omega}:=\begin{cases}
	p'_{\omega},\quad\text{ if }\omega\in\Omega_{f\rightarrow}\\p_{\omega},\quad\text{ if }\omega\in\Omega_{\rightarrow f}.
	\end{cases}$$
	By the second part of Lemma~\ref{Fede}, there is a $r\in\mathbb{R}^{\Omega_f}$ with $\|r-\tilde{p}\|<\min\{\tfrac{\epsilon}{2},\tfrac{\epsilon'}{2}\}$	 
	and a $\tilde{\Psi}\subseteq \Omega_f$ such that $D^f(r)=\{\tilde{\Psi}\}$. By upper hemi-continuity of $D^f$ there is a $\tilde{\epsilon}>0$ such that for $\|\tilde{q}-r\|<\tilde{\epsilon}$ we have $D^f(\tilde{q})=\{\tilde{\Psi}\}=D^f(r).$ We may choose $\tilde{\epsilon}<\min\{\tfrac{\epsilon}{2},\tfrac{\epsilon'}{2}\}$.	  
	By the second part of Lemma~\ref{Fede}, there is a $q\in\mathbb{R}^{\Omega_f}$ with $\|(p+(r-\tilde{p}))-q\|<\tilde{\epsilon}$ and a $\Psi\subseteq \Omega_f$ such that $D^f(q)=\{\Psi\}$.	Let $\tilde{q}:=\tilde{p}+q-p$ and $q':=p'+q-p$.  By construction we have $\|q'-p'\|=\|q-p\|\leq\|(p+(r-\tilde{p}))-q\|+\|r-\tilde{p}\|<\tilde{\epsilon}+\min\{\tfrac{\epsilon}{2},\tfrac{\epsilon'}{2}\}<\min\{\epsilon,\epsilon'\}$. Thus $D^f(q')=\{\Psi'\}$.
	By construction, we have 
	$\|r-\tilde{q}\|=\|r-(\tilde{p}+q-p)\|<\tilde{\epsilon}$ and therefore $D^f(\tilde{q})=\{\tilde{\Psi}\}=D^f(r).$ Applying the weak FS, condition to vector $\tilde{q}$ and $q'$ (note that by construction we have $q'_{\omega}=\tilde{q}_{\omega}$ for $\omega\in\Omega_{\rightarrow f}$) we have $\tilde{\Psi}_{f \rightarrow}\subseteq \Psi'_{\rightarrow f}$ and $\Psi'_{f\rightarrow}\subseteq \tilde{\Psi}_{f\rightarrow }$. Applying the weak LAD this implies $\tilde{\Psi}=\Psi'$.
	Applying the weak FS, condition to vector $q$ and $\tilde{q}$ (note that by construction we have $q_{\omega}=\tilde{q}_{\omega}$ for $\omega\in\Omega_{\rightarrow f}$) we have ${\Psi}_{f \rightarrow}\subseteq \tilde{\Psi}_{f\rightarrow }$ and $\tilde{\Psi}_{\rightarrow f}\subseteq {\Psi}_{\rightarrow f}$. Applying the weak LAS this implies $\tilde{\Psi}=\Psi'$. Thus $\Psi'=\tilde{\Psi}=\Psi\in D^f(q)$ as desired.

	The second part of the lemma follows from the first as follows:
	By Proposition~\ref{NIB}, there is for each $\epsilon>0$ a  $q'\in\mathbb{R}^{\Omega_f}$ with $\|p'-q'\|<\epsilon$ such that $D^f(q')=\{\Psi'\}$. By the first part of the lemma applied to vectors $q'$ and $q:=p+q'-p'$ we have $\Psi'\in D^f(q)$. Thus for each $\epsilon>0$ there is a $q\in\mathbb{R}^{\Omega_f}$ with $\|p-q\|<\epsilon$ and $\Psi'\in D^f(q)$. Thus, by upper hemi-continuity, $\Psi'\in D^f(p).$

\end{proof}
\begin{proposition}[Single Improvement Property]\label{SingleImprovement}
Let $u^f$ be a continuous and monotonic utility function inducing a demand
$D^f$ satisfying FS, LAD and LAS. Let  $p,p'\in\mathbb{R}^{\Omega_f}$ such that there is a  $\omega\in\Omega_{\rightarrow f}$ with $p'_{-\omega}=p_{-\omega}$ and  $p_{\omega}<p_{\omega}'$, or a  $\omega\in\Omega_{f\rightarrow}$ with $p'_{-\omega}=p_{-\omega}$ and $p_{\omega}>p_{\omega}'$. Then for each $\Psi'\in D^f(p')$ there exists a $\Psi\in D^f(p)$ such that 
	$$|\Psi_{\rightarrow f}\setminus\Psi'_{\rightarrow f}|+|\Psi'_{f \rightarrow}\setminus\Psi_{f\rightarrow}|\leq 1,\,\text{ and }\,\, |\Psi'_{\rightarrow f}\setminus\Psi_{\rightarrow f}|+|\Psi_{f \rightarrow}\setminus\Psi'_{f\rightarrow}|\leq 1.$$
\end{proposition}
\begin{proof}
	We show the result for $p_{\omega}<p'_{\omega}$ and $\omega\in\Omega_{\rightarrow f}$. A dual argument establishes the result for $p_{\omega}>p'_{\omega}$ and $\omega\in\Omega_{f\rightarrow}$. Let $\Psi'\in D^f(p')$. If $\omega\in\Psi'$, then by Lemma~\ref{DemandInv} we may choose $\Psi'=\Psi$. If $\omega\notin \Psi'$, then by monotone substitutability (which holds by Corollary~\ref{MON}), there is a $\Psi\in D^f(p)$ such that $\Psi_{\rightarrow f}\setminus\{\omega\}\subseteq \Psi'_{\rightarrow f}$, $\Psi'_{f\rightarrow}\subseteq \Psi_{f\rightarrow}$,
	and $|\Psi_{\rightarrow f}|- |\Psi_{f\rightarrow}|\geq |\Psi'_{\rightarrow f}|- |\Psi'_{f\rightarrow}|$. The inequality can be rearranged to $|\Psi_{\rightarrow f}|- |\Psi'_{\rightarrow f}|\geq |\Psi_{f\rightarrow }|-|\Psi'_{f\rightarrow}|. $ The condition that $\Psi_{\rightarrow f}\setminus\{\omega\}\subseteq \Psi'_{\rightarrow f}$ implies $|\Psi_{\rightarrow f}|- |\Psi'_{\rightarrow f}|\leq 1$, whereas the condition  $\Psi'_{f\rightarrow}\subseteq \Psi_{f\rightarrow}$ implies $|\Psi_{f\rightarrow }|-|\Psi'_{f\rightarrow}|\geq 0$. Thus
	$$1\geq|\Psi_{\rightarrow f}|- |\Psi_{\rightarrow f}'|\geq |\Psi_{f\rightarrow }|- |\Psi'_{f\rightarrow}|\geq0.$$
	Since $\Psi_{\rightarrow f}\setminus\{\omega\}\subseteq \Psi'_{\rightarrow f}$, if $|\Psi_{\rightarrow f}|- |\Psi_{\rightarrow f}'|=1$, then $\Psi_{\rightarrow f}=\Psi'_{\rightarrow f}\cup\{\omega\}$, and if $|\Psi_{\rightarrow f}|- |\Psi_{\rightarrow f}'|=0$, then $\Psi_{\rightarrow f}=\Psi_{\rightarrow f}'$ or $\Psi_{\rightarrow f}=\Psi_{\rightarrow f}'\setminus\{\omega\}\cup\{\omega'\}$ for $\omega'\in \Psi'_{\rightarrow f}$.
	Since $\Psi'_{f\rightarrow}\subseteq \Psi_{f\rightarrow}$, if $|\Psi_{f\rightarrow}|- |\Psi_{f\rightarrow }'|=1$, then $\Psi_{f\rightarrow}=\Psi'_{f\rightarrow}\cup\{\omega'\}$ for an $\omega'\in\Omega_{f\rightarrow}\setminus\Psi'$, and if $|\Psi_{\rightarrow f}|- |\Psi_{\rightarrow f}'|=0$, then $\Psi_{f\rightarrow }=\Psi_{f\rightarrow}'$. In conclusion, if
	$|\Psi_{\rightarrow f}|- |\Psi_{\rightarrow f}'|= |\Psi_{f\rightarrow }|- |\Psi'_{f\rightarrow}|=1$, then $\Psi=\Psi'\cup\{\omega,\omega'\}$ for an $\omega'\in\Omega_{f\rightarrow }\setminus\Psi$, if  $1=|\Psi_{\rightarrow f}|- |\Psi_{\rightarrow f}'|> |\Psi_{f\rightarrow }|- |\Psi'_{f\rightarrow}|=0$, then $\Psi=\Psi'\cup\{\omega\}$, and if $|\Psi_{\rightarrow f}|- |\Psi_{\rightarrow f}'|= |\Psi_{f\rightarrow }|- |\Psi'_{f\rightarrow}|=0$, then $\Psi=\Psi'$ or $\Psi=\Psi'\cup\{\omega\}\setminus\{\omega'\}$ for $\omega'\in\Psi'_{\rightarrow f}$. In either case, we have
	$$|\Psi_{\rightarrow f}\setminus\Psi'_{\rightarrow f}|+|\Psi'_{f \rightarrow}\setminus\Psi_{f\rightarrow}|\leq 1,\,\text{ and }\,\, |\Psi'_{\rightarrow f}\setminus\Psi_{\rightarrow f}|+|\Psi_{f \rightarrow}\setminus\Psi'_{f\rightarrow}|\leq 1.$$
\end{proof}

\section{Proof of Proposition~\ref{weak}}

%The second lemma states, under the assumption of FS, LAD and LAS, we can always  
%
%
% following lemma, which states that if a bundle of trades is demanded by a firm $f$ at a price vector, then this bundle is also demanded by $f$, if all downstream trades for $f$ in the bundle become more expensive, all downstream trades $f$ not in the bundle become cheaper, all upstream trades $f$ in the bundle become cheaper, and all upstream trades for $f$ not in the bundle become more expensive.

\begin{proof}
	We first define the demand $\tilde{D}^f$ and show that it is a selection from $D^f$. Then we rationalize it by a continuous and monotonic utility function that induces the same indirect utility. Afterwards we show that it satisfies FS, LAD and LAS.
	
		For each $\Psi\subseteq \Omega_f$, consider the (possibly empty) set of price vectors $p$ such that $\Psi$ is the unique demanded bundle at $p$:
	$$P_{\Psi}:=\{p\in\mathbb{R}^{\Omega_f}: D^f(p)=\{\Psi\}\}.$$
	
	Let $\bar{P}_{\Psi}$ be the (topological) closure of $P_{\Psi}$. We let $$\tilde{D}^f(p):=\{\Psi\subseteq\Omega_f:p\in\bar{P}_{\Psi}\}.$$	
		By upper hemi-continuity of $D^f$, for each $\Psi\subseteq \Omega_f$ and $p\in\bar{P}_{\Psi}$ we have $\Psi\in D^f(p)$. Thus $\tilde{D}^f(p)\subseteq D^f(p)$ for each $p\in \mathbb{R}^{\Omega_f}$.
%	 Define the projection of the set $\bar{P}_{\Psi}$ to $\mathbb{R}^{\Psi}$, $$\bar{P}^{pro}_{\Psi}:=\{p\in\mathbb{R}^{\Psi}:\exists \tilde{p}\in\bar{P}_{\Psi}\text{ with }\tilde{p}|_{\Psi}=p\}.$$

	\begin{claim}
	For each $\Psi\subseteq \Omega_f$,
	let ${P}^{pr}_{\Psi}:=\{(p_{\omega})_{\omega\in\Psi}: {p}\in P_{\Psi}\}$ be the projection of ${P}_{\Psi}$ to $\mathbb{R}^{\Psi}$ and $\overline{P_{\Psi}^{pr}}$ its (topological) closure. Then there is a continuous and monotonic utility function $\tilde{u}^f(\Psi,\cdot)$ such that 
	\begin{align}
	&\tilde{u}^f(\Psi,p)=u^f(\Psi,p),\quad&\text{if } p\in \overline{P_{\Psi}^{pr}},\\&\tilde{u}^f(\Psi,p)<u^f(\Psi,p),\quad\quad&\text{if }p\notin \overline{P_{\Psi}^{pr}}.
	\end{align}
\end{claim}
	
	%Next we show that $\tilde{D}^f$ can be rationalized by a continuous and monotonic utility function $\tilde{u}^f$.
	%For notional convenience, we define $\tilde{u}^f$ on the larger domain $\mathbb{R}^{\Omega_f}\times 2^{\Omega_f}$ and make sure that $\tilde{u}^f(\Psi,p)=\tilde{u}^f(\Psi,q)$ if $p|_{\Psi}=q|_{\Psi}$. 
	%To rationalize $\tilde{D}^f$ we have to choose a $\tilde{u}^f$ such that
	%%
	%% For each $\Psi\subseteq \Omega_f$ define the projection of the set $\bar{P}_{\Psi}$ to $\mathbb{R}^{\Psi}$, $$\bar{P}^{pro}_{\Psi}:=\{p\in\mathbb{R}^{\Psi}:\exists \tilde{p}\in\bar{P}_{\Psi}\text{ with }\tilde{p}|_{\Psi}=p\}.$$ We show that for each $\Psi\subseteq \Omega_f$ there is a continuous and monotonic function $\tilde{u}^f(\Psi,\cdot):\mathbb{R}^{\Psi}\rightarrow\mathbb{R}\cup\{-\infty\}$ such that 
	%\begin{align}
	%\tilde{u}^f(\Psi,p)=u^f(\Psi,p),\quad\text{if } {p}\in\bar{P}_{\Psi},\\\tilde{u}^f(\Psi,p)<u^f(\Psi,p),\quad\text{if } {p}\notin\bar{P}_{\Psi}.
	%\end{align}
	\begin{proof}[Proof of Claim~1]
Denote for each $p\in\mathbb{R}^{\Psi}$ by $d(p,{P}^{pr}_{\Psi}):=\inf_{q\in P^{pr}_{\Psi}}\|p-q\|$ the distance from $p$ to $P_{\Psi}^{pr}$. We define
	$$\tilde{u}^f(\Psi,p)=\begin{cases}
	u^f(\Psi,p)-d(p,{P}^{pr}_{\Psi}),\quad&\text{if }{P}_{\Psi}\neq \emptyset,\\
	-\infty,\quad &\text{if }{P}_{\Psi}= \emptyset.
	\end{cases}$$
	Suppose $P_{\Psi}\neq \emptyset$. 
	The function $\tilde{u}^f(\Psi,\cdot)$ is continuous since $u^f(\Psi,\cdot)$ is continuous and the distance to a set in $\mathbb{R}^{\Psi}$ is  continuous. Moreover, $d(p,P_{\Psi}^{pr})\geq 0$ with equality if and only if $p\in \overline{P_{\Psi}^{pr}}.$ Thus $(2)$ and $(3)$ hold.
	It remain to show that $\tilde{u}^f(\Psi,\cdot)$ is 
	monotonic.	Let $p,p'\in \mathbb{R}^{\Psi}$ with $p\neq p'$ such that $p'_{\omega}=p_{\omega}$ for $\omega\in{\Psi_{f\rightarrow}}$ and $p_{\omega}\leq p'_{\omega}$ for $\omega\in\Psi_{\rightarrow f}$ (an analogous argument works for downstream trades).
	For each $\epsilon>0$ there is $q'\in{P_{\Psi}^{pr}}$ such that $|\|p'-q'\|-d(p',P_{\Psi}^{pr})|<\epsilon$. Let $r'\in P_{\psi}$ such that $r'|_{\Psi}=q'$. Let $q:=p-(p'-q')$ and define $r\in\mathbb{R}^{\Omega_f}$ by $r_{\omega}=q_{\omega}$ for $\omega\in \Psi$ and $r_{\omega}=r'_{\omega}$ for $\omega\notin \Psi$. Since $r'\in P_{\Psi}$ we have $D^f(r')=\{\Psi\}$  and thus, by the first part of Lemma~\ref{DemandInv}, we have $\Psi\in D^f(r)$. More generally, by upper hemi-continuity of $D^f$ there is a $\epsilon'>0$ such that for each $s'\in\mathbb{R}^{\Omega_f}$ with  $\|s'-r'\|<\epsilon'$ we have $D^f(s')=\{\Psi\}$. Thus, by the first part of Lemma~\ref{DemandInv}, for each $s\in\mathbb{R}^{\Omega_f}$ with  $\|s-r\|<\epsilon'$ we have $\Psi\in D^f(s)$. By the second part of Lemma~\ref{Fede}, this implies that for each $\tilde{\epsilon}>0$ there is a $s\in P_{\Psi}$ with $\|s-r\|<\tilde{\epsilon}$. Therefore  $q=r|_{\Psi}\in (\bar{P}_{\Psi})^{pr}\subseteq \overline{P_{\Psi}^{pr}}$. Thus, $d(p,P_{\Psi}^{pr})=d(p,\overline{P_{\Psi}^{pr}})\leq \|p-q\|=\|p'-q'\|<d(p',P_{\Psi}^{pr})+\epsilon$. Since this holds for any $\epsilon>0$ we have $d(p,P_{\Psi}^{pr})\leq d(p',P_{\Psi}^{pr})$. Thus, $\tilde{u}^f(\Psi,p)>\tilde{u}^f(\Psi,p').$

\end{proof}
Claim~2 implies that $\tilde{D}^f$ can be rationalized by a continuous and monotonic utility function that induces the same indirect utility: By the second part of Lemma~\ref{Fede}, for each ${p}\in\mathbb{R}^{\Omega_f}$ there is a $\Psi\in D^f({p})$ with ${p}\in\bar{P}_{\Psi}$. Since $p\in \bar{P}_{\Psi}$ we have $p|_{\Psi}\in (\bar{P}_{\Psi})^{pr}\subseteq\overline{P_{\Psi}^{pr}}$ and therefore $\tilde{v}^f(p)=\tilde{u}^f(\Psi,p)=u^f(\Psi,p)=v^f({p})$. Next we show that $\tilde{u}^f$ rationalizes $\tilde{D}^f$ by showing that $\tilde{u}^f(\Psi,p)=\tilde{v}^f(p)$ for $\Psi\in \tilde{D}^f(p)$ and $\tilde{u}^f(\Psi,p)<\tilde{v}^f(p)$ for $\Psi\notin \tilde{D}^f(p)$. Let $\Psi\subseteq\Omega_f$. If $\Psi\in \tilde{D}^f({p})$, then ${p}\in\bar{P}_{\Psi}$  and  $\Psi\in D^f({p})$. Since $p\in\bar{P}_{\Psi}$ we have $p|_{\Psi}\in(\bar{P}_{\Psi})^{pr}\subseteq \overline{P_{\Psi}^{pr}}$ and thus $\tilde{u}^f(\Psi,p)=u^f(\Psi,p)=v^f(p)=\tilde{v}^f(p)$. 
If $\Psi\notin \tilde{D}^f(p)$ and $\Psi\notin D^f(p)$, then $\tilde{u}^f(\Psi,p)\leq u^f(\Psi,p)<v^f(p)=\tilde{v}^f(p)$. If $\Psi\in D^f(p)\setminus\tilde{D}^f(p)$, then we show $p|_{\Psi}\notin \overline{P^{pr}_{\Psi}}$ and thus $\tilde{u}^f(\Psi,p)<u^f(\Psi,p)=v^f(p)=\tilde{v}^f(p)$:
By definition of $\tilde{D}^f$ we have $p\notin \bar{P}_{\Psi}$. Thus, there is a $\epsilon>0$ such that for each $q\in\mathbb{R}^{\Omega_f}$ with $\|p-q\|<\epsilon$ we have $q\notin \bar{P}_{\Psi}$. Let $q\in\mathbb{R}^{\Omega_f}$ such that $\|p-q\|<\epsilon$ and $p_{\omega}=q_{\omega}$ for $\omega\in \Psi$, $p_{\omega}>q_{\omega}$ for $\omega\in \Omega_{f\rightarrow}\setminus\Psi$ and $p_{\omega}<q_{\omega}$ for $\omega\in \Omega_{\rightarrow f}\setminus\Psi$. Since $\Psi\in D^f(p)$ and $u^f$ is monotonic, we have $D^f(q)\subseteq 2^{\Psi}$. By Lemma~\ref{Fede}, we can find a $\epsilon-\|p-q\|>\tilde{\epsilon}>0$ such that for $\|r-q\|<\tilde{\epsilon}$ we have $D^f(r)\subseteq D^f(q)\subseteq 2^{\Psi}$. Now suppose for the sake of contradiction that $p|_{\Psi}=q|_{\Psi}\in\overline{P^{pr}_{\Psi}}$. Then, there is a $r\in P_{\Psi}$ such that $\|p|_{\Psi}-r|_{\Psi}\|<\tilde{\epsilon}$. Since $r\in P_{\Psi}$, we have $D^f(r)=\{\Psi\}$ and thus, in particular,
$u^f(\Psi,r)>u^f(\tilde{\Psi},r)$ for each $\tilde{\Psi}\subsetneqq \Psi$. Now define $\tilde{r}\in\mathbb{R}^{\Omega_f}$ by $\tilde{r}_{\omega}=r_{\omega}$ for $\omega\in\Psi$ and $\tilde{r}_{\omega}=q_{\omega}$ for $\omega\notin \Psi.$ By construction, we have $\|\tilde{r}-q\|=\|r|_{\Psi}-p|_{\Psi}\|<\tilde{\epsilon}$. Thus $D^f(\tilde{r})\subseteq 2^{\Psi}$. Moreover, $u^f(\Psi,\tilde{r})=u^f(\Psi,r)>u^f(\tilde{\Psi},r)=u^f(\tilde{\Psi},\tilde{r})$ for each $\tilde{\Psi}\subsetneqq \Psi$. Thus $D^f(\tilde{r})=\{\Psi\}$ and $\tilde{r}\in P_{\Psi}$. However, $\|p-\tilde{r}\|\leq \|p-q\|+\|q-\tilde{r}\|< \|p-q\|+\tilde{\epsilon}<\epsilon$ and therefore $\tilde{r}\notin \bar{P}_{\Psi}$, a contradiction.

%		First we show that for $p\in P_{\Psi}$ and $q\in\mathbb{R}^{\Omega_f}$ with $p_{\omega}\leq \bar{p}_{\omega}$ for $\omega\in\Psi_{\rightarrow f}$ and $p_{\omega}\geq \bar{p}_{\omega}$ for $\omega\in\Psi_{f\rightarrow }$ and $p_{\omega}=q_{\omega}$ there is a $q\in P_{\Psi}$
%		If $\bar{P}_{\Psi}=\emptyset$, the claim trivially holds, with $\bar{p}_{\omega}=-\infty$ for $\omega\in \Psi_{\rightarrow f}$ and $\bar{p}_{\omega}=\infty$ for $\omega\in \Psi_{f\rightarrow}$.
%		Otherwise, let $p,q\in P_{\Psi}$. First we show that   $\bar{p}\in{P}_{\Psi}$ for $\bar{p}_{\omega}:=\max\{p_{\omega},q_{\omega}\}$ for $\omega\in\Omega_{\rightarrow f}$ and $\bar{p}_{\omega}:=\min\{p_{\omega},q_{\omega}\}$ for $\omega\in\Omega_{f\rightarrow}$ (more generally one could show that $P_{\Psi}$ is a lattice).
%		Let $\tilde{\Psi}\in \tilde{D}^f(p)$. We show that $\tilde{\Psi}=\Psi.$ By, 

%	Monotonicity of $\tilde{u}^f(\Psi,\cdot)$ now follows: Consider $q'\in\bar{P}_{\Psi}$ such that $\|p'-q'\|=\min_{q\in\bar{P}_{\Psi}}\|p'-q\|$.
%	By assumption, we have $p'_{\omega}-(p'_{\omega}-q'_{\omega})=p'_{\omega}-(p'_{\omega}-q'_{\omega})=q'_{\omega}$ for $\omega\in{\Psi_{f\rightarrow}}$ and $p_{\omega}-(p'_{\omega}-q'_{\omega})\leq p'_{\omega}-(p'_{\omega}-q'_{\omega})=q'_{\omega}$ for $\omega\in\Psi_{\rightarrow f}$. By the claim $p-(p'-q')\in\bar{P}_{\Psi}$. It follows that  $$\min_{q}\|p-q\|\leq \|p-(p-(p'-q'))\|=\|p'-q'\|=\min_{q}\|p'-q\| $$
%	

	Next we show that that $\tilde{D}^f$ satisfies FS, LAD and LAS. By  Proposition~\ref{single} it suffices to show that $\tilde{D}^f$ satisfies NIB, weak FS, weak LAD and weak LAS. Let $\Psi\subseteq \Omega_f$. Since $\bar{P}_{\Psi}$ is the closure of ${P}_{\Psi}$ there is for each $p\in\bar{P}_{\Psi}$ and $\epsilon>0$ a $q\in {P}_{\Psi}$ with $\|p-q\|<\epsilon$. By definition of $P_{\Psi}$ and $\tilde{D}_f$ we have $\tilde{D}^f(q)=D^f(q)=\{\Psi\}$. Thus $\tilde{D}^f$ satisfies NIB. For the other properties, recall that $D^f$ satisfies weak FS, weak LAD and weak LAS. Thus, it suffices to show that for each $p\in \mathbb{R}^{\Omega_f}$ we have $|\tilde{D}^f(p)|=1$ if and only if $|D^f(p)|=1.$  Let $\Psi\subseteq\Omega_f$ with $p\in\bar{P}_{\Psi}$. If $p\in P_{\Psi}$, then $\tilde{D}^f(p)=D^f(p)=\{\Psi\}$. If $p\in\bar{P}_{\Psi}\setminus P_{\Psi}$, i.e.~if $p$ is on the boundary of $P_{\Psi}$, then for each $\epsilon>0$ there is a $q\in\mathbb{R}^{\Omega_f}\setminus\bar{P}_{\Psi}$ with $\|p-q\|<\epsilon$. By the second part of Lemma~\ref{Fede}, we may choose $q$ such that $|D^f(q)|=1$. Since $\Omega$ is finite, this implies that there is a $\tilde{\Psi}\neq \Psi$ such that for each $\epsilon>0$ there is a $q\in\mathbb{R}^{\Omega_f}\setminus \bar{P}_{\Psi}$ with $\|p-q\|<\epsilon$ and $D^f(q)=\{\tilde{\Psi}\}$. Thus $p\in\bar{P}_{\tilde{\Psi}}$ for $\tilde{\Psi}\neq\Psi$. Hence $|\tilde{D}^f(p)|>1$. Thus, $|\tilde{D}^f(p)|=1$ if and only if $|D^f(p)|=1$ as desired.

\end{proof}

\section{Proofs for Section~\ref{lattice}}\label{anotherApp}

\subsection*{Proof of Lemma~\ref{12}}
\begin{proof}
By Lemma~\ref{Fede}, there exists an ${\epsilon}_0>0$ such that for each $p\in P$ and every $q$ with $\|q-p\|<{\epsilon}_0$ we have $D^f(q)\subseteq D^f(p)$. Let $P=\{p^1,\ldots,p^n\}$. By Lemma~\ref{Fede}, there is a $\epsilon^1\in\mathbb{R}^{\Omega_f}$ with $\|\epsilon^1\|<\epsilon_0$ such that $|D^f(p^1+\epsilon^1)|=1$ and $\Psi\in D^f(p^1)$ for the unique $\Psi\in D^f(p^1+\epsilon^1)$. Consider $P^1:=\{p^1+\epsilon^1,\ldots,p^n+\epsilon^1\}$. For each $i=1,\ldots,n$ we have $D^f(p^i+\epsilon^1)\subseteq D^f(p^i)$. By Lemma~\ref{Fede}, there exists an ${\epsilon}_1>0$ such that for each $p\in P^1$ and every $q$ with $\|q-p\|<{\epsilon}_1$ we have $D^f(q)\subseteq D^f(p)$.  By Lemma~\ref{Fede}, there is a $\epsilon^2\in\mathbb{R}^{\Omega_f}$ with $\|\epsilon^2\|<\epsilon_1$ such that $|D^f(p^2+\epsilon^1+\epsilon^2)|=1$ and $\Psi\in D^f(p^2+\epsilon^2)$ for the unique $\Psi\in D^f(p^2+\epsilon^1)$. Next consider $P^2:=\{p^1+\epsilon^1+\epsilon^2,\ldots,p^n+\epsilon^1+\epsilon^2\}$. For each $i=1,\ldots,n$ we have $D^f(p^i+\epsilon^1+\epsilon^2)\subseteq D^f(p^i+\epsilon^1)\subseteq D^f(p^i)$ and so on. Iterating in this way, we obtain $\epsilon^1,\ldots,\epsilon^n$ such that for each $i=1,\ldots,n$, we have $|D^f(p^i+\sum_{j=1}^n\epsilon^j)|=1$ and $\Psi^i\in D^f(p^i)$ for the unique $\Psi^i\in D^f(p^i+\sum_{j=1}^n\epsilon^j)\subseteq D^f(p^i)$. We define $\tilde{D}^f(p^i)=\Psi^i$. By construction $\tilde{D}^f(p^i)\in D^f(p^i)$. Moreover, as all price vectors are translated by the same vector $\sum_{j=1}^n\epsilon^j$, FS, LAD and LAS follow from weak FS, weak LAD and weak LAS for $D^f$.
\end{proof}

\subsection*{Proof of Lemma~\ref{lemma}}
\begin{proof}

By Lemma~\ref{Fede}, there exists an ${\epsilon}>0$ such that
for each $q\in\{p,p',\bar{p},\underline{p}\}$ and every
$\tilde{q}$ with $\|\tilde{q}-q\|<{\epsilon}$ we have
$D^f(\tilde{q})\subseteq D^f(q)$. Let $\epsilon_0>0$ such that $\epsilon_0<\min_{\omega\in\Omega_f:p'_{\omega}\neq p_{\omega}}|p'_{\omega}-p_{\omega}|$ and $\epsilon_0\sqrt{|\Omega_f|}<\epsilon.$ Define ${\epsilon}'\in\mathbb{R}^{\Omega_f}$ by
$${\epsilon}'_{\omega}=\begin{cases}
{{\epsilon_0}},&\text{ if }\omega\in \Psi'_{f\rightarrow}\text{ and }p'_{\omega}\neq p_{\omega},\\
-{{\epsilon}_0},&\text{ if }\omega\in \Omega_{f\rightarrow}\setminus \Psi'\text{ and }p'_{\omega}\neq p_{\omega},\\
-{{\epsilon_0}},&\text{ if }\omega\in \Psi'_{\rightarrow f}\text{ and }p'_{\omega}\neq p_{\omega},\\
{{\epsilon}_0},&\text{ if }\omega\in \Omega_{\rightarrow f}\setminus \Psi'\text{ and }p'_{\omega}\neq p_{\omega},\\
 0,&\text{ if }p'_{\omega}=p_{\omega}.\end{cases}$$
Note that by construction we have $\|\epsilon'\|=\sqrt{\epsilon_0^2|\{\omega\in \Omega_f:p_{\omega}\neq p_{\omega}'\}|}\leq\epsilon_0\sqrt{|\Omega_f|} <\epsilon$ and thus $D^f(p'+\epsilon')\subseteq D^f(p')$. First we prove the following claim.
\begin{claim}\label{clcl}
For each $\Xi\in D^f(p'+\epsilon')$ we have $\{\omega\in\Psi':p'_{\omega}\neq p_{\omega}\}\subseteq \Xi$ and  $\{\omega\notin \Psi':p'_{\omega}\neq p_{\omega}\}\cap\Xi=\emptyset$.
\end{claim}
\begin{proof}
 First we show that for each $\Xi\in D^f(p'+\epsilon')$ we have $\{\omega\in\Psi':p'_{\omega}\neq p_{\omega}\}\subseteq \Xi$. Suppose not, and there is a $\Xi\in D^f(p'+\epsilon')$ and a $\tilde{\omega}\in\{\omega\in\Psi':p'_{\omega}\neq p_{\omega}\}\setminus\Xi$. Let $\tilde{p}\in\mathbb{R}^{\Omega_f}$ with $\tilde{p}_{\tilde{\omega}}=p'_{\tilde{\omega}}$ and $\tilde{p}_{\omega}=p'_{\omega}+\epsilon'_{\omega}$ for $\omega\neq\tilde{\omega}$. By the second part of Lemma~\ref{DemandInv}, we have $\Psi'\in D^f(\tilde{p})$. Thus, by monotonicity,
we have $u^f(\Xi,p'+\epsilon')=u^f(\Xi,\tilde{p})\leq u^f(\Psi',\tilde{p})<u^f(\Psi',p'+\epsilon')$ contradicting the assumption that $\Xi\in D^f(p'+\epsilon')$.

Next we show that for each $\Xi\in D^f(p'+\epsilon')$ we have $\{\omega\notin \Psi':p'_{\omega}\neq p_{\omega}\}\cap\Xi=\emptyset$. Suppose not, and there is a $\Xi\in D^f(p'+\epsilon')$ and a $\tilde{\omega}\in\{\omega\notin \Psi':p'_{\omega}\neq p_{\omega}\}\cap\Xi$. Let $\tilde{p}\in\mathbb{R}^{\Omega_f}$ with $\tilde{p}_{\tilde{\omega}}=p'_{\tilde{\omega}}$ and $\tilde{p}_{\omega}=p'_{\omega}+\epsilon'_{\omega}$ for $\omega\neq\tilde{\omega}$. By the second part of Lemma~\ref{DemandInv}, we have $\Psi'\in D^f(\tilde{p})$. Thus, by monotonicity,
we have $u^f(\Xi,p'+\epsilon')<u^f(\Xi,\tilde{p})\leq u^f(\Psi',\tilde{p})=u^f(\Psi',p'+\epsilon')$ contradicting the assumption that $\Xi\in D^f(p'+\epsilon')$.\end{proof}

By Lemma~\ref{Fede}, there exists ${{\epsilon}}_1>0$ such that
for every
$q'$ with $\|{q}'-(p'+\epsilon')\|<{{\epsilon}_1}$ we have
$D^f(q')\subseteq D^f(p'+{\epsilon}')$. We may choose $\epsilon_1<\epsilon-\|\epsilon'\|$. 
By Proposition~\ref{NIB}, there is a $q\in\mathbb{R}^{\Omega_f}$ with $\|p-q\|<\epsilon_1$ such that $D^f(q)=\{\Psi\}.$ Define $q':=p'+\epsilon'+(q-p).$ Define $\bar{q}$ as the pairwise maximum of $q$ and $q'$, i.e.~$\bar{q}_{\omega}=\max\{q_{\omega},q'_{\omega}\}$, and
	$\underline{q}$ as the pairwise minimum of $q$ and $q'$, i.e.~$\underline{q}_{\omega}=\min\{q_{\omega},q'_{\omega}\}$.

By construction, we have $D^f(q')\subseteq D^f(p'+\epsilon')\subseteq D^f(p')$, we have $\|\bar{q}-\bar{p}\|\leq \|\epsilon'\|+\epsilon_1<\epsilon$ and thus $D^f(\bar{q})\subseteq D^f(\bar{p})$, and we have $\|\underline{q}-\underline{p}\|\leq \|\epsilon'\|+\epsilon_1<\epsilon$ and thus $D^f(\underline{q})\subseteq D^f(\underline{p})$. Moreover, by construction, $q_{\omega}<q'_{\omega}$ if and only if $p_{\omega}<p'_{\omega}$, $q_{\omega}>q'_{\omega}$ if and only if $p_{\omega}>p'_{\omega}$, and $q_{\omega}=q'_{\omega}$ if and only if $p_{\omega}=p'_{\omega}$. 

 Let $P:=\{\tilde{q}\in\mathbb{R}^{\Omega_f}:\tilde{q}_{\omega}\in\{q_{\omega},q'_{\omega}\}\text{ for all }\omega\in\Omega_f\}$. By Lemma~\ref{12}, there is a single-valued selection $\tilde{D}^f:P\rightarrow2^{\Omega_f}$ from $D^f$ satisfying FS, LAD and LAS. Let $\bar{\Psi}:=\tilde{D}^f(\bar{q})$ and $\underline{\Psi}:=\tilde{D}^f(\underline{q})$.
As $D^f(q)=\{\Psi\}$, we have $\tilde{D}^f(q)=\Psi$. Moreover, by Claim~\ref{clcl} and as $\Psi'':=\tilde{D}^f(q')\in D^f(p'+\epsilon')$, we have $\{\omega\in \Psi':p_{\omega}'\neq p_{\omega}\}\subseteq \Psi''$ and $\{\omega\notin\Psi':p'_{\omega}\neq p_{\omega}\}\cap \Psi''=\emptyset$. By FS of $\tilde{D}^f$ and since  $\{\omega\in \Psi':p_{\omega}'\neq p_{\omega}\}\subseteq \Psi''$, we have
\begin{align*}
&\{\omega\in \Psi_{\rightarrow f}: p_{\omega}\geq p'_{\omega}\}\cup\{\omega\in  \Psi'_{\rightarrow f}:p'_{\omega}>p_{\omega}\}\\
\subseteq&\{\omega\in \Psi_{\rightarrow f}: q_{\omega}\geq q'_{\omega}\}\cup\{\omega\in \Psi''_{\rightarrow f}:q'_{\omega}>q_{\omega}\}\subseteq \bar{\Psi}_{\rightarrow f}.
\end{align*}
 Next we show that
$$\bar{\Psi}_{f\rightarrow}\subseteq\{\omega\in \Psi_{f\rightarrow}: p'_{\omega}\geq p_{\omega}\}\cup\{\omega\in  \Psi'_{f\rightarrow }:p'_{\omega}>p_{\omega}\}.$$
Let $\bar{\omega}\in \bar{\Psi}_{f\rightarrow}$. We consider two cases. Either $\bar{p}_{\bar{\omega}}=p_{\bar{\omega}}$ or $\bar{p}_{\bar{\omega}}=p'_{\bar{\omega}}>p_{\bar{\omega}}$. In the first case, consider $\tilde{p}\in P$ with $\tilde{p}_{\omega}=\bar{q}_{\omega}$ for $\omega\in \Omega_{\rightarrow f}$ and $\tilde{p}_{\omega}=q_{\omega}$ for $\omega\in \Omega_{f\rightarrow }$. Let $\tilde{\Psi}:=\tilde{D}^f(\tilde{p}).$ By SSS of $\tilde{D}^f$, we have $\bar{\omega}\in\tilde{\Psi}_{f\rightarrow}$. By CSC, we have $\tilde{\Psi}_{f\rightarrow}\subseteq{\Psi}_{f\rightarrow}$ and hence $\bar{\omega}\in\Psi_{f\rightarrow}$.

Similarly, if $\bar{p}_{\bar{\omega}}=p'_{\bar{\omega}}>p_{\bar{\omega}}$, consider $\tilde{p}\in P$ with $\tilde{p}_{\omega}=\bar{q}_{\omega}$ for $\omega\in \Omega_{\rightarrow f}$ and $\tilde{p}_{\omega}=q'_{\omega}$ for $\omega\in \Omega_{f\rightarrow }$. Let $\tilde{\Psi}:=\tilde{D}^f(\tilde{p}).$ By SSS of $\tilde{D}^f$  we have $\bar{\omega}\in\tilde{\Psi}_{f\rightarrow}$. By CSC, we have $\tilde{\Psi}_{f\rightarrow}\subseteq\Psi''_{f\rightarrow}$ and hence $\bar{\omega}\in\Psi''_{f\rightarrow}$. Since $\{\omega\notin\Psi':p'_{\omega}\neq p_{\omega}\}\cap \Psi''=\emptyset$ and $p'_{\bar{\omega}}\neq p_{\bar{\omega}}$ this implies $\bar{\omega}\in\Psi'_{f\rightarrow}$.
%and $$\{\omega\in \Psi_{\rightarrow f}: p'_{\omega}\geq p_{\omega}\}\cup\{\omega\in \Psi'_{\rightarrow f}: p_{\omega}> p'_{\omega}\}\subseteq\underline{\Psi}_{f\rightarrow}.$$
%By the laws of aggregate demand and supply for $\tilde{D}^f$, we have
%$$|\underline{\Psi}_{\rightarrow f}|-|\underline{\Psi}_{f\rightarrow}|\geq|\Psi_{\rightarrow f}|-|\Psi_{f\rightarrow}|\geq|\bar{\Psi}_{\rightarrow f}|-|\bar{\Psi}_{f\rightarrow}|.$$
A completely analogous proof shows that $\underline{\Psi}$ has the desired properties. Finally, by LAD and LAS for $\tilde{D}^f$, we have
$$|\underline{\Psi}_{\rightarrow f}|-|\underline{\Psi}_{f\rightarrow }|\geq|{\Psi}_{\rightarrow f}|-|{\Psi}_{f\rightarrow }|\geq|\bar{\Psi}_{\rightarrow f}|-|\bar{\Psi}_{f\rightarrow }|.$$

\end{proof}
\subsection*{Proof of Theorem~\ref{Util}}
\begin{proof}
Let $\Xi\in \mathcal{E}(u,p)$ and $\Xi'\in\mathcal{E}(u,p')$. Define
\begin{align*}
\overline{\Xi}&:=\{\omega\in\Xi: p_{\omega}\geq p'_{\omega}\}\cup\{\omega\in\Xi': p'_{\omega}> p_{\omega}\},\\
\underline{\Xi}&:=\{\omega\in\Xi: p'_{\omega}\geq p_{\omega}\}\cup\{\omega\in\Xi': p_{\omega}> p'_{\omega}\}.
\end{align*}
We show that $\overline{\Xi}\in\mathcal{E}(u,\bar{p})$ and $\underline{\Xi}\in\mathcal{E}(u,\underline{p})$.
Let $f\in F$. By Lemma~\ref{lemma},
with $\Psi=\Xi_f$ and $\Psi'=\Xi'_f$ there is a $\overline{\Psi}_f\in D^f (\bar{p})$ and a $\underline{\Psi}_f\in D^f(\underline{p})$ such that
$\overline{\Xi}_{\rightarrow f}\subseteq\overline{\Psi}_{\rightarrow f}$, $\overline{\Psi}_{f\rightarrow}\subseteq\overline{\Xi}_{f\rightarrow}$, $\underline{\Psi}_{\rightarrow f}\subseteq\underline{\Xi}_{\rightarrow f}$ and $\underline{\Xi}_{f\rightarrow}\subseteq\underline{\Psi}_{f\rightarrow}$ and
$$|\underline{\Psi}_{\rightarrow f}|-|\underline{\Psi}_{f\rightarrow }|\geq|{\Xi}_{\rightarrow f}|-|{\Xi}_{f\rightarrow }|\geq|\overline{\Psi}_{\rightarrow f}|-|\overline{\Psi}_{f\rightarrow }|.$$
Note that this implies
$$|\underline{\Xi}_{\rightarrow f}|-|\underline{\Xi}_{f\rightarrow }|\geq|\underline{\Psi}_{\rightarrow f}|-|\underline{\Psi}_{f\rightarrow }|\geq|{\Xi}_{\rightarrow f}|-|{\Xi}_{f\rightarrow }|\geq|\overline{\Psi}_{\rightarrow f}|-|\overline{\Psi}_{f\rightarrow }|\geq|\overline{\Xi}_{\rightarrow f}|-|\overline{\Xi}_{f\rightarrow }|.$$
Summing the inequalities over all firms, we obtain
$$0\geq\sum_{f\in F}(|\underline{\Psi}_{\rightarrow f}|-|\underline{\Psi}_{f\rightarrow }|)\geq0\geq\sum_{f\in F}(|\overline{\Psi}_{\rightarrow f}|-|\overline{\Psi}_{f\rightarrow }|)\geq0.$$
Thus
$$|\underline{\Xi}|=\sum_{f\in F}|\underline{\Xi}_{f\rightarrow }|\leq\sum_{f\in F}|\underline{\Psi}_{f \rightarrow }|=\sum_{f\in F}|\underline{\Psi}_{\rightarrow f}|\leq\sum_{f\in F}|\underline{\Xi}_{\rightarrow f}|=|\underline{\Xi}|.$$

Therefore $\underline{\Xi}_f=\underline{\Psi}_{f}$ for each $f\in F$.
Moreover,
$$|\overline{\Xi}|=\sum_{f\in F}|\overline{\Xi}_{\rightarrow f}|\leq\sum_{f\in F}|\overline{\Psi}_{\rightarrow f}|=\sum_{f\in F}|\overline{\Psi}_{f\rightarrow }|\leq\sum_{f\in F}|\overline{\Xi}_{\rightarrow f}|=|\overline{\Xi}|.$$
Therefore $\overline{\Xi}_f=\overline{\Psi}_{f}$ for each $f\in F$.

Next we show that the above construction implies the rural hospital theorem: As $\overline{\Psi}=\overline{\Xi}$ and $\underline{\Psi}=\underline{\Xi}$, for each $f\in F$ we have
$$|\underline{\Xi}_{\rightarrow f}|-|\underline{\Xi}_{f\rightarrow }|\geq|{\Xi}_{\rightarrow f}|-|{\Xi}_{f\rightarrow }|\geq|\overline{\Xi}_{\rightarrow f}|-|\overline{\Xi}_{f\rightarrow }|.$$
Summing the inequalities over all $f$, we obtain
$$0=\sum_{f\in F}|\underline{\Xi}_{\rightarrow f}|-\sum_{f\in F}|\underline{\Xi}_{f\rightarrow }|\geq\sum_{f\in F}|{\Xi}_{\rightarrow f}|-\sum_{f\in F}|{\Xi}_{f\rightarrow }|\geq\sum_{f\in F}|\overline{\Xi}_{\rightarrow f}|-\sum_{f\in F}|\overline{\Xi}_{f\rightarrow }|=0.$$
Thus, for each $f\in F$ we have
$$|\underline{\Xi}_{\rightarrow f}|-|\underline{\Xi}_{f\rightarrow }|=|{\Xi}_{\rightarrow f}|-|{\Xi}_{f\rightarrow }|=|\overline{\Xi}_{\rightarrow f}|-|\overline{\Xi}_{f\rightarrow }|.$$
Next define the set of trades
\begin{align*}
\Xi''&:=\{\omega\in\overline{\Xi}: \bar{p}_{\omega}= p_{\omega}\}\cup\{\omega\in\Xi': \bar{p}_{\omega}> p'_{\omega}\}.
\end{align*}
Since $\bar{\Xi}\in\mathcal{E}(u,\bar{p})$ and $\Xi'\in \mathcal{E}(u,p')$, the same argument as before with $\overline{\Xi}$ in the role of $\Xi$, and $\bar{p}$ in the role of $p$ establishes (note that the pairwise minimum of $\bar{p}$ and $p'$ is again $p'$) that $[\Xi'',p']$ is an equilibrium. Moreover, with the same argument as before for each $f\in F$ we have
$$|{\Xi}''_{\rightarrow f}|-|{\Xi}''_{f\rightarrow }|=|\overline{\Xi}_{\rightarrow f}|-|\overline{\Xi}_{f\rightarrow }|.$$
Since $$|\overline{\Xi}_{\rightarrow f}|-|\overline{\Xi}_{f\rightarrow }|=|{\Xi}_{\rightarrow f}|-|{\Xi}_{f\rightarrow }|,$$
this concludes the proof.
%Now consider the meet of $[\bar{\X}]$
% Note that for each $f\in F$ we have $\sum_{f\in F}|\underline{\mu}(f)|=|\mu(f)|=|\bar{\mu}(f)|$. In particular for the infimum $(\underline{\mu},\underline{s})$ of $(\bar{\mu},\bar{s})$ (in the role of $(\mu,s)$) and $(\mu',s')$, for each $f\in F$ we have $|\underline{\mu}(f)|=|\bar{\mu}(f)|$. As $u_w(\bar{\mu},\bar{s})\geq u_w(\mu',s')$ for each $w\in W$, we have $u_w(\underline{\mu},\underline{s})=u_w(\mu',s')$ for each $w\in W$. Thus we can let $(\mu'',s'')=(\underline{\mu},\underline{s})$ to obtain $|\mu''(f)|=|\mu(f)|$ for each $f\in F$. Note furthermore that by the third part of Lemma~\ref{1}, in two core allocations under which all workers have the same utility, all firms obtain the same profit. Thus $\pi_{f}(\mu'',s'')=\pi_f(\mu',s')$ for each $f\in F$.
 \end{proof}
\subsection*{Example for the Failure of the Rural Hospitals Theorem without LAD}
\begin{continuance}{ex3}
Consider the set of trades $\Omega=\{\omega_1,\omega_2\}$ and firm $f$ with the utility function $u^f$ as defined in Example~\ref{ex3}. As observed before, $D^f$ satisfies FS, weak LAD, (and, trivially, LAS), but not LAD. Consider a second firm $f'$ with $f'=s(\omega_1)=s(\omega_2)$ with utility function $u^{f'}$ defined by
\begin{align*}
&u^{f'}(\{\omega_i\},p_{\omega_i})=p_{\omega_i},\quad&\text{ for }i=1,2,\\
&u^{f'}(\{\omega_1,\omega_2\},p)=p_{\omega_1}+p_{\omega_2}-1.5,&\\
&u^{f'}(\emptyset)=0.
\end{align*}
The induced demand satisfies FS and LAS (and, trivially, LAD).
  The set of equilibrium vectors is $\mathcal{E}(u)=\{p:1\leq p_{\omega_1}=p_{\omega_2}\leq1.5\}\cup\{(2,2)\}$. Each $p\in\mathcal{E}(u)\setminus\{(2,2)\}$ is supported by $\{\omega_1\}$ and by $\{\omega_2\}$. The equilibrium prices $(2,2)$ are supported by $\{\omega_1,\omega_2\}$.
  An analogous example can be constructed to show that LAS and not just weak LAS is necessary for the Rural Hospitals Theorem.
  \qed
%%      Note that the Lemma~\ref{lemma} fails to hold under weak FS: For the utility
%     function $u^f$ in Example~1, we have
%    $\Psi'=\{\alpha_1,\alpha_2,\beta_1,\beta_2\}\in D^f(1,1,1,1),$ but
%   $\Psi'\notin D^f(0,1,1,1)$.\qed}
%\qed
\end{continuance}
\subsection*{Proof of Theorem~\ref{extreme}}
We first show that the result holds for utility functions satisfying BWP.
\begin{proposition}\label{BWP}
Under the assumption of BWP, FS, LAD, LAS, there exists a seller-optimal equilibrium, i.e.~a $\bar{p}\in\mathcal{E}(u)$ such that for each terminal seller $f\in F$:
$$v^f(\bar{p})\geq v^f(p)\text{ for each }p\in\mathcal{E}(u),$$
and a buyer-optimal equilibrium, i.e.~a~$\underline{p}\in\mathcal{E}(u)$ such that for each terminal buyer $f\in F$:
$$v^f(\underline{p})\geq v^f(p)\text{ for each }p\in\mathcal{E}(u).$$
\end{proposition}
\begin{proof}
        Following an idea of~\cite{KelsoCrawford1982}, we can characterize competitive equilibria by a zero-surplus condition. Define a surplus function $Z:\mathbb{R}^{\Omega}\rightarrow\mathbb{R}$ by
    $$Z(p):=\min_{\Psi\subseteq\Omega}\max_{f\in F}\max_{\Psi'\subseteq\Omega_f}u^f(\Psi',p)-u^f(\Psi,p).$$
    By definition, for each $f\in F$, we have $D^f(p)=\text{argmax}_{\Psi'\subseteq\Omega_f}u^f(\Psi',p).$ Thus for each arrangement $[\Psi,p]$, we have $\max_{f\in F}\max_{\Psi'\subseteq\Omega_f}u^f(\Psi',p)-u^f(\Psi,p)\geq0$ with equality if and only if $\Psi\in\mathcal{E}(u,p)$. Thus $p\in\mathcal{E}(u)$ if and only if $Z(p)=0.$  The surplus function is continuous, as $u^f(\Psi',p)-u^f(\Psi,p)$ is continuous in $p$ and the maximum resp.~minimum of finitely many continuous functions is continuous. Thus $\mathcal{E}(u)$ is a closed set, as it is the pre-image of the closed set $\{0\}$ under the continuous function $Z$.

    By BWP, there is a $K>0$ such that for all $f\in F$, $p\in\mathbb{R}^{\Omega_f}$ and $\Psi\in D^f(p)$
    if $\omega\in\Psi_{\rightarrow f}$ then $p_{\omega}<K$ and
    if $\omega\in \Psi_{f\rightarrow }$ then $p_{\omega}>-K$. Let $\mathcal{E}'(u):=\mathcal{E}(u)\cap[-K,K]^{\Omega}$. By BWP, for each  $p\in\mathcal{E}(u)$, the vector $p'\in\mathbb{R}^{\Omega}$ defined by $p'_{\omega}=p_{\omega}$ for $-K<p_{\omega}<K$, $p'_{\omega}=K$ for $p_{\omega}>K$, and  $p'_{\omega}=-K$ for $p_{\omega}<-K$ is an equilibrium price vector $p'\in\mathcal{E}'(u)$ with $v^f(p')=v^f(p)$ for each $f\in F$. By Corollary~2 in \cite{Ravi}, $\mathcal{E}(u)$ is non-empty and hence $\mathcal{E}'(u)$ is non-empty. As $\mathcal{E}(u)$ is closed, $\mathcal{E}'(u)$ is compact. From Theorem~\ref{Util}, and observing that the pairwise maximum (minimum) of two vectors in $[-K,K]^{\Omega}$ is an element of $[-K,K]^{\Omega}$, we conclude that $\mathcal{E}'(u)$ is a non-empty, compact sublattice of $\mathbb{R}^{\Omega}$. This implies that $\mathcal{E}'(u)$ has a maximal element $\bar{p}$ and a minimal element $\underline{p}.$ By monotonicity and the previous observation that for each $p\in\mathcal{E}(u)$ there is a $p'\in\mathcal{E}'(u)$ with $v^f(p')=v^f(p)$ for each $f\in F$, for each terminal seller $f$ and $p\in\mathcal{E}(u)$ we have $v^f(\bar{p})\geq v^f(p)$. Thus $\bar{p}$ is a terminal seller optimal equilibrium. Similarly, $\underline{p}$ is a terminal buyer optimal equilibrium $\underline{p}$ under $u$.
     \end{proof}
    To extend the result to profiles satisfying BCV, a generalization of a result of
 \cite{Hatfield2018} is useful. They show (Theorem~2) that for quasi-linear preferences, the class of FS valuations is invariant under "trade endowments". The generalization of this observation to non-quasi-linear preferences is the following lemma:
 \begin{lemma}\label{Alex}
    Let $u^f$ satisfy FS, LAD and LAS. Then for each allocation $(\bar{\Psi},\bar{p})\in\mathcal{A}_f$, the utility function ${u}^f_{(\bar{\Psi},\bar{p})}$ defined by $${u}^f_{(\bar{\Psi},\bar{p})}(\Psi,p):=\max_{\Xi\subseteq \bar{\Psi}\setminus\Psi}u^f(\Psi\cup\Xi,(p|_{\Psi},\bar{p}|_{\Xi}))$$ satisfied FS, LAD and LAS.
 \end{lemma}
\begin{proof}
    Let $\tilde{D}^f$ be the demand induced by $u^f_{(\bar{\Psi},\bar{p})}$.
Let $p,p'\in\mathbb{R}^{\Omega_f}$. Define  $q,q'\in\mathbb{R}^{\Omega_f}$ as follows:
$$q_{\omega}:=\begin{cases}
\bar{p}_{\omega},\quad&\text{if }\omega\in\bar{\Psi}_{\rightarrow f}\text{ and }\bar{p}_{\omega}<p_{\omega},\text{ or }\omega\in\bar{\Psi}_{f\rightarrow }\text{ and }\bar{p}_{\omega}>p_{\omega},\\p_{\omega},\quad&\text{else,}
\end{cases}$$ $$ q'_{\omega}:=\begin{cases}
\bar{p}_{\omega},\quad&\text{if }\omega\in\bar{\Psi}_{\rightarrow f}\text{ and }\bar{p}_{\omega}<p'_{\omega},\text{ or }\omega\in\bar{\Psi}_{f\rightarrow }\text{ and }\bar{p}_{\omega}>p'_{\omega},\\p'_{\omega},\quad&\text{else.}
\end{cases}$$
By the definition of $u^f_{(\bar{\Psi},\bar{p})}$, for each $\Psi\in \tilde{D}^f(p)$ there is a $\Xi\subseteq\bar{\Psi}\setminus\Psi$ with $\Psi\cup\Xi\in D^f(q)$, and for each $\Psi'\in \tilde{D}^f(p')$ there is a $\Xi'\subseteq\bar{\Psi}\setminus\Psi'$ with $\Psi'\cup\Xi'\in D^f(q')$.

 We show that the first part of the SSS condition, the first part of the CSC condition and the LAD holds for $\tilde{D}^f$. A dual argument shows that the second part of the SSS, the second part of the CSC condition and the LAS holds. Let $p\leq p'$ with $p'_{\omega}=p_{\omega}$ for $\omega\in\Omega_{f\rightarrow}$. It suffices to consider the case where $p,p'$ differ only on one trade $\omega'\in\Omega_{\rightarrow f}$, i.e.~$p'_{\omega'}>p_{\omega'}$ and $p'_{-\omega'}=p_{-\omega'}$. There exists by Proposition~\ref{SingleImprovement} applied to $q$ and $q'$  a $\tilde{\Psi}\in \tilde{D}^f(q)$ such that one of the following four cases is true:
%We consider two cases:
%\begin{enumerate}
%   \item
%   $\omega'\in\bar{\Psi}$ and $p_{\omega'}\geq\bar{p}_{\omega'}$,
%   \item $\omega'\notin\bar{\Psi}$ or $\omega'\in\bar{\Psi}$ and $p_{\omega'}< \bar{p}_{\omega'}$.
%   \end{enumerate} In the first case, we have ${q}'=q$, and thus  for each $\Psi'\in \tilde{D}^f(p')$, we have $\Psi'\in \tilde{D}^f(p).$ Thus in this case, the three conditions hold with $\Psi=\Psi'$.
%\begin{align*}|(\Psi'\cup\Xi')_{\rightarrow f}\setminus\tilde{\Psi}_{\rightarrow f }|+|\tilde{\Psi}_{f\rightarrow}\setminus(\Psi'\cup\Xi')_{f\rightarrow }|\leq 1\\|\tilde{\Psi}_{\rightarrow f}\setminus(\Psi'\cup\Xi')_{\rightarrow f }|+|(\Psi'\cup\Xi')_{f\rightarrow}\setminus\tilde{\Psi}_{f\rightarrow }|\leq 1\end{align*}

 %%\begin{align*}
 %\tilde{\Psi}_{\rightarrow f}\setminus\{\omega'\}\subseteq \Psi'_{\rightarrow f}\cup\Xi'_{\rightarrow f}\text{ and }\Psi'_{f\rightarrow }\cup\Xi'_{f\rightarrow}\subseteq  \tilde{\Psi}_{f\rightarrow}\\
 %|\tilde{\Psi}_{\rightarrow f}|-|\tilde{\Psi}_{f\rightarrow }|\geq |(\Psi'\cup\Xi')_{\rightarrow f}|-|(\Psi'\cup\Xi')_{f\rightarrow }|
 %\end{align*}
 \begin{enumerate}
    \item $\tilde{\Psi}=\Psi'\cup\Xi'$,
    \item  $\tilde{\Psi}=\Psi'\cup\Xi'\cup\{\omega'\}$,
    \item
     $\tilde{\Psi}=\Psi'\cup\Xi'\cup\{\omega,\omega'\}$ for a $\omega\in\Omega_{f\rightarrow}\setminus(\Psi'\cup\Xi')$,
    \item $\tilde{\Psi}=\Psi'\cup\Xi'\cup\{\omega'\}\setminus\{\omega\}$ for a $\omega\in(\Psi'\cup\Xi')_{\rightarrow f}.$

 \end{enumerate}
    In the first case, we let
    $$\Psi=\begin{cases}
    \Psi'\cup\{\omega'\},\quad&\text{ if }\omega'\in \Xi',\\\Psi',\quad&\text{ else.}
    \end{cases}\quad
    \Xi=\begin{cases}
    \Xi'\setminus\{\omega'\},\quad&\text{ if }\omega'\in \Xi',\\\Xi'\quad&\text{ else.}
    \end{cases}$$
        In the second case, we let
    $$\Psi=
    \Psi'\cup\{\omega'\},\quad
    \Xi=
    \Xi'.$$
        In the third case, we let
    $$\Psi=\begin{cases}
\Psi'\cup\{\omega,\omega'\},\quad&\text{ if }p_{\omega}\leq \bar{p}_{\omega},\\\Psi'\cup\{\omega'\},\quad&\text{ if }p_{\omega}> \bar{p}_{\omega}.
\end{cases}\quad
\Xi=\begin{cases}
\Xi',\quad&\text{ if }p_{\omega}\leq \bar{p}_{\omega},\\\Xi'\cup\{\omega\},\quad&\text{ if }p_{\omega}> \bar{p}_{\omega}.\end{cases}$$
    In the fourth case, we let
$$\Psi=\begin{cases}
\Psi'\cup\{\omega'\},\quad&\text{ if }{\omega}\in\Xi',\\
\Psi'\cup\{\omega'\}\setminus\{\omega\},\quad&\text{ if }\omega\in\Xi'.
\end{cases}\quad \Xi=\begin{cases}
\Xi'\setminus\{\omega\},\quad&\text{ if }{\omega}\in\Xi',\\
\Xi',\quad&\text{ if }\omega\in\Xi'.
\end{cases}$$
By construction we have $\tilde{\Psi}=\Psi\cup\Xi$, $\Psi\subseteq \{\omega:p_{\omega}\leq\bar{p}_{\omega}\}$ and $\Xi\subseteq \{\omega:p_{\omega}\geq\bar{p}_{\omega}\}$. As $\tilde{\Psi}\in D^f(q)$ we thus have $\Psi\in \tilde{D}^f(p)$. Moreover, by construction, we have $$\Psi_{\rightarrow f}\setminus\{\omega'\}\subseteq \Psi'_{\rightarrow f},\quad \Psi'_{f\rightarrow}\subseteq \Psi_{f\rightarrow},$$
and $$|\Psi_{\rightarrow f}|-|\Psi_{f\rightarrow }|\geq|\Psi'_{\rightarrow f}|-|\Psi'_{f\rightarrow }|.$$
\end{proof}
With the lemma we can prove the result.
\begin{proof}
    Following an idea of~\cite{Ravi}, we can construct a corresponding profile $\tilde{u}$ satisfying BWP such that equilibria under $\tilde{u}$ are equilibria under $u$: By BCV, there exists for each firm $f\in F$ a $K^f\geq0$ such that  $$\inf_{\{[\Psi,p]\in2^{\Omega}\times\mathbb{R}^{\Omega}:u^f(\Psi,p)>u^f(\emptyset)\}}\left(\sum_{\omega\in\Omega_{f\rightarrow}}p_{\omega}-\sum_{\omega\in\Omega_{\rightarrow f}}p_{\omega}\right)>-K^f.$$
    Let $K:=\sum_{f\in F}K^f$.
     For terminal sellers or buyers it is easy to see that BCV implies BWP, and for a firm $f\in F$ that is a terminal buyer or seller, we choose $\tilde{u}^f=u^f$. For a firm $f\in F$ that is neither terminal buyer nor terminal seller we let $$\tilde{u}^f(\Psi,p):=\max_{\Xi\subseteq\Omega_f\setminus\Psi}u^f(\Psi\cup\Xi,p|_{\Psi},(K)_{\omega\in \Xi_{\rightarrow f}},(-K)_{\omega\in \Xi_{f\rightarrow }}).$$

    By Lemma~\ref{Alex} (applied to the allocation $(\Omega_f,(K)_{\omega\in \Omega_{\rightarrow f}},(-K)_{\omega\in \Omega_{f\rightarrow}})$), the utility function $\tilde{u}^f$ satisfies FS, LAD and LAS. Moreover, by construction, $\tilde{u}^f$ satisfies BWP  where $K$ can be chosen as above. With a completely analogous argument to the proof of Lemma~B.2 in~\cite{Ravi}, for each competitive equilibrium $[\Psi,p]$ under $\tilde{u}$ we have $u^f(\Psi,p)=\tilde{u}^f(\Psi,p)$ for each $f\in F$, and $[\Psi,p]$ is also a competitive equilibrium under $u$.

     By Proposition~\ref{BWP}, there exists a terminal seller optimal equilibrium $\tilde{p}\in\mathcal{E}(\tilde{u})\subseteq\mathcal{E}(u)$ for $\tilde{u}$. We show that $\tilde{p}$ is also terminal seller optimal for $u$. Suppose not. Then there is a terminal seller $f$ and an equilibrium $\bar{p}\in\mathcal{E}(u)\setminus\mathcal{E}(\hat{u})$ with $v^f(\bar{p})>v^f(\tilde{p}).$
    Let $\bar{\Psi}\in\mathcal{E}(u,\bar{p})$. Consider the utility function $u^f_{(\bar{\Psi},\bar{p})}$ as in Lemma~\ref{Alex}. First observe that $[\bar{\Psi},\bar{p}]$ is also an equilibrium under $(u^f_{(\bar{\Psi},\bar{p})},u^{-f})$ as $$\max_{\Psi\subseteq \Omega_f}\max_{\Xi\subseteq \bar{\Psi}\setminus\Psi}u^f(\Psi\cup\Xi,\bar{p}|_{\Psi},\bar{p}|_{\Xi})=\max_{\Psi\subseteq \Omega_f} u^f(\Psi,\bar{p})=u^f(\bar{\Psi},\bar{p})=u^f_{(\bar{\Psi},\bar{p})}(\bar{\Psi},\bar{p}).$$

    Next observe that $u^f_{(\bar{\Psi},\bar{p})}$ satisfies BWP, and analogously to the argument for $u$ and $\tilde{u}$, each equilibrium $[\Psi,p]$ under $(u^f_{(\bar{\Psi},\bar{p})},\tilde{u}^{-f})$, is an equilibrium under $(u^f_{(\bar{\Psi},\bar{p})},{u}^{-f})$ that yields the same utility for each firm under both profiles. Let $[\hat{\Psi},\hat{p}]$ be an equilibrium under $(u^f_{(\bar{\Psi},\bar{p})},\tilde{u}^{-f}).$
 As $[\bar{\Psi},\bar{p}]$ and $[\hat{\Psi},\hat{p}]$ are equilibria under $(u^f_{(\bar{\Psi},\bar{p})},{u}^{-f})$, the second part of Theorem~\ref{Util} implies that $[\bar{\Psi},\hat{p}]$ is an equilibrium under $(u^f_{(\bar{\Psi},\bar{p})},u^{-f})$.
By construction,
 $$u^f_{(\bar{\Psi},\bar{p})}(\Psi,p)\geq u^f(\Psi,p),\text{ for each } (\Psi,p)\in\mathcal{A}_f,$$
 with equality if $\Psi=\bar{\Psi}_f$. Thus, $\bar{\Psi}_f\in D^f_{(\bar{\Psi},\bar{p})}(\hat{p})$ implies $\bar{\Psi}_f\in D^f(\hat{p})$.
By construction, for $f'\neq f$,
$$\tilde{u}^{f'}(\Psi,p)\geq u^{f'}(\Psi,p),\text{ for each }(\Psi,p)\in\mathcal{A}_{f'},$$ where, as observed above, the inequality holds with equality for equilibrium allocations, so in particular if $\Psi=\hat{\Psi}_{f'}$ and $p=\hat{p}$.
Thus, $$\tilde{u}^{f'}(\bar{\Psi},\hat{p})\geq{u}^{f'}(\bar{\Psi},\hat{p})= u^{f'}(\hat{\Psi},\hat{p})= \tilde{u}^{f'}(\hat{\Psi},\hat{p})=\max_{\Psi\subseteq\Omega_{f'}}\tilde{u}^{f'}(\Psi,\hat{p}),$$ where the first equality follows as $\bar{\Psi}_{f'},\hat{\Psi}_{f'}\in D^{f'}(\hat{p})$. Therefore $\bar{\Psi}_{f'}\in \tilde{D}^{f'}(\hat{p})$. Hence, $[\bar{\Psi},\hat{p}]$ is an equilibrium under $\tilde{u}.$ But, as $[\bar{\Psi},\hat{p}]$ is an equilibrium under $(u^f_{(\bar{\Psi},\bar{p})},{u}^{-f})$ and by construction of $u^f_{(\bar{\Psi},\bar{p})}$ we have
$$v^f(\hat{p})=u^{f}(\bar{\Psi},\hat{p})=u^f_{(\bar{\Psi},\bar{p})}(\bar{\Psi},\hat{p})\geq u^f_{(\bar{\Psi},\bar{p})}(\emptyset)=u^f(\bar{\Psi},\bar{p})=v^f(\bar{p})>v^f(\tilde{p}).$$ But this contradicts the terminal seller optimality of $\tilde{p}$ among equilibria under $\tilde{u}$.

%   Let $$T(u):=\{(t^f)_{f\in F}\in\mathbb{R}^{F}:\exists p\in\mathbb{R}^{\Omega},\Psi\in \mathcal{E}(u,p), \sum_{\omega\in\Psi_{f\rightarrow}}p_{\omega}-\sum_{\omega\in\Psi_{\rightarrow f}}p_{\omega}=t^f\text{ for }f\in F\}$$ be the set of vectors of total transfers that each agent obtains in some equilibrium.  We show that $T(u)$ is compact. First note that by BCV there is a $K>0$ such that for each equilibrium allocation $(\Psi,p)$ and each $f\in F$ we have $\sum_{\omega\in\Psi_{f\rightarrow}}p_{\omega}-\sum_{\omega\in\Psi_{\rightarrow f}}p_{\omega}>-K$. Now note that by definition of an allocation we have $\sum_{f\in F}(\sum_{\omega\in\Psi_{f\rightarrow}}p_{\omega}-\sum_{\omega\in\Psi_{\rightarrow f}}p_{\omega})=0.$ These two observations imply that for each $f\in F$ we have $\sum_{\omega\in\Psi_{f\rightarrow}}p_{\omega}-\sum_{\omega\in\Psi_{\rightarrow f}}p_{\omega}<K|F|.$ Thus $T(u)$ is bounded.

    %Next, by BWP there is a $K>0$ such that for each $f\in F$ and each $p\in\mathbb{R}^{\Omega}$,  if $\Psi\in D^f(p)$ then $p_{\omega}<K$ for $\omega\in\Omega_{\rightarrow f}$ and $p_{\omega}>-K$ for $\omega\in\Omega_{f\rightarrow}$. Thus for each equilibrium $[\Psi,p]$ there is a $p'\in [-K,K]^{\Omega}$ with $\Psi\in\mathcal{E}(u,p')$ and $p'_{\omega}=p_{\omega}$ for each $\omega\in \Psi$. As $\mathcal{E}(u)$ is closed this implies that $\mathcal{E}'(u):=\mathcal{E}(u)\cap [-K,K]^{\Omega}$ is compact.
    %
\end{proof}
\subsection*{Proof of Theorem~\ref{GSP}}
\begin{proof}
Let $F'\subseteq F$ be the set of terminal buyers. Let $\mathcal{U}=\bigtimes_{f\in F}\mathcal{U}_f$ where for $f\in F'$ the set $\mathcal{U}_f$ is the set of unit demand and BCV utility functions and for each $f\in F\setminus F'$ the set $\mathcal{U}_f$ is the set of BCV, FS, LAD and LAD utility functions. In the following for $\tilde{u}^f,\hat{u}^f\in\mathcal{U}_f$ etc.~we denote the induced demand by $\tilde{D}^f,\hat{D}^f$ etc.

 Let $\mathcal{M}:\mathcal{U}\rightarrow\mathcal{A}$ be a buyer-optimal mechanism. First we establish that $\mathcal{M}$ is immune to truncation strategies.
\begin{claim}\label{trunc}
Let  $f\in F'$.
Let $u,\tilde{u}\in\mathcal{U}$ with $\tilde{u}^{-f}=u^{-f}$ and let $[\Psi,p]$ be a buyer-optimal equilibrium under $u$. If $\Psi_{f}\neq\emptyset$, $\tilde{u}^f(\omega,\cdot)=u^f(\omega,\cdot)$ for each $\omega\in\Omega_{\rightarrow f}$ and $\tilde{u}^f(\emptyset)>\tilde{u}^f(\Psi,p)$, then for each equilibrium $[\tilde{\Psi},\tilde{p}]$ under $\tilde{u}$, we have $\tilde{\Psi}_f=\emptyset$.
\end{claim}
\begin{proof}
Suppose not. Then $\tilde{\Psi}_{f}\neq\emptyset$. Let $\tilde{\Psi}_f=\{\tilde{\omega}\}$. Note that also $\{\tilde{\omega}\}\in D^f(\tilde{p})$. Thus $[\tilde{\Psi},\tilde{p}]$ is an equilibrium under $u$. But since $$u^f(\tilde{\omega},\tilde{p}_{\tilde{\omega}})=\tilde{u}^f(\tilde{\omega},\tilde{p}_{\tilde{\omega}})\geq \tilde{u}^f(\emptyset)>\tilde{u}^f(\Psi,p)={u}^f(\Psi,p)$$ this contradicts the buyer optimality of $[\Psi,p].$
\end{proof}
Second we establish that $\mathcal{M}$ is immune to certain strategies where a single terminal buyer changes the utility function for one trade.
\begin{claim}\label{cl}
Let $f\in F'$. Let $u,\hat{u}\in\mathcal{U}$ with $\hat{u}^{-f}=u^{-f}$ such that there is a $\hat{\omega}\in \Omega_{\rightarrow f}$ with $\hat{u}^f(\omega,\cdot)=u^f(\omega,\cdot)$ for $\omega\neq\hat{\omega}$ and $\hat{u}^f(\emptyset)=u^f(\emptyset)$. Let $[\bar{\Psi},\bar{p}]$ be a buyer-optimal equilibrium under $u$. If for all $p_{\hat{\omega}}\in\mathbb{R}$,
we have
\begin{align*}
u^f(\hat{\omega},p_{\hat{\omega}})\leq u^{f}(\bar{\Psi},\bar{p})\Rightarrow \hat{u}^f(\hat{\omega},p_{\hat{\omega}})= u^f(\hat{\omega},p_{\hat{\omega}}),\\
u^f(\hat{\omega},p_{\hat{\omega}})\geq u^{f}(\bar{\Psi},\bar{p})\Rightarrow \hat{u}^f(\hat{\omega},p_{\hat{\omega}})\geq u^f(\hat{\omega},p_{\hat{\omega}}),
\end{align*}
 then  $[\bar{\Psi},\bar{p}]$ is a buyer-optimal equilibrium under $\hat{u}$.

\end{claim}
\begin{proof}
 %We extend $(\bar{\Psi},\bar{p})$ and $(\hat{\Psi},\hat{p})$ to $[\bar{\Psi},\bar{p}]$ and $[\hat{\Psi},\hat{p}]$ as described above.
Let $[\hat{\Psi},\hat{p}]$ be a buyer-optimal equilibrium under $\hat{u}$.  If
$u^f(\hat{\omega},\hat{p}_{\hat{\omega}})\leq u^{f}(\bar{\Psi},\bar{p})$, then we have $D^f(\hat{p})=\hat{D}^f(\hat{p})$ and $[\hat{\Psi},\hat{p}]$ is an equilibrium under $u$. In this case, by buyer-optimality of $[\bar{\Psi},\bar{p}]$ under $u$, we have $\hat{u}^{f'}(\hat{\Psi},\hat{p})= u^{f'}(\hat{\Psi},\hat{p})\leq u^{f'}(\bar{\Psi},\bar{p})=u^{f'}(\bar{\Psi},\bar{p})$ for each $f'\in F'$. Moreover,  $\hat{u}^f(\hat{\omega},\bar{p}_{\hat{\omega}})=\hat{u}^f(\hat{\omega},\bar{p}_{\hat{\omega}})$, and therefore $[\bar{\Psi},\bar{p}]$ is an equilibrium under $\hat{u}$. Thus, in this case $[\bar{\Psi},\bar{p}]$ is a buyer-optimal equilibrium under $\hat{u}$.  It remains to consider the case that $u^f(\hat{\omega},\hat{p}_{\hat{\omega}})> u^{f}(\bar{\Psi},\bar{p})$. In this case, consider the two sub-cases that $\hat{\Psi}_f=\{\hat{\omega}\}$ or $\hat{\Psi}_f\neq\{\hat{\omega}\}$.

 If $\hat{\Psi}_f\neq\{\hat{\omega}\}$, we can show that  $[\hat{\Psi},\hat{p}]$ is an equilibrium under $u$. Suppose not. Then, as $\hat{\Psi}_f\notin D^f(\hat{p})$ and $u^f(\omega,\hat{p}_{\omega})=\hat{u}^f(\omega,\hat{p}_{\omega})$ for $\omega\neq \hat{\omega}$, we have $u^f(\hat{\omega},\hat{p}_{\hat{\omega}})> u^f(\hat{\Psi},\hat{p})$. Thus $\hat{u}^f(\hat{\omega},\hat{p}_{\hat{\omega}})\geq u^f(\hat{\omega},\hat{p}_{\hat{\omega}})> u^f(\hat{\Psi},\hat{p})=\hat{u}^f(\hat{\Psi},\hat{p})$ and therefore $\hat{\Psi}_f\notin \hat{D}^f(\hat{p})$. This contradicts the assumption that $[\hat{\Psi},\hat{p}]$ is an equilibrium under $\hat{u}$. Thus, $[\hat{\Psi},\hat{p}]$ is an equilibrium under $u$ and by the same reasoning as above, $[\bar{\Psi},\bar{p}]$ is a buyer-optimal equilibrium under $\hat{u}.$

If $\hat{\Psi}_f=\{\hat{\omega}\}$, consider the utility function $\tilde{u}^f$ obtained from ${u}^f$ by truncating as follows: $\tilde{u}^f(\omega,\cdot)={u}^f(\omega,\cdot)$ for all $\omega\in\Omega_{\rightarrow f}$ and ${u}^f(\bar{\Psi},\bar{p})< \tilde{u}^f(\emptyset)<{u}^f(\hat{\omega},\hat{p}_{\hat{\omega}})$. By Claim~\ref{trunc}, for each equilibrium $[\Psi,p]$ under $\tilde{u}:=(\tilde{u}^f,u^{-f})$ we have $\Psi_f=\emptyset$. Define the utility function $\tilde{u}_*^f$ by $\tilde{u}_*^f(\hat{\omega},\cdot)=\tilde{u}^f(\hat{\omega},\cdot)={u}^f(\hat{\omega},\cdot)$, by  $\tilde{u}^f_*(\omega,\cdot)=-\infty$ for each $\omega\neq\hat{\omega}$, and $\tilde{u}^f_*(\emptyset)=\tilde{u}^f(\emptyset)$. As for each equilibrium $[\Psi,p]$ under $\tilde{u}$ we have $\Psi_f=\emptyset$, we have $\mathcal{E}(\tilde{u})\subseteq\mathcal{E}(\tilde{u}_*)$ for $\tilde{u}_*:=(\tilde{u}_*^f,u^{-f})$, and in particular, there is an equilibrium $[\tilde{\Psi},\tilde{p}]$ under $\tilde{u}_*$ with $\tilde{\Psi}_f=\emptyset$. Observe however that $\tilde{u}_*^f(\hat{\omega},\hat{p}_{\hat{\omega}})=\tilde{u}^f(\hat{\omega},\hat{p}_{\hat{\omega}})={u}(\hat{\omega},\hat{p}_{\hat{\omega}})>\tilde{u}_*^f(\emptyset)$. Thus $\tilde{D}^f_*(\hat{p})=\{\{\hat{\omega}\}\}$ and $[\hat{\Psi},\hat{p}]$ is an equilibrium under $\tilde{u}_*$ with $\tilde{u}_*^f(\hat{\Psi},\hat{p})>\tilde{u}_*^f(\emptyset)$. This contradicts the rural hospitals theorem (the second part of Theorem~\ref{Util}).
\end{proof}
With the claim, we can prove the result. Suppose there are profiles $u,\tilde{u}\in\mathcal{U}$ such that $\tilde{u}^{-F'}=u^{-F'}$ and for each $f\in F'$, we have
$u^f(\mathcal{M}(\tilde{u}))> u^f(\mathcal{M}(u)).$
Let $\mathcal{M}(u)=(\bar{\Psi},\bar{p})$ and $\mathcal{M}(\tilde{u})=(\tilde{\Psi},\tilde{p})$.
%By BWP for $u$ and for $\tilde{u}$, there exists a $K>0$ such that for all $p\in\mathbb{R}^{\Omega}$ and all $\omega\in \Omega$, if $\omega\in\Psi\in D^{b(\omega)}(p)$ then $p_{\omega}<K$, if $\omega\in\Psi\in D^{s(\omega)}(p)$ then $p_{\omega}>-K$, if $\omega\in\Psi\in \tilde{D}^{b(\omega)}(p)$ then $p_{\omega}<K$, and if $\omega\in\Psi\in \tilde{D}^{s(\omega)}(p)$, then $p_{\omega}>-K$. Thus we can specify prices for non realized trades at $(\Psi,p)$ to obtain $p\in\mathcal{E}(u)$ with $p\in[-K,K]^{\Omega}$ and similarly we can specify prices for non realized trades at $(\tilde{\Psi},\tilde{p})$ to obtain $\tilde{p}\in\mathcal{E}(\tilde{u})$ with $\tilde{p}\in[-K,K]^{\Omega}$.
We define for each $f\in F'$, a $\hat{u}^f\in\mathcal{U}_f$ as follows: Note that $\tilde{\Psi}_f\neq\emptyset$ as $u^f(\tilde{\Psi},\tilde{p})> u^f(\bar{\Psi},\bar{p})\geq u^f(\emptyset)$. Let $\tilde{\omega}\in\tilde{\Psi}$ be the unique trade in $\tilde{\Psi}$ such that $b(\tilde{\omega})=f$. We let $\hat{u}^f({\omega},\cdot)={u}^f({\omega},\cdot)$ for $\omega\neq \tilde{\omega}$ and we let $\hat{u}^f(\emptyset)=u^f(\emptyset)$. To construct $\hat{u}^f(\tilde{\omega},\cdot)$ we proceed as follows:  Define $\hat{u}^f(\tilde{\omega},p_{\tilde{\omega}}):={u}^f(\tilde{\omega},p_{\tilde{\omega}})$ for each ${p}_{\tilde{\omega}}\in\mathbb{R}$ with $u^f(\tilde{\omega},p_{\tilde{\omega}}) \leq u^f(\bar{\Psi},\bar{p})$. Define  $$\hat{u}^f(\tilde{\omega},\tilde{p}_{\tilde{\omega}}):=\max_{\omega\in\Omega_{\rightarrow f}}u^f(\omega,\tilde{p}_{\omega}).$$ Note that $$\hat{u}^f(\tilde{\omega},\tilde{p}_{\tilde{\omega}})\geq u^f(\tilde{\omega},\tilde{p}_{\tilde{\omega}})>u^f(\bar{\Psi},\bar{p})=\hat{u}^f(\bar{\Psi},\bar{p}).$$
For prices $p_{\tilde{\omega}}\neq \tilde{p}_{\tilde{\omega}}$ with $u^f(\tilde{\omega},p_{\tilde{\omega}})\geq u^f(\bar{\Psi},\bar{p})$, we can choose any continuous and monotonic extension such that $\hat{u}^f(\tilde{\omega},p_{\tilde{\omega}})\geq u^f(\tilde{\omega},p_{\tilde{\omega}})$.  By Claim~\ref{cl}, $[\bar{\Psi},\bar{p}]$ is a buyer-optimal equilibrium for $(\hat{u}^{f},u^{-f})$. Iterating for all $f\in F'$, $[\bar{\Psi},\bar{p}]$ is a buyer-optimal equilibrium under $\hat{u}:=(\hat{u}^{F'},u^{-F})$. Note however that by construction of $\hat{u}$, for each $f\in F'$ we have $\tilde{\Psi}_f\in \hat{D}^f(\tilde{p})$. Thus $[\tilde{\Psi},\tilde{p}]$ is an equilibrium under $\hat{u}$ with $\hat{u}^f(\tilde{\Psi},\tilde{p})>\hat{u}^f(\bar{\Psi},\bar{p})$ for each $f\in F'$. This contradicts the buyer-optimality of $[\bar{\Psi},\bar{p}]$ under $(\hat{u}^{F'},u^{-F})$.
\end{proof}
\section{Proofs for Section~\ref{Exchange}}
\subsection*{Proof of Lemma~\ref{equivalencei}}
\begin{proof}
    First we show the result for non-negative prices. Let $p,p'\in\mathbb{R}_+^{\Omega_f}$. Define $q,q'\in\mathbb{R}_{+}^X$ by
    \begin{align*}
    q_{x}:=\begin{cases}
    \min_{\omega\in\Omega_{\rightarrow f}:x(\omega)=x}p_{\omega},\text{ for }x\notin X_f,\\
    \max_{\omega\in \Omega_{f\rightarrow}:x(\omega)=x}p_{\omega},\text{ for }x\in X_f,
    \end{cases}q'_{x}:=\begin{cases}
    \min_{\omega\in\Omega_{\rightarrow f}:x(\omega)=x}p'_{\omega},\text{ for }x\notin X_f,\\
    \max_{\omega\in \Omega_{f\rightarrow}:x(\omega)=x}p'_{\omega},\text{ for }x\in X_f.
    \end{cases}
    \end{align*}
    By construction we have \begin{align*}D^f(p)&=\{\Psi\subseteq \Omega_f:X_f(\Psi)\in \tilde{D}^f(q),p_{\omega}=q_{x(\omega)}\text{ for }\omega\in\Psi\},\\
    D^f(p')&=\{\Psi'\subseteq \Omega_f:X_f(\Psi')\in \tilde{D}^f(q'),p'_{\omega}=q'_{x(\omega)}\text{ for }\omega\in\Psi'\}.
    \end{align*}

    If $p_{\omega}= p'_{\omega}$ for $\omega\in\Omega_{f\rightarrow}$ and $p_{\omega}\leq p'_{\omega}$ for $\omega\in\Omega_{\rightarrow f}$, then for $\Psi'\in D^f(p')$ there is, by gross substitutability a $Y\in \tilde{D}^f(q)$ with $\{x\in Y:q'_{x}= q_x\}\subseteq X_f(\Psi').$ Thus, if $x\in Y\setminus X_f$ and $q'_x=q_x$, then $x\in X_f(\Psi')$, and if $x\in X_f\setminus X_f(\Psi')$, then, as $q_x'=q_x$,  we have $x\in X_f\setminus Y$. Therefore there is a $\Psi\in D^f(p)$ with
    \begin{align*}
    \{\omega\in\Psi_{\rightarrow f}:p'_{\omega}=q'_{x(\omega)}=q_{x(\omega)}=p_{\omega}\}\subseteq \Psi'_{\rightarrow f},\quad
    \Psi'_{f\rightarrow}\subseteq\Psi_{f\rightarrow}.
    \end{align*}
    Similarly, by the law of aggregate demand, there is a $Y\in\tilde{D}^f(q)$ such that $|Y|\geq|X_f(\Psi')|.$ Then there is a $\Psi\in D^f(p)$ with $Y=X_f(\Psi)$. But then
    \begin{align*}&|\Psi_{\rightarrow f}|-|\Psi_{f\rightarrow}|=|Y\setminus X_f|-|X_f\setminus Y|=|Y|-|X_f|\\\geq&|X_f(\Psi')|-|X_f|=|X_f(\Psi')\setminus X_f|-|X_f\setminus X_f(\Psi')|=|\Psi'_{\rightarrow f}|-|\Psi_{f\rightarrow}|.\end{align*}

    An analogous argument shows that $D^f$ satisfies the second part of the SSS condition, the second part of the CSC condition, and LAS.
    %at for $p_{\omega}= p'_{\omega}$ for $\omega\in\Omega_{\rightarrow f}$ and $p_{\omega}\geq p'_{\omega}$ for $\omega\in\Omega_{f\rightarrow }$, for each $\Psi\in D^f(p')$ there is a $\Psi\in D^f(p)$ such that
%   \begin{align*}
%   \{\omega\in\Psi_{f \rightarrow}:p'_{\omega}=q'_{x(\omega)}=q_{x(\omega)}=p_{\omega}\}\subseteq \Psi'_{f\rightarrow},\quad
%   \Psi'_{\rightarrow f}\subseteq\Psi_{\rightarrow f},
%   \end{align*}
%   and $$|\Psi_{f\rightarrow}|-|\Psi_{\rightarrow f}|\geq|\Psi'_{f\rightarrow}|-|\Psi'_{\rightarrow f}|.$$

Next we establish FS and LAD/LAS on $\mathbb{R}^{\Omega_f}$. Let $p,p'\in\mathbb{R}^{\Omega_f}$ and define $q,q'\in\mathbb{R}^{X}$ as previously. Moreover, define $p^{0}:=(\max\{p_{\omega},0\}_{\omega\in\Omega})\in\mathbb{R}^{\Omega_f}$ and $(p')^{0}:=(\max\{p'_{\omega},0\}_{\omega\in\Omega})\in\mathbb{R}^{\Omega_f}$. By construction of $u^f$ and the assumption that $\tilde{u}^f(Y,t)\leq\tilde{u}^f(Y',t)$ for $Y\subseteq Y'$, we have
\begin{align*}
&D^f(p)=\{\Psi\in D^f(p^0):\{x\in X:q_x<0\}\subseteq X_f(\Psi),p_{\omega}=q_{x(\omega)}\text{ for }\omega\in\Psi\},\\
&D^f(p')=\{\Psi'\in D^f((p')^0):\{x\in X:q'_x<0\}\subseteq X_f(\Psi'),p'_{\omega}=q'_{x(\omega)}\text{ for }\omega\in\Psi'\}.
%&D^f(p)=\{\Psi\subseteq\Omega_f:\exists\tilde{\Psi}\in D^f(p^0), X_f(\Psi)=\{x\in X_f(\tilde{\Psi}):q_x\geq0\}\cup\{x\in X\setminus X_f:q_x<0\}\},\\
%&D^f(p')=\{\Psi\subseteq\Omega_f:\exists\tilde{\Psi}\in D^f((p')^0), X_f(\Psi)=\{x\in X_f(\tilde{\Psi}):q'_x\geq0\}\cup\{x\in X\setminus X_f:q'_x<0\}\}.
\end{align*}
In particular, for $\Psi'\in D^f(p')$ we have $\Psi'\in D^f((p')^0)$ and by FS for non-negative prices, there is a $\tilde{\Psi}\in D^f(p^0)$ with  \begin{align*}
\{\omega\in\tilde{\Psi}_{\rightarrow f}:(p')^0_{\omega}=p_{\omega}^0\}\subseteq \Psi'_{\rightarrow f},\quad
\Psi'_{f\rightarrow}\subseteq\tilde{\Psi}_{f\rightarrow}.
\end{align*}
We can find a $\Psi\in D^f(p)$, such that $\{\omega\in\tilde{\Psi}:q_{x(\omega)}=q'_{x(\omega)}\}\subseteq \Psi.$ Now let $\omega\in \Psi_{\rightarrow f}$ and $p'_{\omega}=p_{\omega}$. Then $q'_{x(\omega)}=p'_{\omega}=p_{\omega}=q_{x(\omega)}$ and $\omega\in \tilde{\Psi}$. Moreover, $(p')_{\omega}^0=p^0_{\omega}$ and therefore $\omega\in\Psi'_{\rightarrow f}$. Similarly, for all $\omega\in\Omega_{f\rightarrow}$ we have $p'_{\omega}=p_{\omega}$. If $p_{\omega}=p'_{\omega}<q'_{x(\omega)}=q_{x(\omega)}$ then $\omega\notin \Psi$ and $\omega\notin \Psi'$. If $p_{\omega}=p'_{\omega}=q'_{x(\omega)}=q_{x(\omega)}$,  then $\omega\notin \Psi_{\rightarrow f}$ implies $\omega\notin \tilde{\Psi}_{\rightarrow f}$. Moreover, $(p')_{\omega}^0=p^0_{\omega}$ and therefore $\omega\notin\Psi'_{\rightarrow f}$.

To establish LAD, let $\Psi'\in D^f(p')$. Since $\Psi'\in D^f((p')^0)$ and by LAD for non-negative price vectors, there is a $\tilde{\Psi}\in D^f(p^0)$ and hence a $\Psi\in D^f(p)$ with $X_f(\Psi)=X_f(\tilde{\Psi})\cup\{x\in X:q_x<0\}$ such that $$|{\Psi}_{\rightarrow f}|-|{\Psi}_{f\rightarrow}|=|X_f(\Psi)\setminus X_f|-|X_f\setminus X_f(\Psi)|\geq|\tilde{\Psi}_{\rightarrow f}|-|\tilde{\Psi}_{f\rightarrow}|\geq|\Psi'_{\rightarrow f}|-|\Psi'_{f\rightarrow}|.$$

    An analogous argument shows that $D^f$ satisfies the second part of the SSS condition, the second part of the CSC condition, and LAS.
\end{proof}
\subsection*{Proof of Corollary~\ref{general}}
\begin{proof}
    For the first part, consider price vectors in the induces trading network $q,q'\in\mathbb{R}_+^{\Omega}$ defined by $q_{\omega}:=p_{x(\omega)}$ and $q'_{\omega}:=p'_{x(\omega)}$ for each $\omega\in\Omega$. By Proposition~\ref{equivalence}, $q$ and $q'$ are equilibrium prices in the induced trading network. By Lemma~\ref{equivalencei}, utility functions in the induced trading network satisfy FS, LAD and LAS. Thus, by Theorem~\ref{Util}, price vectors $\bar{q},\underline{q}\in\mathbb{R}_+^{\Omega}$ with $$\bar{q}_{\omega}=\max\{q_{\omega},q'_{\omega}\},\quad\underline{q}_{\omega}=\min\{q_{\omega},q'_{\omega}\},$$
    are equilibrium prices in the trading network. By construction of $q$ and $q'$, for each $\omega,\omega'\in{\Omega}$ with $x(\omega)=x(\omega')$ we have $q_{\omega}=p_{x(\omega)}=q_{\omega'}$ and $q'_{\omega}=p'_{x(\omega)}=q'_{\omega'}$. Therefore, $$\bar{p}_x=\max_{\omega\in\Omega,x=x(\omega)}\bar{q}_{\omega}\quad\text{and}\quad\underline{p}_x=\max_{\omega\in\Omega,x=x(\omega)}\underline{q}_{\omega},$$
    and, by Proposition~\ref{equivalence}, $\bar{p}$ and $\underline{p}$ are equilibrium price vectors.

    For the second part, define $q$ and $q'$ as before and let $$\Psi:=\{\omega\in \Omega:x(\omega)\in Y_{b(\omega)}\cap X_{s(\omega)}\}.$$ As shown in the proof of Proposition~\ref{equivalence}, $[\Psi,q]$ is an equilibrium of the trading network. By the second part of Theorem~\ref{Util}, there is a $\Psi'\subseteq\Omega$ such that $[\Psi',q']$ is an equilibrium of the trading network with $$|\Psi_{\rightarrow f}|-|\Psi_{f\rightarrow}|=|\Psi'_{\rightarrow f}|-|\Psi'_{f\rightarrow}|.$$

    Let $Y'=(Y'_f)_{f\in F}$ with $Y'_f:=X_f(\Psi')$. As shown in the proof of Proposition~\ref{equivalence}, $[Y',p']$ is an equilibrium of the exchange economy. Moreover,
    \begin{align*}
    |Y_f|=|Y_f\setminus X_f|-|X_f\setminus Y_f|+|X_f|=|\Psi_{\rightarrow f}|-|\Psi_{f\rightarrow}|+|X_f|\\=|\Psi'_{\rightarrow f}|-|\Psi'_{f\rightarrow}|+|X_f|=|Y'_f\setminus X_f|-|X_f\setminus Y'_f|+|X_f|=|Y'_f|.
    \end{align*}

    For the third part, we first show that the set of equilibrium price vectors in the induced trading network, $\mathcal{E}(u)$ is compact. The same argument as in the proof of Proposition~\ref{BWP} establishes that the surplus function $Z:\mathbb{R}_+^{\Omega}\rightarrow\mathbb{R}$ is continuous and hence $\mathcal{E}(u)\subseteq\mathbb{R}^{\Omega}_+$ is closed. To show that $\mathcal{E}(u)$ is bounded, note that by the full range assumption there exists a $K>0$ such that for each $f\in F$ and $Y\subseteq X$ we have $\tilde{u}^f(Y,-K)<\tilde{u}^f(X_f,0)$. For each equilibrium $[\Psi,p]$ in the trading network and each $f\in F$, we have $$u^f(\Psi,p)=\tilde{u}^f(X_f(\Psi),p_f(\Psi))\geq \tilde{u}^f(X_f,0)=u^f(\emptyset),$$ and therefore by monotonicity of utility in transfers $p_f(\Psi)>-K.$ Moreover, $\sum_{f\in F}p_f(\Psi)=0$. Thus, $p_f(\Psi)<|F|\cdot K$ for each $f\in F$. By the full range assumption, there is a $\tilde{K}>0$ such that for each $f\in F$ and  $Y\subseteq X$, we have $\tilde{u}^f(\emptyset,\tilde{K})>\tilde{u}^f(Y,|F|\cdot K).$ Note that for each equilibrium $[\Psi,p]$ of the trading network, each $f\in F$ and each $\Psi'\subseteq \Omega_{f}$ with $X_f(\Psi')=\emptyset$, we have $$\tilde{u}^f(\emptyset,\sum_{\omega\in \Psi'}p_{\omega})=u^f(\Psi',p)\leq u^f(\Psi,p)=\tilde{u}^f(X_f(\Psi),p_f(\Psi))<\tilde{u}^f(X_f(\Psi),|F|\cdot K)<\tilde{u}^f(\emptyset,\tilde{K}).$$
    Thus $\sum_{\omega\in \Psi'}p_{\omega}<\tilde{K}$ and, as $p_{\omega}\geq0$ for each $\omega\in\Omega$, we have $0\leq p_{\omega}<\tilde{K}$ for each $\omega\in \Psi'$. Now note that for each $\omega\in \Omega$, there exists a $\Psi'\subseteq\Omega_{s(\omega)}$ with $X(\Psi')=\emptyset$ and $\omega\in\Psi'$. Thus for each $\omega\in\Omega$ we have $0\leq p_{\omega}<\tilde{K}$. Thus $\mathcal{E}(u)$ is compact, by Propositions~\ref{equivalence} and~\ref{BWP} non-empty. Moreover, by Theorem~\ref{Util}, $\mathcal{E}(u)$ is a sublattice of $\mathbb{R}^{\Omega}.$ Since $\mathcal{E}(u)$ is a non-empty, compact sublattice of $\mathbb{R}^{\Omega}$, there exist $\bar{p},\underline{p}\in\mathcal{E}(u)$ such that for each $p\in\mathcal{E}(u)$ we have $\underline{p}_{\omega}\leq p_{\omega}\leq\bar{p}_{\omega}$ for each $\omega\in\Omega$. By the first part of Proposition~\ref{equivalence}, the vectors $\underline{q},\bar{q}\in\mathbb{R}_{+}^X$ defined by $$\underline{q}_{x}:=\max_{\omega\in\Omega,x=x(\omega)}\underline{p}_{\omega},\quad \bar{q}_x:=\max_{\omega\in\Omega,x=x(\omega)}\bar{p}_{\omega}
    $$ are equilibrium price vectors in the exchange economy. Now let $q\in\mathbb{R}_+^X$ be an equilibrium price vector in the exchange economy. By the second part of Proposition~\ref{equivalence}, the price vector $p\in\mathbb{R}_+^{\Omega}$ defined by $p_{\omega}:=p_{x(\omega)}$ for each $\omega\in\Omega$, is in $\mathcal{E}(u)$.
    Let $x\in X$. Let $\omega\in\Omega$ with $x=x(\omega)$ and $\underline{q}_x=\underline{p}_{\omega}$. Then $\underline{q}_x=\underline{p}_{\omega}\leq p_{\omega}=q_x.$ Similarly, let $\omega\in\Omega$ with $x=x(\omega)$ and $\bar{q}_x=\bar{p}_{\omega}$. Then $\bar{q}_x=\bar{p}_{\omega}\geq p_{\omega}=p_x.$ Thus $\bar{q},\underline{q}$ are the desired price vectors.
\end{proof}
\newpage
\section{Results for Trading Networks with General Preferences}\label{yolo}
\begin{table}[h]
    \centering
    \caption{Sufficient conditions for results for trading networks with general preferences.}
    {\footnotesize\begin{threeparttable}
            \begin{tabular}{|l|l|c|c|c|c|c|c|}\hline
                {Result (Theorem*)} & Source &{C\&M} & {FS} & {LADS} & {BCV} & {BWP}  & {NF}\\\hline\hline
                Existence of Equil. (1) &\cite{Ravi}& x& x& &x&  &\\               \hline
                                1st Welfare Theorem (2) &\cite{Ravi} &x & & & & &x \\\hline
                Rural Hospitals (3) & Theorem~\ref{Util}, part 2& x& x& x& & & \\\hline
                Lattice  (4) & Theorem~\ref{Util}, part 1& x& x& x&& & \\\hline

                Side Optimality (4) & Theorem~\ref{extreme}& x& x& x&x&  & \\\hline

                Equil.$\Rightarrow$ Stable  (5) & \cite{Ravi}& x& & &  &x& \\\hline
                Stable $\Rightarrow$ Equil. (6) & \cite{Ravi}& x& x& &x & & \\
                \hline
                Stable $\Leftrightarrow$ Group-Stable (8 \& 9) & \cite{Ravi}& x& x& &x & &x \\\hline
                Trail-Stable $\Leftrightarrow$ Equil. & \cite{Ravi}& x& x& &&x &  \\\hline
                Chain-Stable $\Leftrightarrow$ Stable & \cite{Hatfield:2018} & &x& x& & &  \\\hline
                Group-Strategy-Proofness   & Theorem~\ref{GSP}& x& x& x&x & & \\\hline

            \end{tabular}
            \medskip

            \begin{tablenotes}
                \small \textbf{Notation}:

        \textbf{Theorem*} Corresponding theorem in~\cite{Hatfield2013} under quasi-linear utility. The existence of a side-optimal equilibrium additionally assumes finite valuations.

                \textbf{C\&M} stands for \textsl{Continuity and Monotonicity},

                \textbf{FS} stands for \textsl{Full Substitutability},

                \textbf{LADS} stands for the \textsl{Laws of Aggregate Demand and Supply},

                \textbf{BCV} stands for \textsl{Bounded Compensating Variations},

                \textbf{BWP} stands for \textsl{Bounded Willingness to Pay},
                and

                \textbf{NF} stands for \textsl{No Frictions}.       \end{tablenotes}
    \end{threeparttable}}
    \label{table:HAproperties}

\end{table}
\end{document}